\documentclass[letter,11pt]{amsart}
\usepackage{geometry}                % See geometry.pdf to learn the layout options. There are lots.
\usepackage{mathpazo}  %Palatino fonts
\usepackage[english]{babel}
\usepackage[usenames,dvipsnames]{color}
\usepackage{url}
\usepackage{amsmath,amssymb,latexsym,mathrsfs,amssymb,amsthm}
\usepackage{enumerate}
\usepackage{mathtools} %Needed for \defeq and \eqdef
\usepackage{epstopdf}
\usepackage{graphicx}
\usepackage[all]{xy}
\usepackage{setspace}
\usepackage{texdraw}
\usepackage{calligra}

\usepackage[usenames,dvipsnames]{xcolor}
\usepackage[colorlinks=true,linkcolor=PineGreen, citecolor=DarkOrchid, urlcolor=BurntOrange, pagebackref]{hyperref}

\usepackage{color,graphicx,tikz, tkz-euclide, float, pgfplots, overpic}
\usepackage{subfigure}
\usetikzlibrary{arrows,calc, decorations.markings, positioning, fadings}
\usetkzobj{all}
\newtheorem{theorem}{Theorem}[section]

\newtheorem{proposition}[theorem]{Proposition}
\newtheorem{corollary}[theorem]{Corollary}

\newtheorem{exercise}{Exercise}

\theoremstyle{definition}
\newtheorem{definition}[theorem]{Definition}

\theoremstyle{remark}
\newtheorem{remark}[theorem]{Remark}
\newtheorem{note}[theorem]{Note}
\newtheorem{example}[theorem]{Example}%[chapter]
\newtheorem{rhp}{RH Problem}
\newcommand{\rhref}[1]{RH Problem~\ref{#1}}

\DeclareMathOperator{\res}{Res}
\DeclareMathOperator{\rank}{rank}
\DeclareMathOperator{\End}{End}

\newcommand{\D}{\ensuremath{\,\mathrm{d}}}			%an upright d for infinitesimals
\newcommand{\defeq}{\vcentcolon=}
\newcommand{\eqdef}{=\vcentcolon}
\newcommand{\tr}{\mbox{tr}}
\newcommand{\diag}{\mbox{diag}}

\newcommand{\E}{\mathrm{e}}

\makeatletter
\renewcommand*\env@matrix[1][\arraystretch]{%
  \edef\arraystretch{#1}%
  \hskip -\arraycolsep
  \let\@ifnextchar\new@ifnextchar
  \array{*\c@MaxMatrixCols c}}
\makeatother

\let\originalleft\left
\let\originalright\right
\renewcommand{\left}{\mathopen{}\mathclose\bgroup\originalleft}
\renewcommand{\right}{\aftergroup\egroup\originalright}
\title{Discrete integrable systems, Darboux transformations, and Yang-Baxter maps}

\author{Deniz Bilman}
\author{Sotiris Konstantinou-Rizos}

\address{
Department of Mathematics, University of Michigan, Ann Arbor, MI 48109, USA.
}
\email{bilman@umich.edu}
\address{
Institute of Mathematical Physics \& Seismodynamics, Kievskaya 33, Grozny 364068, Russia.
}
\email{skonstantin@chesu.ru, skonstantin84@gmail.com}

\date{\today}

\begin{document}

\begin{abstract}
These lecture notes are devoted to the integrability of discrete systems and their relation to the theory of Yang-Baxter (YB) maps. Lax pairs play a significant role in the integrability of discrete systems. We introduce the notion of Lax pair by considering the well-celebrated doubly-infinite Toda lattice. In particular, we present solution of the Cauchy initial value problem via the method of the inverse scattering transform, provide a review of scattering theory of Jacobi matrices, and give the Riemann-Hilbert formulation of the inverse scattering transform. On the other hand, the Lax-Darboux scheme constitutes an important tool in the theory of integrable systems, as it relates several concepts of integrability. We explain the role of Darboux and B\"acklund trasformations in the theory of integrable systems, and we show how they can be used to construct discrete integrable systems via the Lax-Darboux scheme. Moreover, we give an introduction to the theory of Yang-Baxter maps and we show its relation to discrete integrable systems. Finally, we demonstrate the construction of Yang-Baxter maps via Darboux transformations, using the nonlinear Schr\"odinger equation as illustrative example.
\end{abstract}

\maketitle
\section{Introduction}
\label{sec:1}
Discrete systems, namely systems with their independent variables taking discrete values, are of particular interest and have many applications in several sciences as physics, biology, financial mathematics, as well as several other branches of mathematics, since they are essential in numerical analysis. Initially, they were appearing as discretizations of continuous equations, but now discrete integrable systems, and in particular those defined on a two-dimensional lattice, are appreciated in their own right from a theoretical perspective.

As in the continuous case, the definition of integrability for discrete systems is itself highly nontrivial; there are several opinions on what ``integrable'' should mean, which makes the definition of integrability elusive, rather than tangible. In fact, a comprehensive definition of integrability is not yet available. As working definitions we often use the existence of a Lax pair, the solvability of the system by the inverse scattering transform method, the existence of infinitely many symmetries or conservation laws, or the existence of a sufficient number of first integrals which are in involution (Liouville integrability). For infinite dimensional systems, the existence of a Lax pair provides a change of variables (through ``scattering data'' associated to the Lax operator) which linearizes the flow. Thus, existence of a Lax pair can be taken as a practical definition for integrability for systems with infinite dimensional phase space, e.g.\ PDEs. For more information on \textit{what is integrability} one can consult \cite{Zakharov1991} and the references therein.

Historically, the study of discrete systems and their integrability earned its interest in late seventies; Hirota studied particular discrete systems in 1977, in a series of papers \cite{HirotaKdV, HirotaToda, HirotaSG, HirotaTBT} where he derived discrete analogues of many already famous PDEs. In the early eighties, semi-discrete and discrete systems started appearing in field-theoretical models in the work of Jimbo and Miwa; they also provided a method of generating discrete soliton equations \cite{Jimbo-Miwa,Jimbo-Miwa-II,Jimbo-Miwa-III,Jimbo-Miwa-IV,Jimbo-Miwa-V}. Shortly after, Ablowitz and Taha in a series of papers \cite{Ablowitz-I, Ablowitz-II, Ablowitz-III} are using numerical methods in order to find solutions for known integrable PDEs, using as basis of their method some partial difference equations, which are integrable in their own right. Moreover, Capel, Nijhoff, Quispel and collaborators provided some of the first systematic tools for studying discrete integrable systems and, in particular, for the direct construction of integrable lattice equations (we indicatively refer to \cite{Capel-Nijhoff,Capel-Nijhoff-II}); that was a starting point for new systems of discrete equations to appear in the literature.

In 1991 Grammaticos, Papageorgiou and Ramani proposed the first discrete integrability test, known as \textit{singularity confinement}\index{singularity confinement} \cite{GRP}, which is similar to that of the Painlev\'e property for continuous integrability. However, as mentioned in \cite{Gramm-Schw-Tam}, it is not sufficient criterion for predicting integrability, as it does not furnish any information about the rate of growth of the solutions of the discrete integrable system.

As in the continuous case, the usual integrability criterion being used for discrete systems is the existence of a Lax pair. The existence of such pair is the key point to the integrability of a nonlinear system under the inverse scattering transform. On the other hand, Darboux transformations (DTs) associated to Lax operators constitute very important tools in the theory of integrable systems, since they link continuous integrable systems to discrete integrable ones. Moreover, the study of Darboux transformations gives rise to several other notions of integrability, such as B\"acklund transformations (BTs), conservation laws, symmetries etc. Additionally, the associated Darboux matrices can be used to construct Yang-Baxter maps \cite{Sokor-Sasha, SMJP}. Yet another significant integrability criterion for difference equations and systems of difference equations is the so-called 3D-\textit{consistency}\index{3D-consistency} and, by extension, the \textit{multidimensional consistency}\index{multidimensional consistency}, which were proposed independently by Nijhoff in 2001 \cite{Frank4} and Bobenko and Suris in 2002 \cite{Bobenko-Suris}. As we shall see later on, there is a strict relation between the 3D-consistency and the Yang-Baxter equation.

This chapter splits into two logical parts: In the first part, namely Section \ref{sec:2}, we explain the basic steps of the inverse scattering transform method for solution of the Cauchy initial value problem for an integrable equation, using as illustrative example the famous doubly-infinite Toda lattice: a discrete space - continuous time system. More specifically, we begin with the classical problem of a doubly-infinite one-dimensional chain of interacting particles, which becomes completely integrable (the Toda lattice) if the interaction potential is the Toda potential. As is well-known, the Toda lattice equations can be recast as an isospectral deformation on Jacobi matrices and this gives rise to the existence of a Lax pair. Thus, we move on to cover scattering theory for doubly-infinite Jacobi matrices, introduce the direct scattering transform and scattering data associated with a Jacobi matrix. Then, we cover the time evolution of the scattering data under the dynamics induced by the Toda lattice equations, and discuss the Riemann-Hilbert formulation of the inverse scattering transform. We finally give a brief description and literature survey of analogous techniques applied to the finite Toda lattice and periodic Toda lattice systems.

In the second part, namely Section \ref{sec:3} and Section \ref{sec:4}, we introduce the Darboux-Lax (or Lax-Darboux) scheme \cite{Berkeley, SPS, Sasha, Sasha-miky-JP} and we show the relation between Darboux transformations and discrete integrable systems. Moreover, we show the relation between the 3D-consistency and the Yang-Baxter equation, and we show how, using Darboux transformations, we can construct Yang-Baxter maps which can be restricted to completely integrable ones on invariant leaves. More specifically, in Section \ref{sec:3} we give a brief introduction to Darboux and B\"acklund transformations and their role in the theory of integrable systems. Then, we explain the basic points of the Darboux-Lax scheme, and we demonstrate them using the nonlinear Schr\"odinger (NLS) equation as an illustrative example. In particular, studying the Darboux transformations associated to the NLS equation we first derive a discrete integrable system, for which we present the initial value problem on the staircase. Then, using certain first integrals we reduce this discrete system to an Adler-Yamilov type of system \cite{Adler-Yamilov}. Moreover, in the same way we construct the discrete Toda equation \cite{Suris}. In Section \ref{sec:4} we give an introduction to equations on quad graphs \& the 3D-consistency criterion, and we present some recent classification results. Then, we give and introduction to the theory of Yang-Baxter maps, namely set-theoretical solutions of the Yang-Baxter equation, and we explain their relation with 3D-consistent equations. We focus on those Yang-Baxter maps which possess Lax-representation and we show how one can construct them using Darboux transformations. As an illustrative example we use the Darboux transformation of the NLS equation, which was presented in Section \ref{sec:3}, and we use it to construct a six-dimensional Yang-Baxter map \cite{Sokor-Sasha}. The former can be restricted to the completely integrable Adler-Yamilov map on symplectic leaves \cite{Sokor-Sasha}.

\section{The Toda Lattice}
\label{sec:2}

\subsection{One-dimensional chain of particles}
\label{sec:1dchain}
\begin{figure}[htp!]
	\centering
\begin{tikzpicture}[scale=1]
%\filldraw[gray,opacity=.1] (-5,0) circle (2.2pt);
\filldraw[black] (-4,0) circle (2.2pt);
\filldraw[black] (-3,0) circle (2.2pt);
\filldraw[black] (-2,0) circle (2.2pt);
\filldraw[black] (-1,0) circle (2.2pt);
\filldraw[black] (0,0) circle (2.2pt);
\filldraw[black] (1,0) circle (2.2pt);
\filldraw[black] (2,0) circle (2.2pt);
\filldraw[black] (3,0) circle (2.2pt);
\filldraw[black] (4,0) circle (2.2pt);
%\filldraw[gray,opacity=.1] (5,0) circle (2.2pt);
%\draw[thick] (-3,0)--(-2.7,0)--(-2.6,-0.2)--(-2.4,0.2)--(-2.3,0)--(-2,0)--(-1.7,0)--(-1.6,-0.2)--(-1.4,0.2)--(-1.3,0)--(-1,0)--(-0.7,0)--(-0.6,-0.2)--(-0.4,0.2)--(-0.3,0)--(0,0)--(0.3,0)--(0.4,-0.2)--(0.6,0.2)--(0.7,0)--(1,0)--(1.3,0)--(1.4,-0.2)--(1.6,0.2)--(1.7,0)--(2,0)--(2.3,0)--(2.4,-0.2)--(2.6,0.2)--(2.7,0)--(3,0);
%\draw[thick,black,opacity=.2](-4.92,0)--(-4.8,0);
\draw[thick,black,decoration={aspect=0.4, segment length=1mm, amplitude=1mm,coil},decorate,opacity=1,path fading=west] (-4.8,0)--(-4.2,0);
\draw[thick,black,opacity=1](-4.2,0)--(-4.08,0);

\draw[thick,black](-3.92,0)--(-3.8,0);
\draw[thick,black,decoration={aspect=0.4, segment length=1mm, amplitude=1mm,coil},decorate] (-3.8,0)--(-3.2,0);
\draw[thick,black](-3.2,0)--(-3.08,0);

\draw[thick,black](-2.92,0)--(-2.8,0);
\draw[thick,black,decoration={aspect=0.4, segment length=1mm, amplitude=1mm,coil},decorate] (-2.8,0)--(-2.2,0);
\draw[thick,black](-2.2,0)--(-2.08,0);

\draw[thick,black](-1.92,0)--(-1.8,0);
\draw[thick,black,decoration={aspect=0.4, segment length=1mm, amplitude=1mm,coil},decorate] (-1.8,0)--(-1.2,0);
\draw[thick,black](-1.2,0)--(-1.08,0);

\draw[thick,black](-0.92,0)--(-0.8,0);
\draw[thick,black,decoration={aspect=0.4, segment length=1mm, amplitude=1mm,coil},decorate] (-0.8,0)--(-0.2,0);
\draw[thick,black](-0.2,0)--(-0.08,0);

\draw[thick,black](0.08,0)--(0.2,0);
\draw[thick,black,decoration={aspect=0.4, segment length=1mm, amplitude=1mm,coil},decorate] (0.2,0)--(0.8,0);
\draw[thick,black](0.8,0)--(0.92,0);

\draw[thick,black](1.08,0)--(1.2,0);
\draw[thick,black,decoration={aspect=0.4, segment length=1mm, amplitude=1mm,coil},decorate] (1.2,0)--(1.8,0);
\draw[thick,black](1.8,0)--(1.92,0);

\draw[thick,black](2.08,0)--(2.2,0);
\draw[thick,black,decoration={aspect=0.4, segment length=1mm, amplitude=1mm,coil},decorate] (2.2,0)--(2.8,0);
\draw[thick,black](2.8,0)--(2.92,0);

\draw[thick,black](3.08,0)--(3.2,0);
\draw[thick,black,decoration={aspect=0.4, segment length=1mm, amplitude=1mm,coil},decorate,] (3.2,0)--(3.8,0);
\draw[thick,black](3.8,0)--(3.92,0);

\draw[thick,black,opacity=1](4.08,0)--(4.2,0);
\draw[thick,black,decoration={aspect=0.4, segment length=1mm, amplitude=1mm,coil},decorate,opacity=1,path fading=east] (4.2,0)--(4.8,0);
%\draw[thick,black,opacity=.2](4.8,0)--(4.92,0);
%opacity=.1

\node (x) at (0,-0.8){\small \hspace{2cm}$q_{n}(t)$: displacement of the $n^{\text{th}}$ particle};
\node (x2) at (0, -1.2){\small \hspace{2.6cm}from its equilibrium position};
%\node (y) at (0,-1.2){\hspace{2cm}$q_{n}(t)$: displacement of the $n^{\textsf{th}}$ particle};
%\node (ddd1) at (5.5,0){\color{gray!50}$\cdots$};
%\node (ddd2) at (-5.5,0){\color{gray!50}$\cdots$};
\node (n) at (0,0.48){\small$n$};
\node (n1) at (1,0.5){\small $n+1$};
\node (n2) at (2,0.5){$\cdots$};
\node (n3) at (-1,0.5){\small$n-1$};
\node (n4) at (-2, 0.5){$\cdots$};
\draw (0,-0.2)--(0,-0.6);
\draw[->](0.1,-0.6)--(0.5,-0.6);
\draw[->](-0.1,-0.6)--(-0.5,-0.6);
\draw[black,thin,<->] plot[smooth] coordinates {
    (-1.6,-0.2)
     (-2,-0.45)
     (-2.5,-0.7)
     (-3,-0.95)
     (-3.5,-1.1)
     (-4,-0.6)
   };
%\draw[to reversed-to,black](-1.7,-0.2) .. controls +(-1.5,-0.5) and +(1,-1.5) .. +(-2,-0.5);
%(-5.1,-0.6) .. controls +(-4.5,-1.5) and +(3,-4) .. +(-6,-1.4);
%\node (p) at (-12.5,-1.5){nonlinear springs};
\node[black] (p) at (-4,-0.5){\small nonlinear springs};
\end{tikzpicture}
\caption{One-dimensional chain of particles with nearest neighbor interactions.}
\label{F:lattice}
\end{figure}
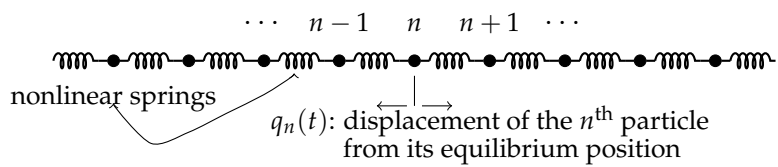
Consider the classical problem of one-dimensional chain of particles on a line with nearest neighbor interactions as depicted in Figure~\ref{F:lattice}. Assume that each particle has unit mass, and that there are no impurities, i.e.\ the potential energies of the springs between the particles are identical. In this treatment, our focus is going to be the doubly-infinite lattices. Therefore, unless otherwise noted, we assume that there are infinitely many particles on a line. We let $V\colon \mathbb{R}\to \mathbb{R}$ denote the uniform interaction potential between the neighboring particles. With the aforementioned assumptions, the equations of motion that govern this system are given by
\begin{equation}
\frac{\D^2 q_n}{\D t^2}=V'(q_{n+1}-q_n) - V'(q_{n}-q_{n-1}),\quad n\in\mathbb{Z},
\label{eq:eom}
\end{equation}
where $q_n$ stands for the displacement of the $n^\mathrm{th}$ particle from its equilibrium position. If $V'(r)=\frac{\D V(r)}{\D r}$, that is if $V'(r)$ is proportional to $r$, that is $V$ is a harmonic potential, then the force which governs the interaction between the particles is linear. In this case, the solutions are given by superpositions of normal modes and there is no transfer of energy between these modes.

The general belief in the early 1950s was that if a nonlinearity is introduced in the interaction of these particles, then energy would flow between the different modes, eventually leading to a stable state of statistical equilibrium, i.e.\ thermalization. In the summer of 1953 at Los Alamos National Laboratory, E.~Fermi, J.~Pasta, and S.~Ulam, together with M.~Tsingou, set out to numerically study the thermalization process in solids by conducting one of the first-ever numerical experiments using the first electronic computer \textbf{MANIAC} -- \textbf{M}athematical \textbf{A}nalyzer, \textbf{N}umerical \textbf{I}ntegrator and \textbf{C}omputer. To model solids, they used the aforementioned one-dimensional chain of particles with 32 particles and 64 particles, equipped interaction potentials which had weak nonlinear terms. More explicitly, they considered the potentials
\begin{equation}
\begin{aligned}
V_{\alpha}(r) &= \frac{1}{2}r^2 + \alpha r^3,\\
V_{\beta}(r)&= \frac{1}{2}r^2 + \beta r^4,
\end{aligned}
\end{equation}
with $\alpha$ and $\beta$ being small and positive constants. The expectation was to observe equipartition of energy as time elapses due to the presence of nonlinear contributions in the potentials. One day, they forgot to terminate the experiment and it went on over the weekend. To their surprise, they found that the system exhibited quasiperiodic behavior. The energy of the system revisited the initially excited modes and the initial state was almost exactly recovered \cite{FPU}. This phenomenon, known as the FPU recurrence, has been studied from various perspectives, including ergodicity, Poincar\'e recurrence theory, and KAM theory (see \cite{FPUsurvey} for a survey article on the so-called FPU Experiment, \cite{Ford1} for an earlier article and the references therein.) 10 years after the experiment, M.~Kurskal and N.~Zabusky made the gound-breaking discovery of the ``soliton'' solution of the Korteweg-de Vries (KdV) equation in their pioneering work \cite{ZK}. N.~Zabusky and M.~Kruskal coined the name `soliton' to these traveling ``solitary wave'' solutions of the KdV equation because they retained their speed and shape upon interacting with such other waves -- they interacted as particles. The observation that the lattices used in the FPU Experiment approximated KdV equation~\cite{ZK}, which exhibited solitonic behavior, in an appropriate continuum limit provided an explanation for the FPU recurrence. This work led to a big growth in research on nonlinear waves, particularly on solitons. It also triggered a search for an interaction potential for which the resulting lattice system possesses traveling solitary wave solutions.
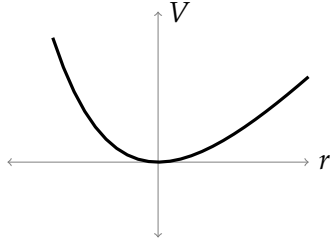
\begin{figure}[htp]
\centering
\begin{tikzpicture}
\draw[help lines,<->] (-2,0) -- (2,0);
\draw[help lines,<->] (0,-1) -- (0,2);
\draw[black, line width=1.2, domain=-1.4:2] plot (\x, {exp(-\x)+\x-1});
\node[right] at (2,0) {$r$};
\node[right] at (0,2) {$V$};
\end{tikzpicture}
\caption{The graph $V=V_{\mathrm{Toda}}(r)$.}
\label{f:Toda}
\end{figure}
In 1972, while working on elliptic functions, M. Toda considered the exponential interaction potential (see Figure~\ref{f:Toda})
\begin{equation}
V_{\mathrm{Toda}}(r) = \E^{-r} + r - 1,
\label{eq:TodaPot}
\end{equation}
and found out that the resulting lattice, now known as the Toda lattice, has soliton solutions~\cite{Toda67, TodaBook}.

The Toda potential \eqref{eq:TodaPot} leads to an explicit form of the evolution equations \eqref{eq:eom}, namely:
\begin{equation}
\frac{\D^2 q_n}{\D t^2} = \E^{q_n - q_{n+1}}-\E^{q_{n-1} - q_{n}},\quad n\in\mathbb{Z},
\label{eq:eom-q}
\end{equation}
where we suppressed the time dependence of $q$\footnote{We employ two notational conventions in this section. We use bold capital letters to denote a matrix, say $\mathbf{A}$, and use the regular capital type of the same letter to denote its entries: $A_{ij}$. Whenever it is clear from the context, we use $x$ to denote a sequence $\lbrace{x_n}\rbrace_{n\in\mathbb{Z}}$.}. If we denote the momentum of the $n^{\mathrm{th}}$ particle at a time $t$ by $p_n(t)$, then we can rewrite \eqref{eq:eom-q} as the first order system:
\begin{equation}
\begin{aligned}
\frac{\D p_n}{\D t} &= \E^{-\left(q_n - q_{n-1}\right)}-\E^{-\left(q_{n+1} - q_{n}\right)},\\
\frac{\D q_n}{\D t} &= p_n,
\end{aligned}
\label{eq:eom-pq}
\end{equation}
for each $n\in\mathbb{Z}$. If we assume that $q_{n+1}-q_n \to 0$ and $p_n \to 0$ sufficiently fast as $|n|\to \infty$, we can recast \eqref{eq:eom-pq} as a Hamiltonian system of equations
\begin{equation}
\label{eq:Hamiltonian}
\begin{aligned}
\frac{\D p_n}{\D t} &= - \frac{\partial \mathcal{H}(p,q)}{\partial q_n},\\
\frac{\D q_n}{\D t} &= \frac{\partial \mathcal{H}(p,q)}{\partial p_n},
\end{aligned}
\end{equation}
with the Hamiltonian $\mathcal{H}(p,q)$:
\begin{equation}
\mathcal{H}(p,q)=\mathcal{H}_{\mathrm{Toda}}(p,q)\defeq\sum_{n\in\mathbb{Z}}\frac{1}{2}p_n^2 + V_{\mathrm{Toda}}\left(q_{n+1}-q_n\right).
\label{eq:Toda-ham}
\end{equation}

\subsection{Solitons}\label{sec:solitons}
As mentioned earlier, the doubly-infinite Toda lattice has soliton solutions. Solitons are localized traveling waves (solitary waves) which interact like particles: they preserve their shapes and speeds after interacting with another wave. If two 1-solitons interact, no dispersive wave is generated after the interaction. A 1-soliton solution travels with a constant speed and constant amplitude, which are proportional. 1-soliton solutions for the displacements in the Toda lattice are given explicitly by the following 2-parameter family:
\begin{equation}
q^{[1]}_n(t) = q_{\text{R}} + \log \left(\frac{1+ \frac{\gamma}{1-\E^{-2\kappa}} \E^{-2\kappa(n-1) - 2k\sigma c t}  }{1+ \frac{\gamma}{1-\E^{-2\kappa}} \E^{-2\kappa n - 2\kappa\sigma c t}} \right),
\label{eq:1sol}
\end{equation}
where $\gamma >0$, $\kappa>0$ is the wave number, $c=\tfrac{\sinh \kappa}{\kappa}>1$ is the speed of propagation, and $\sigma = \pm 1$ is the constant determining direction of propagation. Here $q_{\text{R}}=\lim_{n\to+\infty} q_n(t)$ for all finite $t$. See Figure~\ref{f:toda1sol} for a plot of such a solution.
\begin{figure}
\centering
\includegraphics[scale=0.5]{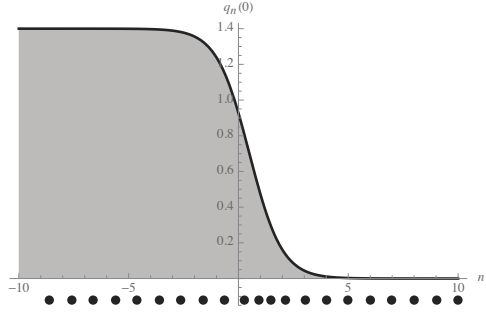}
\caption{1-soliton solution of the Toda lattice at time $t=0$.}
\label{f:toda1sol}
\end{figure}

In fact there are $N$-soliton solutions, for $N\in\mathbb{N}$, whose formulae are of the form
\begin{equation}\label{eq:Nsol}
q^{[N]}_n(t)=q_{\text{R}} + \log\left( \frac{\det(\mathbb{I}_{N} + \mathbf{C}^{[N]}(n,t))}{\det(\mathbb{I} + \mathbf{C}^{[N]}(n+1,t))} \right),
\end{equation}
where $\mathbb{I}_{N}$ is the $N\times N$ identity matrix. The entries of the $N\times N$ matrices $\mathbf{C}^{[N]}(n,t)$ are given by
\begin{equation*}
C^{[N]}_{ij}(n,t) = \left( \frac{\sqrt{\hat{\gamma}_{j}(n,t) \hat{\gamma}_{k}(n,t) } }{1- \E^{-(\kappa_j + \kappa_k)}} \right),\quad \hat{\gamma}_j(n,t) = \gamma_j \E^{-2\kappa_j n - 2 \sigma_j \sinh(\kappa_j)t},
\end{equation*}
for $\gamma_j>0$, $\kappa_j>0$, and $\sigma_j \in \{-1,1\}$.  The case $N = 1$ coincides with the 1-soliton solution \eqref{eq:1sol}. Asymptotically, as $t\to\infty$, the $N$-soliton solution can be written as a sum of 1-soliton solutions as was proved in \cite{KruTes}.

We shall see below that these solutions are \emph{reflectionless} solutions since the reflection coefficient in the associated scattering data is identically zero.

\subsection{Complete Integrability and Scattering Theory}\label{sec:integrability}
Complete integrability of the Toda lattice equations \eqref{eq:Toda-ham} can be established by considering an eigenvalue problem for a second order linear difference operator.
\subsubsection{Existence of a Lax pair}
In 1974, H.~Flaschka~\cite{Flaschka,Flaschka2} and S.~Manakov~\cite{Manakov} independently and simultaneously introduced the following variables to obtain a first-order system of differential equations that is equivalent to the Toda lattice:
\begin{equation}
a_n \defeq\frac{1}{2}\E^{-(q_{n+1} - q_{n})/2},\quad b_n \defeq -\frac{1}{2}p_n.
\label{eq:ab}
\end{equation}
If $q_n \to q_{\text{R},\text{L}}\in\mathbb{R}$ sufficiently fast as $n\to\pm \infty$, the inverse map is given by
\begin{equation*}
q_n = q_{\text{R}} + 2 \log \left(\prod_{k=n}^\infty 2 a_n \right),\quad p_n = - 2 b_n,
\end{equation*}
and \eqref{eq:ab} is a bijection. Note that $q_n \to q_{\text{R},\text{L}}$ and $p_n\to 0$ as $n\to\pm\infty$ corresponds to $a_n \to \tfrac{1}{2}$ and $b_n \to 0$ as $|n|\to \infty$. It is immediately verified that the pair of sequences $(p, q)$ satisfies the equations \eqref{eq:eom-pq} if and only if the pair $(a,b)$ defined via \eqref{eq:ab} satisfies the equations
\begin{equation}
\begin{aligned}
\frac{\D a_n}{\D t} &= a_n \left(b_{n+1} - b_n \right),\\
\frac{\D b_n}{\D t} &= 2 \left(a_n^2 - a_{n-1}^2 \right).
\end{aligned}
\label{eq:eom-ab}
\end{equation}
Using the pair of sequences introduced in $\eqref{eq:ab}$, define the operators $\mathbf{L}$ and $\mathbf{P}$ on the Hilbert space $\ell^{2}(\mathbb{Z})$ of square-summable sequences:
\begin{equation}
\begin{aligned}
(\mathbf{L}\phi)_n &\defeq a_{n-1}\phi_{n-1} + b_n \phi_n + a_n \phi_{n+1},\\
(\mathbf{P}\phi)_n &\defeq -a_{n-1}\phi_{n-1} + a_n \phi_{n+1}.
\end{aligned}
\label{eq:PL-operator}
\end{equation}
In the standard basis, $\mathbf{L}$ and $\mathbf{P}$ have doubly-infinite matrix representations.They are given explicitly by:
\begin{equation}
\mathbf{L}=
\begin{pmatrix}
  \ddots    &\ddots & 	\ddots	 & 		   \\
   \ddots & b_{n-1} & a_{n-1} & 0           \\
 \ddots          & a_{n-1} & b_n    & a_n     & \ddots        	 \\
	   & 0         & a_n    & b_{n+1} & \ddots       	 \\
		   &		  & \ddots      &\ddots & \ddots         \\
\end{pmatrix}\,,\quad
\mathbf{P}=
\begin{pmatrix}
  \ddots    &\ddots & 	\ddots	 & 		   \\
   \ddots & 0 & a_{n-1} & 0           \\
 \ddots          & -a_{n-1} & 0    & a_n     & \ddots        	 \\
	   & 0         & -a_n    &0 & \ddots       	 \\
		   &		  & \ddots      &\ddots & \ddots         \\
\end{pmatrix}.
\label{eq:PL}
\end{equation}
Note that $\mathbf{L}$ is a doubly-infinite Jacobi matrix: symmetric, tridiagonal with positive off-diagonal entries. The following proposition gives the existence of a Lax pair for the Toda lattice.

\begin{proposition}[H.~Flaschka~\cite{Flaschka,Flaschka2}, S.V.~Manakov~\cite{Manakov}]\label{pr:Lax}
The Toda lattice equations \eqref{eq:eom-ab} are equivalent to the matrix equation
\begin{equation}
\frac{\D\mathbf{L}}{\D t} = [\mathbf{P}, \mathbf{L}],
\label{eq:Lax}
\end{equation}
where $[\mathbf{P}, \mathbf{L}]$ denotes the matrix commutator, $[\mathbf{P}, \mathbf{L}]\defeq \mathbf{P}\mathbf{L}-\mathbf{L}\mathbf{P}$.
\end{proposition}
\begin{proof}
By direct calculation. Left as an exercise.
\end{proof}
Equation \eqref{eq:Lax} is called the Lax equation after P.~D.~Lax \cite{Lax}; and the pair $(\mathbf{L}, \mathbf{P})$ is called a Lax pair. We will call $\mathbf{L}$ the Lax operator for the Toda lattice. Note that the matrix $\mathbf{P}$ defined in \eqref{eq:PL} depends on $\mathbf{L}$: $\mathbf{P}= \mathbf{L}_{+} - \mathbf{L_{-}}$, where $\mathbf{L}_{\pm}$ stands for the upper ($+$) and the lower ($-$) diagonal parts of $\mathbf{L}$.

At this point, it must be stressed out, that equation \eqref{eq:Lax}, namely the \textit{Lax formulation}, on one hand constitutes a basis for the inverse scattering method, but it also possesses a deeper property, i.e. the so-called \textit{Hamiltonian formulation}. An equation (or system of equations) is said to have a Hamiltonian formulation if they can be written as a (perhaps infinite-dimensional) classical Hamiltonian system. In general, Hamiltonian formulation is a characteristic of all equations (or systems of equations) which are solvable by the inverse scattering transform. Moreover, in our case, the inverse scattering transform can be understood as a canonical transformation of the associated Hamiltonian structure. The first example demonstrating the relation between the Lax and the Hamiltonian formulation was the famous KdV equation -- to which we shall come back in the next subsection -- in 1971 \cite{FaddeevZakharov}. Just as the KdV equation, the Toda lattice also possesses infinitely many conserved quantities (see Exercise \eqref{ex:conserved}). For further study on Hamiltonian structure of Lax equations for integrable systems, the reader may consult \cite{AbloClar, AbloSegur},  or \cite{FaddeevTakhtajan} (and the references therein) which is a self-contained textbook.

\begin{remark}
Proposition~\ref{pr:Lax} holds for the finite Toda lattice with the boundary condition $a_{-1}=a_{N-1}=0$ for some integer $N>0$. In fact, complete integrability was first proved (\cite{Flaschka,Flaschka2} and \cite{Manakov}) for this finite version of the Toda lattice, and was later generalized to the infinite lattice. The equations of motion in the $(a,b)$ variables for the finite the Toda lattice are given by:
\begin{equation}
\begin{aligned}
\frac{\D b_0}{\D t} &= 2 a_0^2,\\
\frac{\D b_n}{\D t} &= 2\left(a_n^2 - a_{n-1}^2 \right),\quad n=1,2,\dots,N-2,\\
\frac{\D a_n}{\D t} &= a_n \left(b_{n+1} - b_n \right),\quad n=0,1,2,\dots, N-2,\\
\frac{\D b_{N-1}}{\D t} &= - 2 a_{N-2}^2\,.
\end{aligned}
\label{eq:finiteToda}
\end{equation}
Note that the condition $a_{-1}= a_{N-1} = 0$ corresponds to setting the relative displacements $q_N - q_{N-1}$ and $q_0 - q_{-1}$ to be infinite. This is sometimes called a \emph{reservoir} condition in the literature.
\label{r:finiteToda}
\end{remark}
\begin{exercise}
Show that if the Lax equation \eqref{eq:Lax} holds for $\mathbf{L}(t)$, then it also holds for the matrix power $\mathbf{L}(t)^m$, for any $m\in\mathbb{N}$. Using this, find an infinite sequence of conserved quantities for the Toda lattice equations: $\mathrm{trace}\,\left(\mathbf{L}(t)^m - \mathbf{L_\infty}^{m}\right)$, where $\mathbf{L_\infty}$ is the doubly infinite Jacobi matrix with $\lim_{|n|\to\infty}b_n = 0$ on its diagonal and $\lim_{|n|\to\infty}a_n=1/2$ on its off-diagonal.
\label{ex:conserved}
\end{exercise}

\begin{exercise}[G.~Teschl~\cite{TeschlBook}]\label{ex:Teschl}For two sequences $\psi,\phi\in\ell(\mathbb{Z})$ define
    \begin{equation*}
        \mathcal{G}(\psi,\phi)(n)= \psi_n (\mathbf{L}\phi)_n- \phi_n (\mathbf{L}\psi)_n.
    \end{equation*}
Prove Green's formula:
\begin{equation}
\sum\limits_{j=m}^{n}\mathcal{G}(\phi,\psi)(j) =\mathcal{W}_{n}(\psi,\phi)-\mathcal{W}_{m-1}(\psi,\phi),
\label{ex:GreensId}
\end{equation}
where $\mathcal{W}_n(\phi,\psi)$ stands for the Wronskian which is defined by
\begin{equation}
\mathcal{W}_n(\psi,\phi) = a_n\left(\psi_n \phi_{n+1} - \psi_{n+1}\phi_n \right).
\label{eq:wronski}
\end{equation}
\end{exercise}

Before we cover scattering data associated with $\mathbf{L}$ we have the following theorem that summarizes its basic properties.
\begin{theorem}[from Theorem 1.5~\cite{TeschlBook}]
Assume that $a,b \in \ell^{\infty}(\mathbb{Z})$, with $a_n >0$ for all $n\in \mathbb{Z}$. Then $\mathbf{L}$ defined in \eqref{eq:PL-operator} is a bounded self-adjoint operator on $\ell^2(\mathbb{Z})$. Moreover, $a,b \in \ell^{\infty}(\mathbb{Z})$ if and only if $\mathbf{L}$ is bounded on $\ell^2(\mathbb{Z})$.
\label{T:propL}
\end{theorem}
\begin{proof}
For $\psi,\phi\in\ell^2(\mathbb{Z})$ we have $\lim_{n\to\pm\infty} \mathcal{W}_n(\psi,\phi) = 0$, where $\mathcal{W}_n(\cdot, \cdot)$ is the Wronskian defined in \eqref{eq:wronski}. Using this together with Green's formula \eqref{ex:GreensId} from Exercise~\ref{ex:Teschl} implies that
\begin{equation*}
\langle \phi,\mathbf{L}\psi \rangle_{\ell^2} =\langle \mathbf{L}\phi,\psi \rangle_{\ell^2},
\end{equation*}
for all $\phi,\psi\in\ell^2(\mathbb{Z})$, proving that $\mathbf{L}$ is self-adjoint. Now, if $a,b\in\ell^{\infty}(\mathbb{Z})$, then for any $\phi\in\ell^2(\mathbb{Z})$
\begin{equation*}
|\langle \phi, \mathbf{L}\phi\rangle_{\ell^2}| \leq \left(2 \|a\|_{\infty} + \|b\|_{\infty}\right) \|\phi \|_2^2,
\end{equation*}
which implies that $\|\mathbf{L} \| \leq 2 \|a\|_{\infty} + \|b\|_{\infty}$, where $\| \cdot \|$ denotes the operator norm. On the other hand, assume that $\mathbf{L}$ is bounded. Let $\big\{ \delta^{[k]}_n \big\}_{n\in\mathbb{Z}}$ denote the sequence defined by $\delta^{[k]}_n = 0$ if $n\neq k$ and $\delta^{[k]}_k = 1$. Then for any $k\in\mathbb{Z}$
\begin{equation*}
a_k^2 + a_{k-1}^2 + b_k^2 = \|\mathbf{L}\delta^{[k]} \|_2^2 \leq \| \mathbf{L}\|^2,
\end{equation*}
which implies that $a$ and $b$ belong to $\ell^\infty(\mathbb{Z})$. This completes the proof.
\end{proof}
For a detailed treatment on Jacobi matrices and the associated difference operators we refer the reader to~\cite{TeschlBook}.

As is well-known, the integrability of the Toda lattice can be exploited via the bijective correspondence between the Lax operator, which in this case is the Jacobi matrix $\mathbf{L}$, and its scattering data. This correspondence goes under the name of direct and inverse scattering theory, and has already been studied in detail. While we do not attempt to give a comprehensive survey of the relevant references, the interested reader may enter the subject, for example, through \cite{TeschlBook} or \cite{thesis}, and the references therein. We now proceed with the spectral properties of the Lax operator $\mathbf{L}$ and definition of the scattering data.

\subsubsection{Spectral properties of the Lax operator $\mathbf{L}$}
For the moment, we forget about the time dependence and begin with a brief study of the spectrum associated with the doubly-infinite Jacobi matrix $\mathbf{L}$ given in \eqref{eq:PL}. First, recall from Theorem~\ref{T:propL} that $\mathbf{L}$ is a bounded self-adjoint operator on $\ell^2(\mathbb{Z})$. Thus it bears no residual spectrum and its spectrum $\sigma(\mathbf{L})$ is a subset of $\mathbb{R}$. Let $\mathcal{M}$ denote the \emph{Marchenko} class of Jacobi matrices whose coefficients $(a,b)\in\ell^{\infty}(\mathbb{Z})\oplus \ell^{\infty}(\mathbb{Z})$ satisfy
\begin{equation}
\sum\limits_{n\in\mathbb{Z}}(1+|n|)\left( \left|a_n - \tfrac{1}{2}\right| + \left| b_n\right|\right)<\infty\,.
\label{eq:Marchenko}
\end{equation}
Throughout this section, we assume that $a_n>0$ for all $n\in\mathbb{Z}$ and that \eqref{eq:Marchenko} holds. These two conditions are preserved by the Toda lattice equations \eqref{eq:eom-ab} (see Theorem~2.5 in \cite{Teschl10}).

The following theorem from \cite{TeschlBook} locates the essential spectrum of $\mathbf{L}$ under an assumption that is weaker than \eqref{eq:Marchenko}. We present it in a version that is simplified for our setting and purposes.
\begin{theorem}[from Theorem 3.19, \cite{TeschlBook}]
\label{th:ac-spec}
Suppose that the sequences $(a,b)\in\ell^{\infty}(\mathbb{Z})\oplus\ell^{\infty}(\mathbb{Z})$, with $a_n >0$ for all $n\in\mathbb{Z}$, satisfy
\begin{equation}
\sum\limits_{n\in\mathbb{Z}}\left|a_{n+1}-a_n \right| + |b_{n+1}-b_n |<\infty.
\end{equation}
Suppose further that $\lim_{|n|\to\infty} a_n = \tfrac{1}{2}$ and $\lim_{|n|\to\infty} b_n =0$. Then the following are true for the essential spectrum $\sigma_{\mathrm{ess}}({\mathbf{L}})$ and the pure point spectrum $\sigma_{\mathrm{pp}}({\mathbf{L}})$ of the associated Jacobi matrix $\mathbf{L}$:
\begin{equation}
\begin{aligned}
\sigma_{\mathrm{ess}}(\mathbf{L}) &= \sigma_{\mathrm{ac}}(\mathbf{L}) = [-1,1],\\
\sigma_{\mathrm{pp}}(\mathbf{L}) &\subset \overline{\mathbb{R}\setminus\sigma_{\mathrm{ess}}(\mathbf{L})}.
\end{aligned}
\label{eq:ac-spec}
\end{equation}
\end{theorem}

Note that due to the closure present in \eqref{eq:ac-spec}, Theorem~\ref{th:ac-spec} does not exclude eigenvalues at the boundary of the essential spectrum, but it does prevent eigenvalues embedded in the interior of the essential spectrum. Here is an example illustrating presence of eigenvalues at the edge of the essential spectrum.
\begin{example}[G.~Teschl, \cite{TeschlBook}]
Consider the sequences $(a,b)$ with
\begin{equation}
a_n = \frac{1}{2},\quad b_n=\frac{2-3n^2}{4+n^4}.
\end{equation}
Then the associated Jacobi matrix has indeed an eigenvalue at $\lambda=1$ with the corresponding eigenfunction $\phi \in \ell^2(\mathbb{Z})$ given by
\begin{equation}
\phi_{n} = \frac{1}{1+n^2}.
\end{equation}
The reason for this is that \eqref{eq:Marchenko} is violated: the first moment $\sum\limits_{n\in\mathbb{Z}}{|n b_n|}$ is divergent.
\end{example}

Under the assumption~\eqref{eq:Marchenko}, $\mathbf{L}$ has finitely many simple eigenvalues with associated eigenvectors in $\ell^{2}(\mathbb{Z})$. $\sigma(\mathbf{L})$ consists of an absolutely continuous part $\sigma_{\text{ac}}(\mathbf{L}) = [-1,1]$ and a finite simple pure point part $\sigma_{\text{pp}}(\mathbf{L})$:
\begin{equation*}
\sigma_{\text{pp}}(\mathbf{L})=\lbrace \lambda_j \colon j=1,2,\dots,N \rbrace \subset (-\infty, -1) \cup (1,+\infty),
\end{equation*}
for some $N\in\mathbb{N}$. The spectrum of $\mathbf{L}$ typically is of the form depicted in Figure~\ref{F:spec_L}.

\begin{figure}[htp]
\centering
\begin{tikzpicture}[scale=1]
\coordinate (z1) at (-3.5,0);
\coordinate (z2) at (-1.8, 0);
\coordinate (z3) at (2.4, 0);
\coordinate (dots) at (3.4,0);
\coordinate (z4) at (4.4, 0);

\draw[black] (-5,0) -- (5,0);
\draw[line width=3,black] (-1,0) -- (1,0);%\foreach \Point in {(-3,0), (0,0), (3,0)}{
%    \node at \Point {$\circ$};
%}
\foreach \Point in {(z1), (z2), (z3), (z4)}{
    \node at \Point {\color{black}\large\textbullet};
}
%    \foreach \Point in {(-2.4,0), (2,0)}{
%        \node[red] at \Point {\textbullet};
%    }
\node[ right] at (5,0) {\small$\mathbb{R}$};
\node at (-1, 0 ) {$[$};
\node at (1, 0 ) {$]$};
\node[below,yshift=-0.4em] at (-1, 0 ) {$-1$};
\node[below,yshift=-0.4em] at (1, 0 ) {$1$};
\node[above] at (z1) {$\lambda_1$};
\node[above] at (-2.65, 0) {$\cdots$};
\node[above] at (z2) {$\lambda_j$};
\node[above] at (z3) {$\lambda_{j+1}$};
\node[above] at (dots) {$\cdots$};
\node[above] at (z4) {$\lambda_N$};
\node[below] at (0,0) {$\color{black}\sigma_{\mathrm{ac}}\left( \mathbf{L}\right)$};
\node[anchor=north west] at (-5,0) {\color{black}$\sigma\left(\mathbf{L}\right)$};
\end{tikzpicture}
\caption{Spectrum $\sigma(\mathbf{L})$ of $\mathbf{L}$ on the $\lambda$-plane, consisting of finitely many real simple eigenvalues $\{\lambda_j\}_{j=1}^{N}$ and the absolutely continuous part $\sigma_{\text{ac}}(\mathbf{L})=[-1,1]$.}
\label{F:spec_L}
\end{figure}
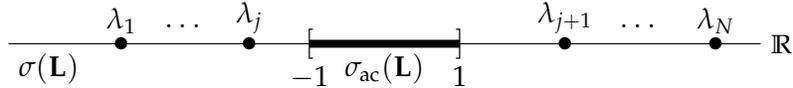
We close this section with a series of remarks.
\begin{remark}
    A Jacobi matrix $\mathbf{J}$ with bounded sequences $b$ on the diagonal and $a$, with $a_n \neq 0$ for all $n\in\mathbb{Z}$, on the off-diagonal, is unitarily equivalent to the Jacobi matrix $\tilde{\mathbf{J}}$ with $b$ on the diagonal and $\{|a_n|\}_{n\in\mathbb{Z}}$ on the off-diagonal. Therefore, the assumption that $a_n>0$ made in this section makes no difference from a spectral theory point of view. If $a_n=0$ for some $n$ however, $\mathbf{J}$ is decomposed into a direct sum of two half-line (infinite) Jacobi matrices. The assumption that $a_n \neq 0$ for all $n$ guarantees that the equation
    \begin{equation*}
        \mathbf{J}\psi = \lambda \psi
    \end{equation*}
    has exactly two linearly independent solutions in $\ell(\mathbb{Z})$ for any $\lambda\in\mathbb{C}$, and that the spectrum of $\mathbf{J}$ has multiplicity at most two \cite{TeschlBook}. Thus, we need the assumption $a_n \neq 0$ to have a well defined, bijective direct and inverse scattering theory to solve the Cauchy initial value problem. Noting that the dynamics given in \eqref{eq:eom-ab} preserve the signs of $a_n$, we can make the assumption $a_n>0$.
\end{remark}
\begin{remark}
The spectral problem
\begin{equation*}
\mathbf{L}\psi = \lambda \psi
\end{equation*}
is a discrete analogue of the Sturm-Liouville eigenvalue problem. For a discrete version of Sturm oscillation theory, see \cite{TeschlOsc} or \cite{SimonOsc} and the references therein.
\end{remark}
\begin{remark}
The spectral properties of the Lax operator $\mathbf{L}$ are similar to the Schr\"odinger operator
\begin{equation}
H\defeq -\frac{\D}{\D x} + u(x)
\end{equation}
on the line, which arises as the Lax operator for the KdV equation
\begin{equation*}
u_t + uu_x + u_{xxx} = 0,
\end{equation*}
as was discovered by C.~S.~Gardner, J.~M. Greene, M.~D.~Kruskal, and R.~M.~Miura in~\cite{GGKM} (see also the seminal work of P.~D.~Lax \cite{Lax}).
Just as in the Toda case, if $u$ in the Marchenko class
\begin{equation}
\int\limits_{\mathbb{R}} (1 + |x|)|u(x)|\,\D x,
\end{equation}
then $H$ has finitely many real simple $L^2$-eigenvalues $-E_j^2$ on $(-\infty,0)$ and an absolutely continuous spectrum $[0,+\infty)$. A difference is that the Jacobi matrix $\mathbf{L}$ has two sides to its continuous spectrum, where the Schr\"odinger operator has only one side. This manifests itself in the fact that Toda solitons can propagate in two directions whereas KdV solitons propagate in only one direction on the line.
\end{remark}

\subsubsection{Scattering data}
Our aim in this subsection is to define ``spectral data'' from the spectral problem
\begin{equation}
    \mathbf{L}\psi = \lambda\psi,
\end{equation}
that is sufficient to reconstruct $\mathbf{L}$ from. First, it is convenient to map the spectrum of $\mathbf{L}$ via the Joukowski transformation:
\begin{equation*}
\lambda = \tfrac{1}{2}\left(z + z^{-1}\right), \phantom{x} z = \lambda - \sqrt{ \lambda^2 - 1 }, \phantom{x} \lambda \in \mathbb{C}, \phantom{x} |z| \leq 1\,.
\end{equation*}
Here the square root $\sqrt{\lambda^{2}-1}$ is defined to be positive for $\lambda>1$ with $\sigma_{\text{ac}}(\mathbf{L})=[-1,1]$ being its only branch cut. This is, of course, a 1-to-2 map. Under this transformation, the absolutely continuous spectrum is mapped to the unit circle, denoted by $\mathbb{T}$, and the eigenvalues $\lambda_j$ are mapped to $\zeta^{\pm 1}_j$, with $\zeta_j \in (-1,0) \cup (0,1)$ via
\begin{equation}\label{E:zeta}
\lambda_j = \tfrac{1}{2}\left(\zeta_j + \zeta_j^{-1} \right),
\end{equation}
for $j=1,2,\dots, N$. In these new coordinates the spectrum of $\mathbf{L}$, which is depicted in Figure~\ref{F:spec_L}, takes the form of the set of points illustrated in Figure~\ref{F:spec_L_z}.
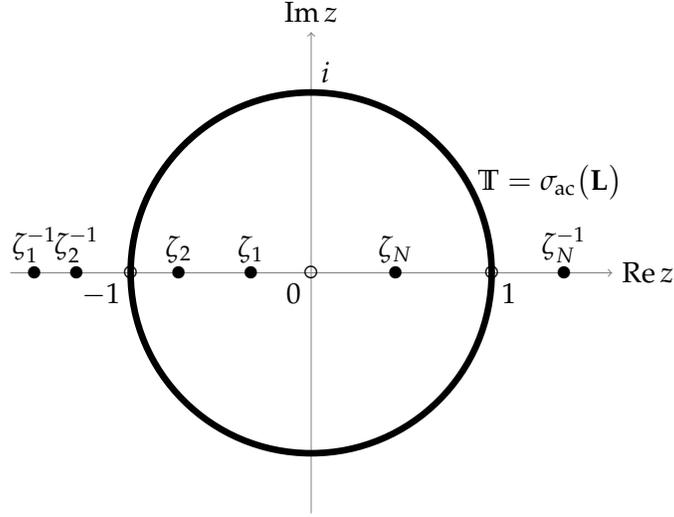
\begin{figure}[h!]
\centering
\begin{tikzpicture}[scale=0.8]
% The Coordinates
\coordinate (z2) at (-2.2,0);
\coordinate (z1) at (-1, 0);
\coordinate (z1i) at (-4.6, 0);
\coordinate (z2i) at (-3.9, 0);
\coordinate (z3) at (1.4, 0);
\coordinate (z3i) at (4.2, 0);
% The axes
\draw[help lines,->] (-5,0) -- (5,0) coordinate (xaxis);
\draw[help lines,->] (0,-4) -- (0,4) coordinate (yaxis);

% The paths
\path[draw,line width=2.5pt,black]
(3,0) arc(0:360:3);

\foreach \Point in {(z1), (z2), (z3), (z1i), (z2i), (z3i)}{
    \node at \Point {\color{black}\large\textbullet};
}
\foreach \Point in {(-3,0), (3,0), (0,0)}{
    \node at \Point {\large$\circ$};
}

% The labels
%\node[right] at (-5, 4) {\small$\lambda_j = \tfrac{1}{2}\Big(\zeta_j + \tfrac{1}{\zeta_j}\Big)$};
%\node[left] at (5, 4) {\small $|\lambda_1| > |\lambda_2| > |\lambda_3|$};
\node[right] at (xaxis) {$\text{Re}\, z$};
\node[above] at (yaxis) {$\text{Im}\, z$};
\node[above] at (z1) {$\zeta_1$};
\node[above] at (z2) {$\zeta_2$};
\node[above] at (z3) {$\zeta_N$};
\node[above] at (z1i) {$\zeta_1^{-1}$};
\node[above] at (z2i) {$\zeta_2^{-1}$};
\node[above] at (z3i) {$\zeta_N^{-1}$};
\node[above right] at (0,3) {$i$};
\node[below left] {$0$};
\node[right] at (30:3) {\color{black}$\mathbb{T} = \sigma_{\text{ac}}\big(\mathbf{L}\big)$};
\node[below left] at (-3, 0 ) {$-1$};
\node[below right] at (3, 0 ) {$1$};
\end{tikzpicture}
\caption{$\sigma(\mathbf{L})$ in the $z$-plane, $z = \lambda - \sqrt{\lambda^{2} -1}$.}
\label{F:spec_L_z}
\end{figure}

It is a standard result (see, for example, Theorem 10.2 in \cite{TeschlBook}) that for any $z\in\mathbb{C}$ with $0<|z|\leq 1$, the linear problem (the 3-term recurrence relation)
\begin{equation*}
\mathbf{L}\psi = \frac{z+z^{-1}}{2}\psi
\label{eq:specprob}
\end{equation*}
has two unique solutions, $\varphi_+(z,\cdot)$ and $\varphi_-(z,\cdot)$, normalized such that
\begin{equation*}
\lim_{n\to\pm \infty} z^{\mp n}\varphi^{\pm}(z;n)=1.
\end{equation*}
These are called \emph{Jost} solutions, named after Swiss theoretical physicist R.~Jost. Moreover, the functions $z\mapsto \varphi_{\pm}(z;n)$ are analytic for $0<|z|<1$ (for each value of the parameter $n$) with continuous boundary values for $z=1$. The Jost solutions have the following asymptotic expansions near $z=0$:
\begin{equation}
\varphi_{\pm}(z;n) = \frac{z^{\pm n}}{A_{\pm}(n)}\left(1 + 2 B_{\pm}(n) z  + O\big(z^2\big)\right),~\text{as}~z\to 0,
\label{eq:Jost-asy}
\end{equation}
where
\begin{equation}
\begin{aligned}
A_+(n,t) &= \prod\limits_{j=n}^{\infty} 2a_j(t)\quad\text{and}\quad B_{+}(n,t)= -\sum_{j=n+1}^{\infty}b_j(t),\\
A_-(n,t) &= \prod\limits_{j=-\infty}^{n-1} 2a_j(t)\quad\text{and}\quad B_{-}(n,t)= -\sum_{-\infty}^{n-1}b_j(t).
\end{aligned}
\label{eq:AB}
\end{equation}

Note that the functions $z\mapsto A_{\pm}(n) z^{\mp n} \varphi_{\pm}(z;n)$ are analytic for $|z|<1$ and they extend continuously to the boundary $|z|=1$. Moreover, $\mu_{\pm}(z;n)\defeq A_{\pm}(n) z^{\mp n} \varphi_{\pm}(z;n)$ satisfies $\mu_{\pm}(0;n)=1$.

Before we proceed with the definition of scattering data, we have an exercise.
\begin{exercise}
Suppose that $\phi(z)$ and $\psi(z)$ are two different solutions of \eqref{eq:specprob} for the same value of $z+z^{-1}$. Show that their Wronskian $\mathcal{W}_n(\phi,\psi)$ is independent of $n$. \emph{Hint: Use Green's identity.}
\end{exercise}
From here on, we drop the subscript $n$ in the Wronskian whenever it is independent of $n$. In addition to the result of the above exercise, $\mathcal{W}(\phi,\psi)=0$ holds if and only if $\phi = c\psi$ for some constant $c\in\mathbb{C}$. Now, $\varphi_{+}(z;\cdot)$ and $\varphi_{+}(z^{-1};\cdot)$ solve \eqref{eq:specprob} for the same value of $z+z^{-1}$. Evaluating their Wronskian as $n\to +\infty$ gives
\begin{equation}
\mathcal{W}\left(\varphi_+(z;\cdot), \varphi_+\big(z^{-1};\cdot\big)\right) =\frac{1}{2} \left(z^n z^{-(n+1)} - z^{n+1}z^{-n}\right)=\frac{z^{-1}-z}{2}.
\label{eq:Wminus}
\end{equation}
It similarly follows that
\begin{equation}
\mathcal{W}\left(\varphi_-(z;\cdot), \varphi_-\big(z^{-1};\cdot\big)\right) = \frac{z-z^{-1}}{2}.
\label{eq:Wplus}
\end{equation}
\eqref{eq:Wminus} and \eqref{eq:Wplus} show that $\left\lbrace \varphi_{-}(z;\cdot), \varphi_{-}\big(z^{-1};\cdot\big)\right\rbrace$ and $\left\lbrace \varphi_{+}(z;\cdot), \varphi_{+}\big(z^{-1};\cdot\big)\right\rbrace$ both form a set of linearly independent solutions of \eqref{eq:specprob} for $|z|=1$ with $z^2 \neq 1$. Therefore, we can write
\begin{equation}
\varphi_+(z;n) = \beta_{-}(z) \varphi_{-}(z;n) + \alpha_+(z)\varphi_{-}(z^{-1};n),\label{eq:4solns1}
\end{equation}
\begin{equation}
\varphi_-(z;n) = \beta_{+}(z) \varphi_{+}(z;n) + \alpha_-(z)\varphi_{+}(z^{-1};n).\label{eq:4solns2}
\end{equation}
Note that $\alpha_{\pm}(z)$ and $\beta_{\pm}(z)$ are independent of $n$. Now, using~\eqref{eq:4solns1} and the asymptotics for the Jost solutions as $n\to-\infty$ we calculate
\begin{equation}
\mathcal{W}\left( \varphi_+ (z;\cdot), \varphi_-(z;\cdot)\right) = \alpha_+(z)\frac{z^{-1} - z}{2} \neq 0
\end{equation}
for $|z|=1$ and $z^2\neq 1$. On the other hand, using~\eqref{eq:4solns2} and the asymptotics as $n\to +\infty$ we obtain
\begin{equation}
\mathcal{W}\left( \varphi_+ (z;\cdot), \varphi_-(z;\cdot)\right) = \alpha_-(z)\frac{z^{-1} - z}{2} \neq 0.
\end{equation}
These imply that
\begin{equation*}
\alpha_{+}(z)\equiv \alpha_{-}(z),
\end{equation*}
hence we rename these quantities as $\Delta(z)\defeq \alpha_{+}(z) = \alpha_{-}(z) $. Then
\begin{equation}
\Delta(z) = \frac{2z}{1-z^2}\mathcal{W}\left( \varphi_+ (z;\cdot), \varphi_-(z;\cdot)\right).
\label{eq:Delta}
\end{equation}
The above formula tells us that $\Delta(z)$ has no zeros on the unit circle, for $z^2\neq 1$. We will come back to zeros of $\Delta(z)$ later. Now we move on to obtain formulae for $\beta_{\pm}$(z). Using \eqref{eq:4solns2} and the asymptotics for the Jost solutions as $n\to +\infty$ yields
\begin{equation}
\mathcal{W}\left( \varphi_- (z;\cdot), \varphi_+\big(z^{-1};\cdot\big)\right) = \beta_{+}(z)\frac{z^{-1} - z}{2} \neq 0,
\label{eq:betaminus}
\end{equation}
for $|z|=1$, $z^2 \neq 1$. Similarly, using \eqref{eq:4solns1} and the asymptotics for the Jost solutions as $n\to -\infty$ gives us
\begin{equation}
\mathcal{W}\left( \varphi_- \big(z^{-1};\cdot\big), \varphi_+(z;\cdot)\right) = \beta_{-}(z)\frac{z^{-1} - z}{2}.
\label{eq:betaplus}
\end{equation}
Hence
\begin{equation}
\begin{aligned}
\beta_{+}(z)&=\frac{2z}{1-z^2}\mathcal{W}\left( \varphi_- (z;\cdot), \varphi_+\big(z^{-1};\cdot\big)\right),\\
\beta_{-}(z)&=\frac{2z}{1-z^2}\mathcal{W}\left( \varphi_- (z^{-1};\cdot), \varphi_+\big(z;\cdot\big)\right).
\end{aligned}
\end{equation}
Moreover, it follows that
\begin{equation*}
\begin{aligned}
\varphi_-(z;\cdot)=& \left(1-\beta_{+}(z)\beta_{-}(z) -\Delta(z)\Delta\big(z^{-1}\big)\right)\varphi_-(z;\cdot)+\\
&+\left( \beta_{+}(z) \Delta(z) + \beta_{+}(z) \beta^+\big(z^{-1}\big)\right)\varphi_-\big(z^{-1};\cdot\big),
\end{aligned}
\end{equation*}
and by linear independence we obtain the following relations:
\begin{equation}
\begin{aligned}
\beta_{-}(z)\beta_{+}(z) + \Delta(z)\Delta\big(z^{-1}\big) &= 1,\\
-\beta_{+}(z) &= \beta_{-}\big(z^{-1}\big).
\end{aligned}
\label{eq:scat1}
\end{equation}
A straightforward calculation shows that $\varphi_+(z;\cdot)$ and its Schwarz reflection $\varphi_+(z^*;\cdot)^*$, where $z^*$ denotes the complex conjugate of $z$, solve \eqref{eq:specprob} for the same value of $z$, and both of these solutions have the same asymptotic behavior:
\begin{equation}
\begin{aligned}
\lim_{n\to+\infty}\varphi_+(z^*;n)^* z^{-n} &= 1,\\
\lim_{n\to+\infty}\varphi_+(z;n) z^{-n} &= 1.
\end{aligned}
\end{equation}
Therefore, by uniqueness, $\varphi_+(z^*;\cdot)^* = \varphi_+(z;\cdot)$ for $0<|z|\leq 1$. It follows by the same argument that $\varphi_-(z^*;\cdot)^* = \varphi_-(z;\cdot)$ for $0<|z|\leq 1$. We can use these two facts to deduce the symmetries of $\Delta(z)$ and $\beta_{\pm}(z)$. Taking complex conjugates of both sides in \eqref{eq:4solns1} for $|z|=1$ gives
\begin{equation*}
\varphi_+(z^*;\cdot)^* = \beta_{-}(z^*)^*\varphi_-(z;\cdot) + \Delta(z^*)^*\varphi_-\big(z^{-1};\cdot\big).
\end{equation*}
Then, by independence, we have
\begin{equation}
\begin{aligned}
\beta_{-}(z^*)^* &= \beta_{-}(z),\\
\Delta(z^*)^* &= \Delta(z).
\end{aligned}
\label{eq:scat2}
\end{equation}
By an analogous argument using \eqref{eq:4solns2} we deduce that
\begin{equation}
\beta_{+}(z^*)^* = \beta_{+}(z).
\label{eq:scat3}
\end{equation}
Then using \eqref{eq:scat1}, \eqref{eq:scat2}, and \eqref{eq:scat3} yields
\begin{equation}
\begin{aligned}
|\Delta(z)|^2 &= 1 + |\beta_{-}(z)|^2,\\
|\Delta(z)|^2 &= 1 + |\beta_{+}(z)|^2.
\end{aligned}
\label{eq:scat4}
\end{equation}
Now we define a new quantity $T(z)$ by
\begin{equation*}
T(z) = \frac{1}{\Delta(z)},
\end{equation*}
and rewrite \eqref{eq:4solns1} and \eqref{eq:4solns2} as
\begin{equation}
T(z) \varphi_+(z;\cdot) = \frac{\beta_{-}(z)}{\Delta(z)} \varphi_-(z;\cdot) +\varphi_-\big(z^{-1};\cdot \big),
\label{eq:tranref1}
\end{equation}
and
\begin{equation}
T(z) \varphi_-(z;\cdot) = \frac{\beta_{+}(z)}{\Delta(z)} \varphi_+(z;\cdot) +\varphi_+\big(z^{-1};\cdot \big)
\label{eq:tranref2}
\end{equation}
respectively. In addition to these, \eqref{eq:scat4} takes the form:
\begin{equation}
\begin{aligned}
1&= |T(z)|^2+ \frac{|\beta_{-}(z)|^2}{|\Delta(z)|^2 },\\
1&= |T(z)|^2+ \frac{|\beta_{+}(z)|^2}{|\Delta(z)|^2 }.
\end{aligned}
\label{eq:scat5}
\end{equation}
For $|z|=1$, \eqref{eq:tranref1} has the following wave reflection interpretation. Since $$\lim_{n\to-\infty}\varphi_-\big(z^{-1};n\big)z^n = 1,$$ we imagine $\varphi_-\big(z^{-1};n \big)$ to be a left incident wave with unit amplitude placed at the left end ($n\to -\infty$) of the lattice. Then we picture that the wave $\varphi_-\big(z^{-1};n \big)$ is incident from the left end of the lattice, gets
reflected by the lattice potential and $\frac{\beta_{-}(z)}{\Delta(z)} \varphi_-(z;n)$ is the reflected wave
moving left. Seen from the right end of the lattice $(n\to+\infty)$, we have the transmitted wave $T(z) \varphi_+(z;n)$. In light of this interpretation, we define
\begin{equation}
R_-(z)\defeq \frac{\beta_{-}(z)}{\Delta(z)}
\label{eq:refco2}
\end{equation}
to be the \emph{left reflection coefficient} and $T(z)$ is called the \emph{transmission coefficient}. The analogous interpretation for \eqref{eq:tranref2} has $\frac{\beta_{+}(z)}{\Delta(z)} \varphi_+(z;n)$ as the reflected wave moving right, and we define
we define
\begin{equation}
R_+(z)\defeq  \frac{\beta_{+}(z)}{\Delta(z)}
\label{eq:refco1}
\end{equation}
to be the \emph{right reflection coefficient}. With these definitions, we have the following scattering relations:
\begin{equation*}
\begin{aligned}
|T(z)|^2 + |R_+(z)|^2 &= 1,\\
|T(z)|^2 + |R_-(z)|^2 &= 1,\\
T(z)R_+(z^*) + T(z^*)R_-(z)&=0,
\end{aligned}
\label{eq:scatrel1}
\end{equation*}
and
\begin{equation}
T(z)^*=T(z^*),\quad R_{\pm}(z)^* = R_{\pm}(z^*).
\end{equation}
\begin{exercise}
Using the Laurent expansion of $\Delta(z)$ at $z=0$, show that
\begin{equation}
T(z)= A\left(1 + z \sum_{n\in\mathbb{Z}}2b_n + O\big(z^2\big) \right),
\end{equation}
where $A= \prod_{n\in\mathbb{Z}}2a_n$.
\label{ex:trans}
\end{exercise}
\begin{proposition}
T(z) has a meromorphic extension inside the unit disk. The only poles of $T(z)$ inside the unit circle are at $z=\zeta_j$, $j=1,2,\dots,N$, for which $\lambda_j\defeq (\zeta_j + \zeta_j^{-1})/2$ is an eigenvalue of $\mathbf{L}$.
\label{pr:trans}
\end{proposition}
\begin{proof}
Since both $\varphi_{\pm}(z;\cdot)$ are both analytic for $0<|z|<1$, $\Delta(z)$ has an analytic extension inside the punctured unit disk, given by the formula \eqref{eq:Delta}, and we define $\Delta(0)=1/A$ using the series expansion in Exercise~\ref{ex:trans}. The poles of $T(z)$ are precisely the zeros of $\Delta(z)$. Suppose that $\Delta(\xi) = 0$ for some fixed $\xi\in\mathbb{C}$ with $0<|\xi|<1$. Then
\begin{equation*}
\mathcal{W}(\varphi_+(\xi;\cdot), \varphi_-(\xi;\cdot))=0,
\end{equation*}
implying that $\varphi_+(\xi;\cdot) = c(\xi) \varphi_-(\xi;\cdot)$ for some constant $c(\xi)\in\mathbb{C}$. Then we have
\begin{equation}
\lim_{n\to +\infty}\varphi_+(\xi;n)\xi^{-n} = 1,
\label{eq:jostdecay1}
\end{equation}
and
\begin{equation}
\lim_{n\to -\infty}\varphi_+(\xi;n)\xi^{n} =\lim_{n\to -\infty}c(\xi)\varphi_-(\xi;n)\xi^{n}= c(\xi).
\label{eq:jostdecay2}
\end{equation}
The equations \eqref{eq:jostdecay1} and \eqref{eq:jostdecay2} imply that $\varphi_{+}(\xi;n) \to 0$ exponentially fast as $|n|\to +\infty$, hence $\varphi_{\pm}(\xi;\cdot)\in\ell^2(\mathbb{Z})$. This means that $T(z)$ has a pole at $z=\xi$ inside the unit disk if and only if $\mathbf{L}$ has an eigenvalue at
\begin{equation*}
\lambda = \frac{1}{2}\left(\xi + \xi^{-1}\right),
\end{equation*}
with the associated eigenfunction $\varphi(\xi,\cdot)\defeq\varphi_+(\xi,\cdot)=c(\xi) \varphi_-(\xi;\cdot)$ in $\ell^2(\mathbb{Z})$.
\end{proof}
Note that positive $\zeta_j$'s correspond to the eigenvalues above $1$, while negative $\zeta_j$'s correspond to the eigenvalues below $-1$. By analyticity, the zeros of $\Delta(z)$ inside the unit disk are isolated and they cannot have an accumulation point inside the unit disk. Note that if the zeros of $\Delta(z)$ do not have an accumulation point on the boundary $|z|=1$, then the zero set of $\Delta(z)$ is finite. The following exercise addresses this.
\begin{exercise}
Prove that $\Delta(z)$ has finitely many isolated simple zeros inside the unit disk by showing that the zero set does not have an accumulation point on the boundary $|z|=1$. \emph{Hint: It might be useful to treat the cases $\mathcal{W}(\varphi_+(z;\cdot), \varphi_- (z;\cdot))=0$ and $\mathcal{W}(\varphi_+(z;\cdot), \varphi_- (z;\cdot))\neq0$ separately.}
\end{exercise}

Now let $c_j\in\mathbb{C}$, for $j=1,2,\dots,N$, be the proportionality constants given by $\varphi_+(\zeta_j;\cdot)= c_j \varphi_-(\zeta_j;\cdot)$. The residues of $T(z)$ at $z=\zeta_j$, $j=1,2,\dots,N$, are given by
\begin{equation}
\underset{\zeta_j}{\res}\,T(z) = -c_j \zeta_j \gamma_{j,+} = - \frac{\zeta_j \gamma_{j,-}}{c_j},
\end{equation}
where
\begin{equation}
\gamma_{j,\pm} \defeq \| \varphi_{\pm}(\zeta_j;\cdot)\|_2^{-2}.
\end{equation}
The constants $\gamma_{j, -}$ are called the \emph{left norming constants} and the constants $\gamma_{j, +}$ are called the \emph{right norming constants}. The sets
\begin{equation*}
\mathcal{S}_{\pm} (\mathbf{L})\defeq \left\lbrace R_{\pm}(z)~\text{for}~|z|=1, \lbrace{\zeta_{j}\rbrace}_{j=1}^{N},\lbrace{\gamma_{j,\pm}\rbrace}_{j=1}^{N}  \right\rbrace
\end{equation*}
are called the \emph{right and the left scattering data} for $\mathbf{L}$, respectively. In fact, a symmetry calculation shows that either of these sets determine the other and the transmission coefficient $T(z)$ via an application of the Poisson-Jensen formula \cite[Lemma 10.7]{TeschlBook}. Therefore it is enough to work with only one of these sets. We set
\begin{equation*}
R(z)\defeq  R_{+}(z),\quad \gamma_j \defeq \gamma_{j,+}
\end{equation*}
and call
\begin{equation*}
\mathcal{S}(\mathbf{L}) \defeq \left\lbrace R(z)~\text{for}~|z|=1, \lbrace{\zeta_{j}\rbrace}_{j=1}^{N},\lbrace{\gamma_j\rbrace}_{j=1}^{N}  \right\rbrace
\end{equation*}
the \emph{scattering data} for $\mathbf{L}$. The mapping $\mathbf{L}\mapsto \mathcal{S}(\mathbf{L})$ is called the \emph{direct scattering transform}.

It is a fundamental fact of scattering theory that Jacobi matrices $\mathbf{L}$ whose coefficients satisfy \eqref{eq:Marchenko} are in bijective correspondence with their scattering data $\mathcal{S}(\mathbf{L})$ (see, for example, \cite[Chapter 11]{TeschlBook}). Thus, there exists an inverse mapping $\mathcal{S}(\mathbf{L})\mapsto \mathbf{L}$, which is called the \emph{inverse scattering transform}. Moreover, as we shall see in the next section, time evolution of the scattering data for $\mathbf{L}(t)$ is governed by simple linear ordinary differential equations when $\mathbf{L}(t)$ evolves according to the Toda lattice equations~\eqref{eq:Lax}. This fact equips us with a method to solve the Cauchy initial value problem for the Toda lattice as depicted in Figure~\ref{F:ist}:
\begin{enumerate}
\item Given initial data $(a^0,b^0)$ satisfying \eqref{eq:Marchenko}, compute the scattering data $\mathcal{S}(\mathbf{L}(0))$ for the Jacobi matrix $\mathbf{L}(0)$ with coefficients $(a^0,b^0)$.
\item Compute the time evolution of the scattering data $\mathcal{S}(\mathbf{L}(t))$ at a desired time $t\in\mathbb{R}$.
\item Compute the inverse scattering transform to reconstruct $\mathbf{L}(t)$ from $\mathcal{S}(\mathbf{L}(t))$, hence obtaining $(a(t), b(t))$.
\end{enumerate}
Following this procedure transfers the inherent difficulty of dealing with a nonlinear equation into the study of the scattering and inverse scattering transform for an operator. The latter can be approached with the specific and powerful methods, mainly with the method of nonlinear steepest descent introduced by P.~Deift and X.~Zhou \cite{DZ}, developed for Riemann-Hilbert problems.
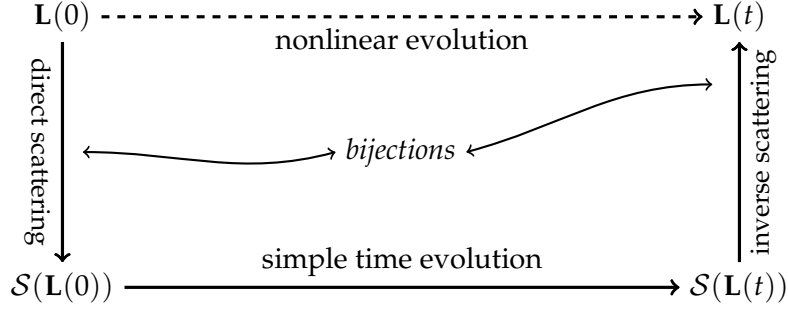
\begin{figure}
\begin{center}
\begin{tikzpicture}[scale=0.9]
\draw
(0,0) node (id) {$\mathbf{L}(0)$}
(10,0) node (soln) {$\mathbf{L}(t)$}
(0,-4) node (scat) {$\mathcal{S}(\mathbf{L}(0))$}
(10,-4) node (scatsoln) {$\mathcal{S}(\mathbf{L}(t))$}
(5,-2) node (bij) {\color{black}\emph{bijections}};
\draw[->,very thick,dashed,decoration={markings,mark= at position 1 with{\arrow[scale=1]{>}}}, postaction={decorate}] (id) -- (soln) node[midway,sloped,below] {nonlinear evolution};

\draw[->,very thick,decoration={markings,mark=at position 1 with {\arrow[scale=1]{>}}}, postaction={decorate}] (scat) -- (scatsoln) node[midway,sloped,above] {\color{black}simple time evolution};

\draw[->, very thick,decoration={markings,mark=at position 1 with {\arrow[scale=1]{>}}}, postaction={decorate}] (id) -- (scat) node[midway,sloped,below] (dst) {\small \color{black}direct scattering};

\draw[->, very thick,decoration={markings,mark=at position 1 with {\arrow[scale=1]{>}}}, postaction={decorate}] (scatsoln) -- (soln) node[midway,sloped,below] (ist) {\small\color{black} inverse scattering};

\path[<->,black,thick] (0.3,-2) edge [out= 0, in= 195] (bij.west);
\path[<->,black,thick] (9.6,-1) edge [out= 180, in= 15] (bij.east);
%\draw(0.6,-1) node (s1) {\small Step (1)};
\end{tikzpicture}
\end{center}
\caption{Inverse scattering transform method for solving the Cauchy initial value problem for the Toda lattice.}
\label{F:ist}
\end{figure}

We end the section with a theorem which establishes that one generically has $R(\pm 1)=-1$ for Jacobi matrices in the Marchenko class $\mathcal{M}$ (see \eqref{eq:Marchenko}).
\begin{theorem}[Theorem 1, \cite{BT}]\label{T:genericity}
The set of doubly infinite Jacobi matrices $\mathbf{J}(a, b)$ (with the sequence $a$ in the off-diagonal entries and the sequence $b$ in the diagonal entries) in $\mathcal{M}$ with the associated reflection coefficient satisfying $R(\pm 1) = -1$ forms an open and dense subset of $\mathcal{M}$ in the topology induced by the norm
\begin{equation*}
\| \mathbf{J}(a,b) \|_{\mathcal{M}} = \sum_{n\in\mathbb{Z}}{(1+|n|)\left(\left|a_n - \tfrac{1}{2}\right | + |b_n| \right)}.
\end{equation*}
\end{theorem}

\subsubsection{Time dependence of scattering data}
Existence of a Lax pair has several important consequences. We begin with the one that concerns the spectrum of the Lax operator.
\begin{proposition}
The Lax equation \eqref{eq:Lax} is an isospectral evolution on Jacobi matrices.
\label{pr:iso}
\end{proposition}
\begin{proof}
Let $\mathbf{Q}(t)$ be the unique solution of the matrix Cauchy initial value problem
\begin{equation}
\frac{\D \mathbf{Q}}{\D t} = \mathbf{P} \mathbf{Q},\quad \mathbf{Q}(0)=\mathbb{I},
\end{equation}
where $\mathbb{I}$ is the identity. Solutions for the initial value problem exist locally and they are unique (see, for example, \cite[Theorem 4.1.5]{Marsden}). Since $\|\mathbf{P}(t)\|$ is uniformly bounded on any compact time interval, the solution is global. The skew-symmetry of $\mathbf{P}$ implies that $\frac{\D \mathbf{Q}^{*}}{\D t} = -\mathbf{Q}^* \mathbf{P}$. Therefore
\begin{equation}
\frac{\D}{\D t}\left( \mathbf{Q}^* \mathbf{Q}\right)=-\mathbf{Q}^* \mathbf{P} \mathbf{Q} + \mathbf{Q}^* \mathbf{P} \mathbf{Q}=0,
\end{equation}
implying that $\mathbf{Q}(t)$ is unitary for all $t\in\mathbb{R}$. Furthermore,
\begin{equation}
\frac{\D}{\D t}\left(\mathbf{Q}^* \mathbf{L} \mathbf{Q} \right)=-\mathbf{Q}^*\mathbf{P}\mathbf{L}\mathbf{Q}+\mathbf{Q}^*(\mathbf{P}\mathbf{L}-\mathbf{L}\mathbf{P})\mathbf{Q}+\mathbf{Q}^*\mathbf{L}\mathbf{P}\mathbf{Q}=0,
\end{equation}
which means $\left(\mathbf{Q}^*\mathbf{L}\mathbf{Q}\right)(t)=\mathbf{L}(0)$, hence
\begin{equation}
\mathbf{L}(t) = \mathbf{Q}(t)\mathbf{L}(0)\mathbf{Q}(t)^{*}.
\end{equation}
We showed that $\mathbf{L}(t)$ is unitarily equivalent to $\mathbf{L}(0)$ for all $t\in\mathbb{R}$. This completes the proof.
\end{proof}

By proving Proposition~\ref{pr:iso} we established the fact that $\sigma(\mathbf{L}(0)) = \sigma(\mathbf{L}(t))$ for all time $t\in\mathbb{R}$ if $\mathbf{L}(t)$ evolves according to \eqref{eq:Lax}.
We have the corollary:
\begin{corollary}
Each eigenvalue $\lambda_j$, $j=1,2,\cdots,N$, of $\mathbf{L}$ is a constant of motion of the Toda lattice equations \eqref{eq:eom-ab}.
\end{corollary}
Another consequence of Proposition~\ref{pr:iso} is
\begin{equation}
\|{\mathbf{L}(t)}\| = \|\mathbf{L}(0)\|,\quad\text{for all} t\geq 0,
\end{equation}
which implies that the solution $(a,b)$ of the Cauchy initial value problem in $\ell^{\infty}(\mathbb{Z})\oplus\ell^{\infty}(\mathbb{Z})$ for the Toda lattice equations \eqref{eq:eom-ab} is global in time since
\begin{equation*}
\|a \|_{\infty} + \|b\|_{\infty} \leq 2\|\mathbf{L}(t) \|=2\|\mathbf{L}(0)\|.
\end{equation*}

We know present the time evolution of the scattering data.

\begin{proposition}[Theorem 13.4, \cite{TeschlBook}]
The time evolutions of the quantities in the scattering data are as follows:
\begin{equation*}
\begin{aligned}
\zeta_j(t) &= \zeta_j,\quad\text{for}~j=1,2,\cdots,N,\\
\gamma_j(t) &= \gamma_j \E^{(\zeta_j - \zeta_j^{-1})t},\quad\text{for}~j=1,2,\cdots,N,\\
R(z;t) &= R(z)\E^{(z-z^{-1})t},\quad\text{for}~|z|=1,
\end{aligned}
\end{equation*}
where $R(z)=R(z;0)$, $\zeta_j = \zeta_j(0)$, and $\gamma_j=\gamma_j(0)$.
\label{pr:scattime}
\end{proposition}

\begin{exercise}
Prove Proposition~\ref{pr:scattime}. Define the \emph{scattering matrix} $\mathbf{S}(z)$ by
\begin{equation*}
\mathbf{S}(z)\defeq \begin{pmatrix}\Delta(z) & -\beta_+(z) \\ \beta_-(z) & \Delta(z^{-1}) \end{pmatrix},
\end{equation*}
and then the ``$\infty \times 2$'' matrices $\boldsymbol{\Psi}(z)$ and $\boldsymbol{\Phi}(z)$ whose $n^\text{th}$ rows are
\begin{equation}
\begin{aligned}
\boldsymbol{\Psi}(z;n) &= \begin{pmatrix} \varphi_+(z;n) & \varphi_+ (z^{-1};n)\end{pmatrix},\\
\boldsymbol{\Phi}(z;n) &= \begin{pmatrix} \varphi_-(z^{-1};n) & \varphi_- (z;n)\end{pmatrix},
\end{aligned}
\end{equation}
respectively, so that we have $\boldsymbol{\Psi}(z;n)= \boldsymbol{\Phi}(z;n) \mathbf{S}(z)$ for all $n\in\mathbb{Z}$. Differentiate both sides of \eqref{eq:specprob} and obtain the time evolution $\mathbf{S}(z;t)$ (see the Appendix in~\cite{BN}).
\end{exercise}

\begin{exercise}[J.~Moser,~\cite{Moser}]
Consider the finite version of the Toda lattice as described in \eqref{eq:finiteToda}. Show that the time evolution of the first component $\psi^{[k]}_1$ of the normalized eigenvector $\psi^{[k]}$ associated to the eigenvalue $\lambda_k$ is given explicitly by:
\begin{equation}
\displaystyle\psi^{[k]}_1(t)^2 = \frac{\E^{2\lambda_k t}\psi^{[k]}_1(0)^2}{\sum\limits_{j=1}^N \E^{2\lambda_j t}\psi^{[j]}_1(0)^2}.
\label{eq:pdm}
\end{equation}
\emph{See also \cite{PDM} for a detailed proof of \eqref{eq:pdm}.}
\label{ex:pdm}
\end{exercise}

We close this section with a calculation which is due to J.~Moser~\cite{Moser} and dates back to 1975.
\begin{exercise}[J.~Moser,~\cite{Moser}]
For the finite Toda lattice~\eqref{eq:finiteToda} considered in Remark~\ref{r:finiteToda} show that the off-diagonal elements $L_{j+1, j}(t)=L_{j, j+1}(t)$ of $\mathbf{L}(t)$ converge to zero as $t$ tends to infinity.
\end{exercise}

\subsection{Inverse problem}\label{sec:inverse}
The problem of reconstructing $\mathbf{L}$ from its scattering data is often called the inverse problem (and the mapping that achieves this is called the inverse scattering transform). This problem can be formulated as a Riemann-Hilbert factorization problem. Before giving a brief description of a Riemann-Hilbert problem, we need to introduce some notation. For an oriented contour $\Sigma$ in the complex plane, its $+$ side is to the left of $\Sigma$ as it is traversed in the direction of orientation, and its $-$ side is to the right. Given this convention, we denote by $\phi^{+}$ and $\phi^{-}$ the nontangential limits of a function $\phi$ on $\Sigma$ from $+$ side and $-$ side of $\Sigma$, respectively. Loosely speaking, a Riemann-Hilbert problem is the problem of finding a \emph{sectionally analytic} (or \emph{sectionally meromorphic}) function $\phi$ that is discontinuous across an oriented contour $\Sigma\subset\mathbb{C}$ with the jump condition
\begin{equation*}
\phi^{+}(z) =  \phi^{-}(z)G(z) + F(z),\quad \text{for}~z\in\Sigma.
\end{equation*}
In case the sought after function $\phi$ is sectionally meromorphic, there are prescribed residue conditions along with the jump condition. To assure that the problem at hand is uniquely solvable, one specifies normalization of $\phi$ at a point $\alpha\in\mathbb{C}\cup\{\infty\}$
and, perhaps, a symmetry condition. A detailed overview and discussion on Riemann-Hilbert problems is beyond the scope of these lectures. The interested reader is encouraged to see \cite{AbFokas,Clancey, DeiftBook,TomBook}. We now proceed with presenting the Riemann-Hilbert formulation of the inverse scattering transform for the Toda lattice.

Based on the domains of analyticity of the Jost solutions $z\mapsto\varphi_{\pm}(z;\cdot,\cdot) $, we define the following row-vector valued meromorphic function in the complex plane:
\begin{equation}
m(z;n,t)\defeq\begin{cases}
\begin{pmatrix} T(z) \varphi_-(z;n,t)z^n & \varphi_+(z;n,t)z^{-n}\end{pmatrix},&\quad |z|<1,\vspace{0.5em}\\
\begin{pmatrix} \varphi_+(z^{-1};n,t)z^n & T(z^{-1})\varphi_-(z^{-1};n,t)z^{-n}\end{pmatrix},&\quad |z|>1,\\
\end{cases}
\label{eq:m}
\end{equation}
We first find what jump condition $m(z;n,t)$ satisfies on the unit circle $\mathbb{T}$. Using the notation given above for the nontangential limits, define
\begin{equation*}
m^{\pm}(z_0;n,t)\defeq \lim_{\substack{z\to z_0 \\ |z|^{\pm}< 1}} m(z),~\quad\text{for}~|z_0|=1.
\end{equation*}
We note here that these $\pm$ superscripts should not be confused with the $\pm$ signs used in subscripts for labeling the asymptotic behavior of Jost solutions $\varphi_{\pm}$.

From the relations \eqref{eq:4solns1} and \eqref{eq:4solns2} we see that the jump condition satisfied by $m$ across the unit circle $\mathbb{T}$ is
\begin{equation}
m^{+}(z;n,t) = m^{-}(z;n,t)\begin{pmatrix}
1 - |R(z)|^2 & -R(z)^*\E^{-\theta(z;n,t)} \\
R(z)\E^{\theta(z;n,t)} & 1
\end{pmatrix},\quad\text{for}~z\in\mathbb{T},
\end{equation}
where
\begin{equation*}
\theta(z;n,t)=t\big(z-z^{-1}\big)+2n\log{z}.
\end{equation*}
We have arrived at the following fact~\cite[Theorem 3.3]{KruTes}. $m(z;n,t)$ is a solution of the following vector Riemann-Hilbert problem:

\begin{rhp}
Find a row vector-valued function $m(\cdot;n,t)\colon \mathbb{C}\setminus\mathbb{T}\to\mathbb{C}^{1\times 2}$ which satisfies the following conditions:
\begin{itemize}
\item \emph{Analyticity condition}: $m(z;n,t)$ is sectionally meromorphic away from the unit circle $\mathbb{T}$ with simple poles at $\zeta_j^{\pm 1}$, $j=1,2,\dots,N$.
\item \emph{Jump condition}:
\begin{equation*}
m^{+}(z;,n,t) = m^{-}(z;n,t)\begin{pmatrix}
1 - |R(z)|^2 & -R(z)^*\E^{-\theta(z;n,t)} \\
R(z)\E^{\theta(z;n,t)} & 1
\end{pmatrix},\quad\text{for}~z\in\mathbb{T},
\end{equation*}
\item \emph{Residue conditions}:
\begin{equation*}
\begin{aligned}
\underset{\zeta_j}{\res}\,m(z;n,t) &=\lim_{z\to\zeta_j} m(z)\begin{pmatrix} 0 & 0\\
-\zeta_j \gamma_j \E^{\theta(\zeta_j;n,t)} & 0
\end{pmatrix},\\
\underset{\zeta_j^{-1}}{\res}\,m(z;n,t) &=\lim_{z\to\zeta_j} m(z)\begin{pmatrix} 0 & \zeta_j^{-1} \gamma_j \E^{\theta(\zeta_j;n,t)}\\
0 & 0
\end{pmatrix},
\end{aligned}
\end{equation*}
\item \emph{Symmetry condition}:
\begin{equation*}
m\big(z^{-1}\big) = m(z)\begin{pmatrix} 0 & 1\\1 & 0
\end{pmatrix},
\end{equation*}
\item \emph{Normalization condition}:
\begin{equation*}
\lim_{z\to\infty} m(z;n,t) \eqdef m(\infty;n,t) = \begin{pmatrix} m_1 & m_2 \end{pmatrix},~m_1 \cdot m_2=1,~m_1>0.
\end{equation*}
\end{itemize}
\label{rhp:m}
\end{rhp}

The solution of this Riemann-Hilbert problem is unique (see \cite{KT_sol}) and the symmetry condition is essential for uniqueness in presence of poles. Note that once $m(z;n,t)$ is obtained, we can extract the solution $\big(a_n(t), b_n(t)\big)$ of the Cauchy initial value problem for the Toda lattice equations from the asymtptotic expansion of $m(z;n,t)$ as $z\to \infty$:
\begin{equation}
m(z;n,t)=\begin{pmatrix}
\frac{1}{A_+(n,t)}\left(1 + 2B_{+}(n,t)z\right)&A_+(n,t)\left(1 - 2B_{+}(n-1,t)z\right)
\end{pmatrix}
+O(z^{-2}).
\label{eq:m-asy}
\end{equation}
We recover $(a,b)$ from the formulae
\begin{equation*}
a_n(t) = \frac{A_+(n+1,t)}{A_+(n,t)}\quad\text{and}\quad b_n(t)= B_{+}(n) - B_{+}(n-1).
\label{eq:recover}
\end{equation*}

As can be seen, solving the Cauchy initial value problem for sufficiently decaying initial data for the Toda lattice equations is equivalent to solving \rhref{rhp:m} which depends parametrically on the independent variables $(n,t)$ of the nonlinear evolution equations \eqref{eq:eom-ab}. However, obtaining an explicit formula for the solution $m(z;n,t)$ for any given values of $(n,t)$ is difficult, if possible. One can, on the other hand, obtain rigorous long-time asymptotics (as $t\to\infty$) for the solutions of the Cauchy initial value problem through the analysis of Riemann-Hilbert formulation of the inverse scattering transform as $t\to\infty$. Note that as $t\to\infty$, the jump matrix in \rhref{rhp:m} becomes highly oscillatory since the coefficient of the $t$-term in $\theta(z;n,t)$ is purely imaginary for $|z|=1$. One can approach this problem by the method of nonlinear stationary phase/steepest descent (see \cite{DZ}). Loosely speaking, this method involves finding the stationary phase points of the exponential terms in the jump matrix, and then deforming the jump contours so that the deformed contours pass locally from the directions of steepest descent of these exponential terms at the stationary phase points and that the new jump matrices on those deformed contours tend to the identity matrix as $t\to\infty$ exponentially fast away from the stationary phase points.

This method has yielded a large number of rigorous asymptotic results for various completely integrable partial differential equations (see, for example, \cite{DVZ} and \cite{DZP2}, among many others) as well as for the Toda lattice (see \cite{KruTes, KT_sol}). Moreover, recent advances in numerical solution of the Riemann-Hilbert problems~\cite{OlverP2,OlverRHP,OTRHP} has led to numerical implementations of the inverse scattering transform method for integrable PDEs \cite{TONLS,TODKdV} and for the Toda lattice \cite{BT}. Therefore, it is possible to numerically compute and plot solutions accurately in arbitrarily long time scales without using any time-stepping methods.

The deformations that are employed in the process of analyzing the solutions of the Riemann-Hilbert problem in the long time regime are determined by the asymptotic region that $(n,t)$ lies in as $t\to\infty$ (see Figure~\ref{F:regions} for asymptotic regions for the Toda lattice).
\begin{figure}
	\centering
\includegraphics[scale=0.4]{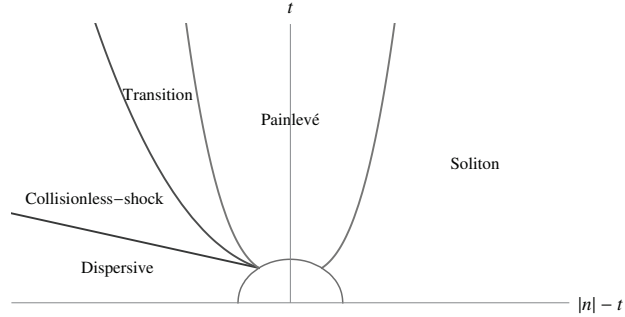}
\caption{Asymptotic regions for the Toda lattice, figure taken from \cite{BT}.}
\label{F:regions}
\end{figure}
The Toda lattice has the following asymptotic regions~\cite{BT}: for constants $k_j  > 0$:
\begin{enumerate}
\item[1.] \emph{The dispersive region.} This region is defined for $|n| \leq k_1 t$, with $0<k_1<1$. Asymptotics in this region were obtained in \cite{KruTes}.

\item[2.] \emph{The collisionless shock region.} This region is defined by the relation $  c_1 t\leq |n|  \leq t - c_2 t^{1/3}\left(\log{t} \right)^{2/3} $. For the behavior of the solutions obtained via the numerical inverse scattering transform see \cite{BT}.

\item[3.] \emph{The transition region.} This region is also not present in the literature for the Toda lattice. {The region is} defined by the relation $ t - c_2 t^{1/3}\left(\log{t} \right)^{2/3}   \leq |n| \leq t - c_3 t^{1/3}$~\cite{BT}. An analogue of this region was first introduced for KdV in \cite{TODKdV}.

\item[4.] \emph{The Painlev\'e region.} This region is defined for $t - c_3 t^{1/3} \leq |n| \leq t + c_3 t^{1/3}$. Asymptotics in this region were obtained in \cite{Kam} in absence of solitons and under the additional assumption that $|R(z)|<1$.

\item[5.] \emph{The soliton region.} This region is defined for $|n| > t + c_3 t^{1/3}$. Let $v_k >1$ denote the velocity of the $k^{\text{th}}$ soliton and choose $\nu>0$ so that the intervals $(v_k - \nu, v_k +\nu)$, $k=1,2,\dots, N$, are disjoint. If $|n/t - v_k| < \nu$, the asymptotics in this region were obtained in \cite{KT_sol} and \cite{KruTes}.
\end{enumerate}

For the deformations used in each region for the Toda lattice we refer the reader to~\cite{KruTes} and \cite{BT}.
\begin{remark}
The collisionless shock region and the transition region appear as $t\to\infty$ only if the reflection coefficient $R(z)$ corresponding to the initial data attains the value $-1$ at the edges of the continuous spectrum $z=\pm 1$. Theorem~\ref{T:genericity} tells us that this is generically the case: for an open dense subset of initial data inside the Marchenko class \eqref{eq:Marchenko}, the long time behavior of the solution differs from that in the Painlev\'e region. If $|R(\pm 1)|<1$ the collisionless shock and the transition regions are absent.
\end{remark}

The following rigorous results were obtained in \cite{KruTes, KT_sol}. The solution corresponding to sufficiently decaying (as in \eqref{eq:Marchenko}) splits into a sum of $N$ independent solitons ($N$ being the number of eigenvalues of the Lax operator) and radiation. This result shows that in the region $\frac{n}{t} > 1$, the solution is asymptotically given by an $N_-$-soliton solution, where $N_-$ is the number of $\zeta_j\in(-1,0)$. Similarly, in the region $\frac{n}{t}<-1$ the solution is asymptotically given by an $N_+$-soliton solution, where $N_+$ is the number of $\zeta_j\in(0,1)$. Each soliton is formed by a distinct pair of eigenvalue and the associated norming constant, that is, by a $(\zeta_j, \gamma_j)$, for some $j\in\lbrace 1,2,\dots \mathbb{N} \rbrace$. In the dispersive region $\frac{n}{t}<1$, the solution (radiation) decays to the background, i.e.\ $a_n(t)\to\frac{1}{2}$ and $b_n(t)\to 0$ in the sup-norm as $t\to\infty$.

We close this subsection with a special case that leads to an explicit formula for the solution. If the potential Jacobi matrix is reflectionless, that is if $R(z)\equiv 0$ on the entire unit circle, then \rhref{rhp:m} can be solved explicitly. Below we present the solution $m(z;n,t)$ in the case where there is only one eigenvalue $\zeta \in (-1,0)\cup(0,1)$ and the associated norming constant $\gamma > 0$ in the scattering data. Such scattering data results in a pure 1-soliton solution of the Toda lattice.
\begin{proposition}[Lemma 2.6, \cite{KT_sol}]
If $\mathcal{S}(\mathbf{L})=\lbrace R(z)\equiv 0~\text{for}~|z|=1, \zeta, \gamma \rbrace$, then the unique solution of \rhref{rhp:m} is given by
\begin{equation}
m(z;n,t) = \begin{pmatrix} f(z) & f(1/z)\end{pmatrix},
\end{equation}
where
\begin{equation}
f(z) = \frac{1}{\sqrt{1-\zeta^2 + \gamma \E^{\theta(\zeta;n,t)} }\sqrt{1-\zeta^2 + \zeta^2\gamma \E^{\theta(\zeta;n,t)} }}\left(1 - \zeta^2 + \zeta^2 \gamma \E^{\theta(\zeta;n,t)}\frac{z-\zeta^{-1}}{z-\zeta} \right).
\end{equation}
\label{pr:sol}
\end{proposition}
\begin{proof}
The symmetry condition forces the solution to be of the form
\begin{equation*}
m(z;n,t) = \begin{pmatrix}f(z;n,t) & f(1/z;n,t) \end{pmatrix},
\end{equation*}
where $f(z;n,t)$ is meromorphic in $\mathbb{C}\cup {\infty}$ with simple pole at $z=\zeta$. Therefore $f$ must be of the form
\begin{equation}
f(z) = \frac{1}{G(n,t)}\left(1 + 2 \frac{H(n,t)}{z-\zeta} \right),
\end{equation}
where the unknown constants (in $z$) $G$ and $H$ are uniquely determined by the residue condition
\begin{equation*}
\underset{\zeta}{\res}\,m(z;n,t) = -f\big(\zeta^{-1}\big)\zeta \gamma \E^{\theta(\zeta;n,t)}
\end{equation*}
and the normalization $f(0)f(\infty) =1$, $f(0)>0$. 
\end{proof}

From Proposition~\ref{pr:sol} we obtain
\begin{equation}
A_{+}(n,t) = \sqrt{ \frac{ 1-\zeta^2 +\gamma \E^{\theta(\zeta;n,t)}  }{  1-\zeta^2 +\zeta^2\gamma \E^{\theta(\zeta;n,t)} }   }\quad\text{and}\quad
B_{+}(n,t)=\frac{ \zeta^2\gamma \E^{\theta(\zeta;n,t)} \zeta\big(\zeta^2 - 1 \big)  }{ 2 \left(1-\zeta^2 +\zeta^2\gamma \E^{\theta(\zeta;n,t)}\right)    },
\end{equation}
from which the solution $(a_n(t), b_n(t))$ to the Cauchy initial value problem can be extracted using \eqref{eq:recover}.

The solutions $(a,b)$ that correspond to reflectionless Jacobi matrices (Lax operators) $\mathbf{L}$ are precisely the $N$-soliton solutions of the Toda lattice given in \eqref{eq:Nsol}, where $N$ is the number of eigenvalues $\lambda_j = \tfrac{1}{2}\left(\zeta_j + \tfrac{1}{\zeta_j}\right)$ of $\mathbf{L}$, i.e. the number of pole pairs $(\zeta_j, \zeta_j^{-1})$ in the Riemann-Hilbert formulation of the inverse scattering transform.

The inverse scattering transform method for solution of the initial value problem also applies to the finite Toda lattice and the periodic Toda lattice. While it is not possible to give a full exposure to these cases here, we provide a summary and references below for the interested reader.

\subsubsection{The solution of the finite Toda lattice}
For the finite $N\times N$ version\footnote{Finite Toda lattice is sometimes called the open Toda lattice.} of the Toda lattice given in \eqref{eq:finiteToda}, the Lax operator becomes
\begin{equation}
\mathbf{L}=
\begin{pmatrix}
 b_0 & a_0 & 		   \\
 a_0 & b_1 & a_1 &            \\
        & a_1 & b_2 & a_2  & \\
 &   & \ddots   & \ddots & \ddots & &       	 \\
		  & &		  & a_{N-3}     &b_{N-2} & a_{N-2}         \\
		   & &		  &      &a_{N-2} & b_{N-1}         \\
\end{pmatrix}
\end{equation}
and the scattering data consists of $N$ real simple eigenvalues $\lambda_j$, $j=1,2,\dots,N$ of the Jacobi matrix $\mathbf{L}(t)$ and 
\begin{equation*}
w_j(t)\defeq \left(\psi_1^{[j]}(t)\right)^2,
\end{equation*} 
where $\psi_1^{[j]}$ is the first component of the $j^{\text{th}}$ normalized eigenvector (associated with $\lambda_j$). As mentioned in Exercise~\ref{ex:pdm}, $w_j(t)$ has simple time evolution:
\begin{equation}
w_j(t) = \frac{\E^{2\lambda_j t}w_j(0)}{\sum\limits_{k=1}^{N}\E^{2\lambda_k t}w_k(0)},~\quad j=1,2,\dots,N,
\end{equation}
when $\mathbf{L}(t)$ evolves under the finite Toda dynamics (observed by Moser~\cite{Moser}) and $\lambda_j$ remain constant in time~\cite{Flaschka}:
\begin{equation*}
\lambda_j(t) = \lambda_j(0),\quad j=1,2,\dots,N.
\end{equation*}
It is well-known that $\mathbf{L}(t)$ can be reconstructed from the data $\lbrace \lambda_j, w_j(t)\rbrace_{j=1}^{N}$ via Hankel determinants (see \cite{Szego}, \cite{DeiftKen}, or \cite{DeiftBook} and the references therein.) Moreover, the inverse scattering transform method to solve the initial value problem for the finite Toda equations becomes the Gramm-Schmidt orthogonalization process:
\begin{enumerate}
\item Given initial data $\mathbf{L}(0)$ and some time $t>0$ at which the solution $\mathbf{L}(t)$ is desired, compute the $QR$-factorization:
\begin{equation*}
\E^{t \mathbf{L}(0)}= \mathbf{Q}(t) \mathbf{R}(t).
\end{equation*} 
\item The solution of the Toda lattice is then obtained via
\begin{equation*}
\mathbf{L}(t) = \mathbf{Q}(t)^{\top}\mathbf{L}(0)\mathbf{Q}(t),
\end{equation*}
where $\mathbf{Q}(t)^\top$ is the transpose of the orthogonal matrix $\mathbf{Q}(t)$. For details, we refer the reader to \cite{Moser} and \cite{Symes1,Symes2}.
\end{enumerate}

\subsubsection{The periodic and quasi-periodic solutions of the Toda lattice}
The inverse scattering transform methods in the context of nonlinear evolution equations are not limited to spatially decaying solutions. Beginning with the seminal works of Novikov~\cite{Novikov}, Lax~\cite{Lax2}, Dubrovin~\cite{Dubrovin1,Dubrovin2}, Its and Matveev~\cite{Its-Matveev1, Its-Matveev2}, and McKean and van Moerbeke~\cite{McKean-vanMoer}, the inverse scattering transform methods were extended to construct spatially periodic and quasi-periodic solutions of the KdV equation in the late 1970s. These methods make extensive use of the algebro-geometric theory of Riemann surfaces. By the early 1980s analogous methods were developed for other integrable nonlinear evolution equations such as the nonlinear Schr\"odinger equation. Explicit formulae for periodic and quasi-periodic solutions, which give soliton solutions as certain limiting cases, became available in terms of the Riemann theta function (see \cite{Dubrovin-theta} for a survey article on theta functions and nonlinear equations). The fundamental object of the algebro-geometric approach to integrate nonlinear evolution equations is the Baker-Akheizer function (a function with certain analyticity properties on a Riemann surface)~\cite{Baker,Akhiezer1,Akhiezer2}. Such a connection between nonlinear evolution equations and algebro-geometric methods led to rapid developments both in algebraic geometry and in the theory of integrable PDEs. While it is not possible to list here the vast literature, the interested reader is encouraged to see the recent survey article of Matveev~\cite{Matveev-gap} on finite gap theory, a not so recent survey article by Krichever~\cite{Krichever-survey} and his previous work \cite{Krichever-previous, Krichever-elliptic}.

Integrability of the periodic Toda lattice consisting of $N$ particles, with
\begin{equation}
a_{n+N}=a_n,\quad b_{n+N} = b_n
\end{equation}
was first established by H\'{e}non~\cite{Henon} in 1974 (see also \cite{Flaschka-McLaughlin}, where Flaschka and McLaughlin obtained the action-angle variables, and \cite{vanMoerbeke-periodic}). In 1976 Dubrovin, Matveev, and Novikov~\cite{DMN} and Date and Tanaka~\cite{Date-Tanaka} simultaneously gave the explicit formula for the diagonal elements $b_n(t)$ of the solution of the periodic Toda lattice in terms of Riemann theta functions. 3 years after this development Krichever gave the explicit formulae (again in terms of theta functions) for the off-diagonal elements $a_n(t)$~\cite{KricheverTodaPaper}. In his work, Krichever provides the solution for both quasi-periodic and the periodic Toda lattices. In these cases the spectrum of the Lax operator consists of finitely many bands on the real line and the scattering data is generally referred as ``algebro-geometric'' data obtained from the resulting Riemann surface. For a detailed treatment on algebro-geometric solution methods for the periodic and quasi-periodic Toda lattice (and the entire Toda hierarchy), we refer the reader to \cite{GHMTv2} and the references therein.

\section{Darboux \& B\"acklund transformations. Derivation of discrete systems}
\label{sec:3}
B\"acklund and Darboux (or Darboux type) transformations originate from differential geometry of surfaces in the nineteenth century, and they constitute an important and very well studied connection with the modern soliton theory and the theory of integrable systems. In the modern theory of integrable systems, these transformations are used to generate solutions of partial differential equations, starting from known solutions, even trivial ones. In fact, Darboux transformations apply to systems of linear equations, while B\"acklund transformations are generally related to systems of nonlinear equations.

For further information on B\"acklund and Darboux transformations we indicatively refer to \cite{GuChaohao, Matveev-Salle, Rog-Schief} (and the references therein).

%%%%%%%%%%%%%%%%%%%%%%%%%%%%%%%%%%%%%%%%%%%%%%%%%%%%%%
%%%              Darboux transformations           %%%
%%%%%%%%%%%%%%%%%%%%%%%%%%%%%%%%%%%%%%%%%%%%%%%%%%%%%%

\subsection{Darboux transformations}
In 1882 Jean Gaston Darboux \cite{Darboux} presented the so-called ``Darboux theorem'' which states that a Sturm-Liouville problem is covariant with respect to a linear transformation. In the recent literature, this is called the \textit{Darboux transformation} \cite{Matveev-Salle, Rog-Schief}. The first book devoted to the relation between Darboux transformations and the soliton theory is that of Matveev and Salle \cite{Matveev-Salle}.

\subsubsection{Darboux's theorem}
Darboux's original result is related to the so-called \textit{one-dimensional, time-independent Schr\"odinger} equation, namely
\begin{equation}\label{Sturm-Liouville}
y''+(\lambda-u)y=0, \quad u=u(x),
\end{equation}
which can be found in the literature as a \textit{Sturm-Liouville problem} of finding eigenvalues and eigenfunctions. Moreover, we refer to $u$ as a \textit{potential function}, or just \textit{potential}.

In particular we have the following.

\begin{theorem}(Darboux)\label{DarbouxTheorem}
Let $y_1=y_1(x)$ be a particular integral of the Sturm-Liouville problem (\ref{Sturm-Liouville}), for the value of the spectral parameter $\lambda=\lambda_1$. Consider also the following (Darboux) transformation
\begin{equation}\label{DT}
y\mapsto y[1]:=\left(\frac{d}{dx}-l_1\right)y,
\end{equation}
of an arbitrary solution, $y$, of (\ref{Sturm-Liouville}),  where $l_1= l_1(y_1)=y_{1,x}y_1^{-1}$ is the logarithmic derivative of $y_1$. Then, $y[1]$ obeys the following equation
\begin{subequations}\label{transEq}
\begin{align}
y''[1]+&(\lambda-u[1])y[1]=0,\label{transEq-1}
\intertext{where $u[1]$ is given by}
&u[1]=u-2l_1'.\label{transEq-2}
\end{align}
\end{subequations}
\end{theorem}

Darboux's theorem states that function $y[1]$ given in (\ref{DT}) obeys a Sturm-Liouville problem of the same structure with (\ref{Sturm-Liouville}), namely the same equation (\ref{Sturm-Liouville}) but with an updated potential $u[1]$. In other words, equation (\ref{Sturm-Liouville}) is covariant with respect to the Darboux transformation, $y\mapsto y[1]$, $u\mapsto u[1]$.

\subsubsection{Darboux transformation for the KdV equation}
The significance of the Darboux theorem lies in the fact that transformation (\ref{DT}) maps solutions of a Sturm-Liouville equation (\ref{Sturm-Liouville}) to other solutions of the same equation, which allows us to construct hierarchies of such solutions. At the same time, the theorem provides us with a relation between the ``old'' and the ``new'' potential. In fact, if the potential $u$ obeys a nonlinear ODE (or more importantly a nonlinear PDE \footnote{Potential $u$ may depend on a temporal parameter $t$, namely $u=u(x,t)$.}), then relation (\ref{transEq}) may allow us to construct new non-trivial solutions starting from trivial ones, such as the zero solution.

\begin{example}\normalfont
Consider the Sturm-Liouville equation (\ref{Sturm-Liouville}) in the case where the potential, $u$, satisfies the KdV equation. Therefore, both the eigenfunction $y$ and the potential $u$ depend on $t$, which slips into their expressions as a parameter.

In this case, equation (\ref{Sturm-Liouville}) is nothing else but the spatial part of the Lax pair for the KdV equation, namely:
\begin{equation}\label{x-LaxKdV}
\textbf{L} y=\lambda y \quad \text{or}\quad y_{xx}+(\lambda-u(x,t))y=0.
\end{equation}

Now, according to theorem \ref{DarbouxTheorem}, for a known solution of the KdV equation, say $u$, we can solve (\ref{Sturm-Liouville}) to obtain $y=y(x,t;\lambda)$. Evaluating at $\lambda=\lambda_1$, we get $y_1(x,t)=y(x,t;\;\lambda_1)$ and thus, using equation (\ref{transEq-2}), a new potential $u[1]$. Therefore, we simultaneously obtain new solutions, $(y[1],u[1])$, for both the linear equation (\ref{x-LaxKdV}) and the KdV equation\footnote{Potential $u[1]$ is a solution of the KdV equation, since it can be readily shown that the pair $(y[1],u[1])$ also satisfies the temporal part of the Lax pair for KdV.}, which are given by
\begin{subequations}\label{eqy1u1}
\begin{align}
y[1]&=(\partial_x-l_1)y,\label{eqy1u1-1}\\
u[1]&=u-2l_{1,x},        \label{eqy1u1-2}
\end{align}
\end{subequations}
respectively.

Now, applying the Darboux transformation once more, we can construct a second solution of the KdV equation in a fully algebraic manner. Specifically, first we consider the solution $y_2[1]$, which is $y[1]$ evaluated at $\lambda=\lambda_2$, namely
\begin{equation}
y_2[1]=(\partial_x-l_1)y_2.
\end{equation}
where $y_2=y(x,t;\lambda_2)$
Then, we obtain a second pair of solutions, $(y[2],u[2])$, for (\ref{x-LaxKdV}) and the KdV equation, given by
\begin{subequations}
\begin{align}
y[2]&=(\partial_x-l_2)y[1]\stackrel{(\ref{eqy1u1-1})}{=}(\partial_x-l_2)(\partial_x-l_1)y,\\
u[2]&=u[1]-2l_{2,x}\stackrel{(\ref{eqy1u1-2})}{=}u-2(l_{1,x}+l_{2,x}).
\end{align}
\end{subequations}

This procedure can be repeated successively, in order to construct hierarchies of solutions for the KdV equation, namely
\begin{equation}
(y[1],u[1])\rightarrow(y[2],u[2])\rightarrow\cdots\rightarrow(y[n],u[n])\rightarrow\cdots,
\end{equation}
where $(y[n],u[n])$ are given by
\begin{equation}\label{eqynun}
y[n]=\left(\prod_{k=1}^{\stackrel{\curvearrowleft}{n}}(\partial_x-l_k)\right)y,\quad u[n]=u-2\sum_{k=1}^n(l_{k,x}),
\end{equation}
where ``$\curvearrowleft$" indicates that the terms of the above ``product" are arranged from the right to the left.
\end{example}

In these notes, we understand Darboux transformations as gauge-like transformations which depend on a spectral parameter. In fact, as we shall see in the next chapter, their dependence on the spectral parameter is essential to construct discrete integrable systems.

%%%%%%%%%%%%%%%%%%%%%%%%%%%%%%%%%%%%%%%%%%%%%%%%%%%%%%
%%%           B\"acklund transformations           %%%
%%%%%%%%%%%%%%%%%%%%%%%%%%%%%%%%%%%%%%%%%%%%%%%%%%%%%%

\subsection{B\"acklund transformations}
As mentioned earlier, B\"acklund transformations  originate in differential geometry in the 1880s and, in particular, they arose as certain transformations between surfaces.

In the theory of integrable systems, they are seen as relations between solutions of the same PDE (auto-BT) or as relations between solutions of two different PDEs (hetero-BT). Regarding the nonlinear equations which have Lax representation, Darboux transformations apply to the associated linear problem (Lax pair), while B\"acklund transformations are related to the nonlinear equation itself. Therefore, unlike DTs, BTs do not depend on the spectral parameter which appears in the definition of the Lax pair. Yet, both DTs  and BTs serve the same purpose; they are used to construct non-trivial solutions starting from trivial ones.

\begin{definition} (BT-loose Def.) Consider the following partial differential equations for $u$ and $v$:
\begin{subequations}
\begin{align}
F(u,u_x,u_t,u_{xx},u_{xt},\ldots)&=0,\\
G(v,v_x,v_t,v_{xx},v_{xt},\ldots)&=0.
\end{align}
\end{subequations}
Consider also the following pair of relations
\begin{equation}
\mathcal{B}_i(u,u_x,u_t,\ldots,v,v_x,v_t,\ldots)=0,
\end{equation}
between $u$, $v$ and their derivatives. If $\mathcal{B}_i=0$ is integrable for $v$, $\mod \langle F=0 \rangle$, and the resulting $v$ is a solution of $G=0$, and vice versa, then it is an hetero-B\"acklund transformation. Moreover, if $F\equiv G$, the relations $\mathcal{B}_i=0$ is an auto-B\"acklund transformation.
\end{definition}

The simplest example of BT are the well-known Cauchy-Riemann relations in complex analysis, for the analyticity of a complex function, $f=u(x,t)+v(x,t)i$.

\begin{example}\normalfont (Laplace equation)
Functions $u=u(x,t)$ and $v=v(x,t)$ are harmonic, namely
\begin{equation}\label{Laplace}
\nabla^2u=0, \quad \nabla^2v=0,
\end{equation}
if the following Cauchy-Riemann relations hold
\begin{equation}\label{Cauchy-Riemann}
u_x=v_t, \quad u_t=-v_x.
\end{equation}
The latter equations constitute an auto-BT for the Laplace equation (\ref{Laplace}) and can be used to construct solutions of the same equations, starting with known ones. For instance, consider the simple solution $v(x,t)=xt$. Then, according to (\ref{Cauchy-Riemann}), a second solution of (\ref{Laplace}), $u$, has to satisfy $u_x=x$ and $u_t=-t$. Therefore, $u$ is given by
\begin{equation}
u=\frac{1}{2}(x^2-t^2).
\end{equation}
\end{example}

However, even though Laplace's equation is linear, the same idea works for nonlinear equations.

%%%%%%%%%%%%%%%%%%%%%%%%%%%%%%%%%%%%%%%%%%%%%%%%%%%%%%%%%%%%%%%%%%%%%%%%%%%%%%%%%%%%%%

\subsubsection{B\"acklund transformation for the PKdV equation}
An auto-B\"acklund transformation associated to the PKdV equation,
\begin{equation}\label{pkdv}
u_t=6 u_x^2-u_{xxx},
\end{equation}
is given by the following relations
\begin{align}\label{BTKdV}
\mathcal{B}_\alpha:
\begin{cases}
\left(u+v\right)_x&=2\alpha+\frac{1}{2}(u-v)^2,\\
\left(u-v\right)_t&=3(u_x^2-v_x^2)-(u-v)_{xxx},
\end{cases}
\end{align}
which was first presented in 1973 in a paper of Wahlquist and Estabrook \cite{Wahlquist-Estabrook}. In this subsection we show how we can construct algebraically a solution of the PKdV equation, using Bianchi's permutability.

\begin{remark}\normalfont
We shall refer to the first equation of (\ref{BTKdV}) as the \textit{spatial part} (or $x$-\textit{part}) of the BT, while we refer to the second one as the \textit{temporal part} (or $t$-\textit{part}) of the BT.
\end{remark}

\begin{exercise}
Show that relations (\ref{BTKdV}) constitute an auto-B\"acklund transformation for the PKdV equation \eqref{pkdv}.
\end{exercise}

Now, let $u=u(x,t)$ be a function satisfying the PKdV equation. Focusing on the spatial part of the BT (\ref{BTKdV}), we can construct a new solution of the PKdV equation, $u_1=\mathcal{B}_{\alpha_1}(u)$, i.e.
\begin{equation}
\left(u_1+u\right)_x=2\alpha_1+\frac{1}{2}(u_1-u)^2.\label{BTKdVu1}
\end{equation}
Moreover, using another parameter, $\alpha_2$, we can construct a second solution $u_2=\mathcal{B}_{\alpha_2}(u)$, given by
\begin{equation}
\left(u_2+u\right)_x=2\alpha_2+\frac{1}{2}(u_2-u)^2.\label{BTKdVu2}
\end{equation}

Starting with the solutions $u_1$ and $u_2$, we can construct two more solutions from relations $u_{12}=\mathcal{B}_{\alpha_2}(u_1)$ and $u_{21}=\mathcal{B}_{\alpha_1}(u_2)$, namely
\begin{subequations}
\begin{eqnarray}
\left(u_{12}+u_1\right)_x&=&2\alpha_2+\frac{1}{2}(u_{12}-u_1)^2,\label{BTKdVu12}\\
\left(u_{21}+u_2\right)_x&=&2\alpha_1+\frac{1}{2}(u_{21}-u_2)^2,\label{BTKdVu21}
\end{eqnarray}
\end{subequations}
as represented in Figure \ref{BianchiPer}-(a).

\begin{figure}[ht]
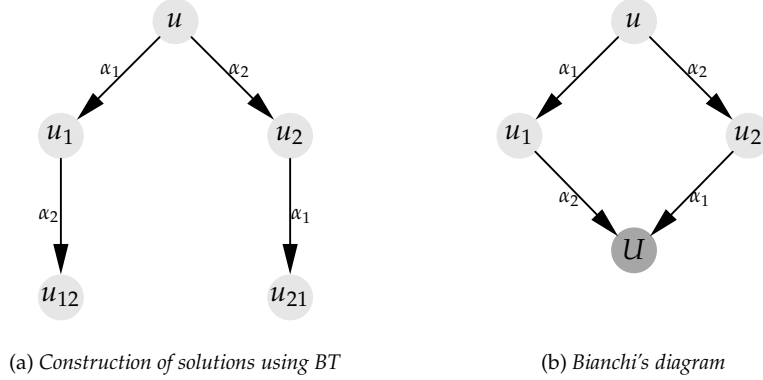

\centering
\centertexdraw{ \setunitscale 0.6
%Circles
\move (0 1)\fcir f:.9 r:.2  \move (-1 0)\fcir f:.9 r:.2 \move (1 0)\fcir f:.9 r:.2 \move (-1 -1.41)\fcir f:.9 r:.2 \move (1 -1.41)\fcir f:.9 r:.2
\move (4 1)\fcir f:.9 r:.2  \move (3 0)\fcir f:.9 r:.2 \move (5 0)\fcir f:.9 r:.2 \move (4 -1)\fcir f:.65 r:.2
%variables_in_circles
\textref h:C v:C \htext(0 1){$u$} \textref h:C v:C \htext(-1 0){$u_1$} \textref h:C v:C \htext(1 0){$u_2$} \textref h:C v:C \htext(-1 -1.41){$u_{12}$} \textref h:C v:C \htext(1 -1.41){$u_{21}$}
\textref h:C v:C \htext(4 1){$u$} \textref h:C v:C \htext(3 0){$u_1$} \textref h:C v:C \htext(5 0){$u_2$} \textref h:C v:C \htext(4 -1){$U$}
%arrows
\move (-.14 .86) \arrowheadtype t:F \avec(-.86 .14) \move (.14 .86) \arrowheadtype t:F \avec(.86 .14)
\move (-1 -.2) \arrowheadtype t:F \avec(-1 -1.21) \move (1 -.2) \arrowheadtype t:F \avec(1 -1.21)
\move (3.86 .86) \arrowheadtype t:F \avec(3.14 .14) \move (4.14 .86) \arrowheadtype t:F \avec(4.86 .14)
\move (3.14 -.14) \arrowheadtype t:F \avec(3.86 -0.88)
\move (4.86 -.14) \arrowheadtype t:F \avec (4.14 -0.88)
%Labels
\textref h:C v:C \scriptsize{\htext(0 -2){(a) \textit{Construction of solutions using BT}}}
\textref h:C v:C \scriptsize{\htext(4 -2){(b) \textit{Bianchi's diagram}}}
%parameters
\textref h:C v:C \htext(-.56 .56){$\alpha_1$}
\textref h:C v:C \htext(.56 .56){$\alpha_2$}
\textref h:C v:C \htext(-1.1 -.705){$\alpha_2$}
\textref h:C v:C \htext(1.1 -.705){$\alpha_1$}
\textref h:C v:C \htext(3.44 .56){$\alpha_1$}
\textref h:C v:C \htext(4.56 .56){$\alpha_2$}
\textref h:C v:C \htext(3.44 -.56){$\alpha_2$}
\textref h:C v:C \htext(4.576 -.56){$\alpha_1$}
}
\caption{Bianchi's permutability}\label{BianchiPer}
\end{figure}

Nevertheless, the above relations need integration in order to derive the actual solutions $u_1$, $u_2$ and, in retrospect, solutions $u_{12}$ and $u_{21}$. Yet, having at our disposal solutions $u_1$ and $u_2$, a new solution can be constructed using Bianchi's permutativity (see Figure \ref{BianchiPer}-(b)) in a purely algebraic way. Specifically, we have the following.

\begin{proposition}
Imposing the condition $u_{12}=u_{21}$, the BTs \{(\ref{BTKdVu1})-(\ref{BTKdVu21})\} imply the following solution of the PKdV equation:
\begin{equation}
U=u-4\frac{\alpha_1-\alpha_2}{u_1-u_2}.
\end{equation}
\end{proposition}

\begin{proof}
It is straightforward calculation; one needs to subtract (\ref{BTKdVu2}) and (\ref{BTKdVu21}) by (\ref{BTKdVu1}) and (\ref{BTKdVu12}), respectively, and subtract the resulting equations.
\end{proof}

Later on we study the Darboux transformations for the NLS equation, and we shall see that BT arise naturally in the derivation of DT.

\begin{exercise} (\textit{B\"acklund transform for the sine-Gordon equation})
Show that the following relations
	\begin{align}\label{BTSG}
\mathcal{B}_\alpha:
\begin{cases}
\left(\frac{u+v}{2}\right)_x&=\alpha \sin\left(\frac{u-v}{2}\right),\\
\left(\frac{u-v}{2}\right)_t&=\frac{1}{\alpha} \sin\left(\frac{u+v}{2}\right),
\end{cases}
\end{align}
between functions $u=u(x,t)$ and $v=v(x,t)$, constitute an auto-BT for the sine-Gordon equation:
\begin{equation}\label{SG-eq}
u_{xt}=\sin u.
\end{equation}
\item Following the same procedure as for the PKdV equation, namely using the BT (\ref{BTSG}) and Bianchi's permutability, show that:
\begin{equation}\label{U-solSG}
U=u+4\arctan\left[\frac{\alpha_1+\alpha_2}{\alpha_2-\alpha_1}\tan\left(\frac{u_1-u_2}{4}\right)\right],
\end{equation}
where $u$ satisfies the equation (\ref{SG-eq}), is also a solution of (\ref{SG-eq}).
\end{exercise}

%%%%%%%%%%%%%%%%%%%%%%%%%%%%%%%%%%%%%%%%%%%%%%%%%%%%%%%%%%%%%%%%%%%%%%%%%%%%
%%%           From Darboux transformations to Discrete systems           %%%
%%%%%%%%%%%%%%%%%%%%%%%%%%%%%%%%%%%%%%%%%%%%%%%%%%%%%%%%%%%%%%%%%%%%%%%%%%%%

\subsection{From Darboux transformations to Discrete systems}
In this section, we shall show how one can derive discrete integrable systems using Darboux transformations. In particular, our starting point will be continuous Lax operators, and by considering the associated Darboux transformations, we derive semi-discrete and fully discrete Lax operators. That is, in the derivation of Darboux transformations, we derive differential-difference (B\"acklund transformations, \cite{levi1980, levi1981}) and difference-difference integrable systems.

\subsubsection{Lax-Darboux scheme}
With the single term \textit{Lax-Darboux scheme}, we describe several structures which are related to each other and all of them are related to integrability. To be more precise, the Lax-Darboux scheme incorporates Lax operators, corresponding Darboux matrices and Darboux transformations, as well as the Bianchi permutability of the latter transformations.

In what follows we present the basic points of the scheme:

\begin{itemize}
	\item We consider Lax operators of the following AKNS form:
	\begin{equation*}
\mathbf{L}:=D_x+U(p,q;\lambda), \quad p=p(x),\, q=q(x),
\end{equation*}
where $U$ is a matrix which belongs in the Lie algebra ${\mathfrak{sl}}(2,{\mathbb{C}})$.
\begin{remark}\normalfont
The abbreviation AKNS is due to Ablowitz, Kaup, Newell and Segur who solved the sine-Gordon equation \eqref{SG-eq} (see \cite{AKNS}) writing it as compatibility condition of a set of Lax operators in the form $\textbf{L}=D_x+U$ and $\textbf{T}=D_t+V$, where $U$ and $V$ are certain matrices. Moreover, they generalized this method to cover a number of PDEs \cite{AKNS2}. It is worth mentioning that the authors of 
\cite{AKNS, AKNS2} were motivated by Zakharov and Shabat \cite{ZS-1} who applied the inverse scattering method to the nonlinear Schr\"odinger equation.
\end{remark}

\item By \textit{Darboux transformation}, here, we understand a map
\begin{equation}\label{DarbouxDef}
\mathbf{L} \rightarrow M\mathbf{L}M^{-1}=\mathbf{L}_1:=D_x+\underbrace{U(p_{10},q_{10};\lambda)}_{U_{10}}.
\end{equation}
Matrix $M$ is called \textit{Darboux matrix} and satisfies the following equation (see exercise \ref{exDefDar}).
\begin{equation}\label{DMeq}
D_xM+U_{10}M-MU=0,\quad M=M(p,q,p_{10},q_{10},f;\lambda).
\end{equation}
The Darboux matrix $M$ depends on the ``old" potentials, $p,q$ and the ``updated" ones, $p_{10},q_{10}$. It also depends on the spectral parameter and may depend on an auxiliary function, as we shall see later in the example.

\textit{Darboux transformations} consist of \textit{Darboux matrices} $M$ along with corresponding \textit{B\"acklund transformations} (or dressing chains).

\begin{remark}\normalfont
It is obvious from relations (\ref{DarbouxDef})-(\ref{DMeq}) that, for a given operator $\mathbf{L}$ we cannot determine $M$ in full generality, and we need to make an ansatz. In what follows, we shall be assuming that $M$ has the same dependence on the spectral parameter as $\mathbf{L}$.
\end{remark}

\begin{exercise}\label{exDefDar}
 Show that according to (\ref{DarbouxDef}) for a Darboux transformation, the corresponding Darboux matrix satisfies equation (\ref{DMeq}).
\end{exercise}

\item The associated B\"acklund transformation is a set of differential equations relating the potentials and the auxiliary functions involved in $\mathbf{L}$ and 
$\widetilde{\mathbf{L}}_1$. It can be regarded as integrable systems of differential-difference equations D$\Delta$Es \cite{levi1980, levi1981}.  This follows from the interpretation of the corresponding Darboux transformation defining a sequence
$$\cdots \stackrel{M}{\rightarrow} (p_{-10},q_{-10}) \stackrel{M}{\rightarrow} (p,q) \stackrel{M}{\rightarrow} (p_{10},q_{10})  \cdots.$$

\item A Darboux matrix $M$ maps a fundamental solution of the equation $\mathbf{L}\Psi=0$, where $\Psi=\Psi(x,\lambda)$ to a fundamental solution $\Psi_{10}$ of $\mathbf{L}_1\Psi_{10}=0$ according to $\Psi_{10}=M\Psi$. In general, matrix $M$ is invertible and depends on $p$ and $q$, their updates $p_{10}$ and $q_{10}$, the spectral parameter $\lambda$, and some auxiliary functions, i.e. $M=M(p,q,p_{10},q_{10};\lambda)$.

Moreover, the determinant of the Darboux matrix is independent of $x$ (see exercise \ref{fundSolDet}).

\begin{exercise}\label{fundSolDet}
 Using definition (\ref{DarbouxDef}), where $U\in {\mathfrak{sl}}(2,{\mathbb{C}})$ show that matrix $M$:
\begin{enumerate}
		\item maps fundamental solutions of $\mathbf{L}\Psi=0$ to fundamental solutions of $\mathbf{L}_1\Psi=0$.
		\item has determinant independent of $x$, namely $\partial_x(\det M)=0$. (Hint: Use Liouville's well-known formula for the determinant of the solution of the equation $\frac{d}{dx}\Psi=M\Psi$).
\end{enumerate}
\end{exercise}

\item We employ Darboux Matrices to derive two new fundamental solutions $\Psi_{10}$ and $\Psi_{01}$:
We can employ Darboux Matrices to derive two new fundamental solutions $\Psi_{10}$ and $\Psi_{01}$:
\begin{equation*}
\Psi_{10}=M(p,q,p_{10},q_{10};\lambda)\Psi \equiv M\Psi, \quad \Psi_{01}=M(p,q,p_{01},q_{01};\lambda)\Psi \equiv K\Psi
\end{equation*}

\item Considering compatibility around the square, a third solution can be derived as descriped in the \textit{Bianchi-type diagram} \ref{fig:bianchi}, where
$M_{01}:=M(p_{01},q_{01},p_{11},q_{11};\lambda)$ and $K_{10}:=M(p_{10},q_{10},p_{11},q_{11};\lambda)$.
\begin{figure}[ht!]
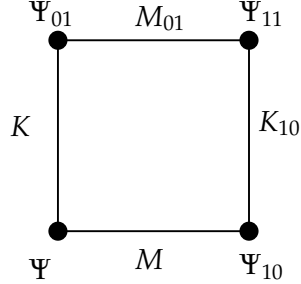

\centertexdraw{ \setunitscale 0.5
\linewd 0.02 \arrowheadtype t:F
\htext(0 0.3) {\phantom{T}}
\move (-1 -2) \lvec (1 -2)
\move(-1 -2) \lvec (-1 0) \move(1 -2) \lvec (1 0) \move(-1 0) \lvec(1 0)
\move (1 -2) \fcir f:0.0 r:0.1 \move (-1 -2) \fcir f:0.0 r:0.1
 \move (-1 0) \fcir f:0.0 r:0.1 \move (1 0) \fcir f:0.0 r:0.1
\htext (-1.3 -2.5) {$\Psi$} \htext (.9 -2.5) {$\Psi_{10}$} \htext (-.2 -2.4) {$M$}
\htext (-1.3 .15) {$\Psi_{01}$} \htext (.9 .15) {$\Psi_{11}$} \htext (-.2 .1) {$M_{01}$}
\htext (-1.5 -1) {$K$} \htext (1.1 -1) {$K_{10}$}}
\caption{{\em{Bianchi commuting diagram}}} \label{fig:bianchi}
\end{figure}

The compatibility around the square yields the following condition:
\begin{equation}
M_{01}K-K_{01}M=0,
\end{equation}
If the latter condition is written out explicitly, it results in algebraic relations among the various potentials involved, namely a system of P$\Delta$E equations.

We can interpret the above construction in a discrete way. Particularly, let us assume that $p$ and $q$ are functions depending not only on $x$ but also on two discrete variables $n$ and $m$, i.e. $p=p(x;n,m)$ and $q=q(x;n,m)$. Furthermore, we define the \textit{shift operators} $\mathcal{S}$ and $\mathcal{T}$ acting on a function $f=f(n,m)$ as
\begin{equation}\label{shiftST}
f_{10}:=\mathcal{S}f(n,m)=f(n+1,m), \qquad f_{01}:=\mathcal{T}f(n,m)=f(n,m+1).
\end{equation}
In general,
\begin{equation}\label{shiftSiTj}
f_{ij}:=\mathcal{S}^iT^jf(n,m)=f(n+i,m+j)
\end{equation}
We shall refer to $\mathcal{S}$ and $\mathcal{T}$ as the \textit{shift operators in the} $n$ and the $m$-\textit{direction}, respectively.
\end{itemize}

\subsubsection{An example: Lax-Darboux scheme for the nonlinear Schr\"odinger equation}\label{sec-NLS}
Consider the following Lax operator
\begin{subequations}\label{NLS-U}
\begin{align}
&\mathbf{L} := D_x + U(p,q;\lambda)=D_x +\lambda U^{1}+U^{0},
\intertext{where $U^{1}$ and $U^0$ are given by}
&U^1\equiv \sigma_3={\diag}(1,-1),\quad U^0=\left(\begin{array}{cc} 0 & 2p \\ 2q & 0\end{array}\right).
\end{align}
\end{subequations}
Operator (\ref{NLS-U}) constitutes the spatial part of the Lax pair for the nonlinear Schr\"odinger equation,
\begin{equation*}
p_t=p_{xx}+4p^2q,~~~q_t=-q_{xx}-4pq^2.
\end{equation*}

Now, we seek \textit{elementary} Darboux transformation, namely Darboux transformation that cannot be decomposed to others. We assume that the associated Darboux matrix, $M$, has linear $\lambda$-dependence, namely it is of the form $M=\lambda M_{1}+M_{0}$. Substitution of the latter to equation (\ref{DMeq}) implies a second order algebraic equation in $\lambda$. Equating the coefficients of the several powers of $\lambda$ equal to zero, we obtain the following system of equations
\begin{subequations}
\begin{align}
\lambda^2: & \quad \left[\sigma_3,M_1\right]=0,\label{M-syst-1}\\
\lambda^1:   & \quad M_{1,x}+\left[\sigma_3,M_0\right]+U^{0}_{10}M_1-M_1U^0=0,\label{M-syst-2}\\
\lambda^0: & \quad M_{0,x}+U^0_{10}M_0-M_0U^0=0,\label{M-syst-3}
\end{align}
\end{subequations}
where with $\left[\sigma_3,M_1\right]$ we denote the commutator of $\sigma_3$ and $M_1$.

Equation (\ref{M-syst-1}) implies that $M_1$ must be diagonal, i.e. $M_1=\diag(c_1,c_2)$. Then, we substitute $M_1$ to (\ref{M-syst-2}).

Now, for simplicity of the notation, we denote the $(1,1)$ and $(2,2)$ entries of $M_0$ by $f$ and $g$ respectively. Then, it follows from equation (\ref{M-syst-3}) that the entries of matrix $M_0$ must satisfy the following equations
\begin{subequations} \label{fgpq10}
\begin{eqnarray}
\partial_x f&=&2 (M_{0,12}q-p_{10}M_{0,21}),\label{fgpq10-1}\\
\partial_x g&=&2 (M_{0,21}p-q_{10}M_{0,12}),\label{fgpq10-2}\\
\partial_x M_{0,12} &=&2 (p f -g p_{10}), \label{fgpq10-3}\\
\partial_x M_{0,21}&=&2 ( q g-q_{10}f).\label{fgpq10-4}
\end{eqnarray}
\end{subequations}

The off-diagonal part of (\ref{M-syst-2}) implies that the $(1,2)$ and $(2,1)$ entries of matrix $M_0$ are given by
\begin{equation}\label{M012M021}
M_{0,12}=c_1 p-c_2 p_{10}, \quad M_{0,21}=c_1q_{10}-c_2 q.
\end{equation}
Additionally, from the diagonal part of (\ref{M-syst-2}) we deduce that $c_{1,x}=c_{2,x}=0$. Therefore, $c_i$, $i=1,2$ can be either zero or nonzero. Thus, after rescaling we can choose either $c_1=1$, $c_2=0$ ($\rank{M_1}=1$), or $c_1=0$, $c_2=-1$ ($\rank{M_1}=1$), or $c_1=1, c_2=-1$ ($\rank{M_1}=2$). The first two choices correspond to gauge equivalent elementary Darboux matrices and the third one leads to a Darboux matrix which can be given as a composition of the two previous elementary matrices.

Indeed, the choice $c_1=1$, $c_2=0$ implies $M_{0,12}=p$ and $M_{0,21}=q_{10}$. Moreover, (\ref{fgpq10-2}) implies that $g=\text{const.}=\alpha$, i.e.
\begin{equation}
M_0=\left(\begin{array}{cc}f & p\\
q_{10} & \alpha \end{array}\right).
\end{equation}
In this case the Darboux matrix is given by
\begin{equation}\label{M-NLS-const.}
\lambda \left(\begin{array}{cc}1 & 0\\0 & 0\end{array}\right)+\left(\begin{array}{cc}f & p\\q_{10} & \alpha \end{array}\right),
\end{equation}
where, according to (\ref{fgpq10}), its entries satisfy
\begin{equation}\label{nls-comp-cond-alpha}
\partial_x f=2 (pq-p_{10}q_{10}),\quad \partial_x p =2 (p f -\alpha p_{10}),\quad \partial_x q_{10}=2 ( cq -q_{10}f).
\end{equation}

Now, if $\alpha\neq 0$, it can be rescaled to $\alpha=1$.

In the case where $\alpha=0$, from (\ref{nls-comp-cond-alpha}) we deduce
\begin{equation}\label{pxq}
p_x=2fp,\quad q_{10,x}=-2fq_{10}.
\end{equation}
Thus, the Darboux matrix in this case is given by
\begin{equation}
M(p,q_{10},f):=\lambda \left(\begin{array}{cc}1 & 0\\0 & 0\end{array}\right)+\left(\begin{array}{cc}f & p\\q_{10} & 0 \end{array}\right).
\end{equation}
In addition, after an integration with respect to $x$, equations (\ref{pxq}) imply that $q_{10}=c/p$.

In general, we have the following.

\begin{proposition}\label{M-NLS-proposition}
Let $M$ be an elementary Darboux matrix for the Lax operator (\ref{NLS-U}) and suppose it is linear in $\lambda$. Then, up to a gauge transformation, $M$ is given by
\begin{equation}\label{M-NLS}
M(p,q_{10},f):=\lambda \left(\begin{array}{cc}1 & 0\\0 & 0\end{array}\right)+\left(\begin{array}{cc}f & p\\q_{10} & 1 \end{array}\right),
\end{equation}
where the potentials $p$ and $q$ satisfy the following differential-difference equations
\begin{subequations} \label{nls-comp-cond}
\begin{eqnarray}
&& \partial_x f=2 (pq-p_{10}q_{10}),\label{nls-comp-cond-1}\\
&& \partial_x p =2 (p f -p_{10}), \label{nls-comp-cond-2}\\
&& \partial_x q_{10}=2 ( q-q_{10}f).\label{nls-comp-cond-3}
\end{eqnarray}
\end{subequations}
Moreover, matrix (\ref{M-NLS}) degenerates to
\begin{equation}\label{M-degen}
M_c(p,f)=\lambda \left(\begin{array}{cc}1 & 0\\0 & 0\end{array}\right)+\left(\begin{array}{cc}f & p\\\frac{c}{p} & 0 \end{array}\right),\quad f=\frac{p_x}{2p}.
\end{equation}
\end{proposition}

It is straightforward to show that system (\ref{nls-comp-cond}) admits the following first integral
\begin{equation} \label{nls-const}
\partial_x \left(f-p\,q_{10} \right)\,=\,0\,.
\end{equation}
which implies that $\partial_x\det M=0$.

\begin{exercise}
Show that the choice $c_1=0, c_2=-1$ in (\ref{M012M021}) leads to a Darboux matrix gauge equivalent to (\ref{M-NLS}), and, in particular, show that matrix $\sigma_1 N(p_{10},q,g) \sigma_1^{-1}$ is of the form (\ref{M-NLS}).
\end{exercise}

\begin{remark}\normalfont
In this case, of the nonlinear Schr\"odinger equation, an elementary Darboux transformation consists of an elementary Darboux matrix \eqref{M-NLS}, which is $\lambda$-depended, and a system of D$\Delta$Es, namely system \eqref{nls-comp-cond}. The latter is nothing else but the spatial part of a B\"acklund transformation associated to the nonlinear Schr\"odinger equation, and it does not depend on the spectral parameter, $\lambda$.
\end{remark}

%%%%%%%%%%%%%%%%%%%%%%%%%%%%%%%%%%%%%%%%%%%%%%%%%%%%%%%%%%%%%%%%%%%%%%%%%%%%%%%%%%%%%%
%%%%          Derivation of discrete systems and initial value problems           %%%%
%%%%%%%%%%%%%%%%%%%%%%%%%%%%%%%%%%%%%%%%%%%%%%%%%%%%%%%%%%%%%%%%%%%%%%%%%%%%%%%%%%%%%%

\subsection{Derivation of discrete systems and initial value problems}
In this section we employ the Darboux matrices derived in the previous section to derive discrete integrable systems. We shall present only the pairs of Darboux matrices which lead to genuinely non-trivial discrete integrable systems.
For these systems we consider an initial value problem on the staircase.

\subsubsection{Nonlinear Schr\"odinger equation and related discrete systems}
Having derived two Darboux matrices for operator (\ref{NLS-U}), we focus on the one given in (\ref{M-NLS}) and consider the following discrete Lax pair
\begin{equation}\label{dLaxP}
\Psi_{10} = M \Psi,\quad \Psi_{01} = K\Psi,
\end{equation}
where $M$ and $K$ are given by
\begin{subequations} \label{NLS-disc-LP}
\begin{align}
&M\equiv M(p,q_{10},f)=\lambda \left(\begin{array}{cc} 1 & 0\\0 & 0 \end{array}\right)+\left(\begin{array}{cc} f & p\\q_{10} & 1\end{array}\right),\\
&K\equiv M(p,q_{01},g)= \lambda \left(\begin{array}{cc} 1 & 0\\ 0 & 0\end{array}\right)+\left(\begin{array}{cc} g & p\\ q_{01} & 1\end{array}\right).
\end{align}
\end{subequations}
The compatibility condition of (\ref{NLS-disc-LP}) results to
\begin{subequations} \label{nls-comp}
\begin{eqnarray}
 f_{01} -f - \left( g_{10}-g\right) = 0,&&\label{nls-comp-1}\\
 f_{01}\,g-fg_{10}-p_{10}q_{10}+p_{01}q_{01}=0,&&\\
 p \left(f_{01}-g_{10} \right)-p_{10}+p_{01}=0,&&\\
 q_{11}\left(f-g\right)-q_{01}+q_{10}=0.&&
\end{eqnarray}
\end{subequations}
This system can be solved either for $(p_{01},q_{01},f_{01},g)$ or for $(p_{10},q_{10},f,g_{10})$. In either of these cases, we derive two solutions. The first one is
\begin{equation} \label{triv-sol}
p_{10} = p_{01},\quad q_{10} = q_{01}, \quad f = g,\quad g_{10} = f_{01},
\end{equation}
which is trivial and corresponds to $M(p,q_{10},f)=M(p,q_{01},g)$.

The second solution is given by
\begin{subequations} \label{LPcompatEq}
\begin{eqnarray}
\hspace{-.3cm}&& p_{01} = \frac{q_{10} p^2 + (g_{10} - f) p + p_{10}}{1+p q_{11}},\quad
q_{01} = \frac{p_{10}{ q_{11}}^{2} + (f-g_{10}) q_{11} + q_{10}}{1+p q_{11}},  \\
\hspace{-.3cm}&& f_{01} = \frac{q_{11} (p_{10} + p g_{10}) + f - p q_{10}}{1+ p q_{11}},\quad
g = \frac{q_{11} (p f- p_{10}) + g_{10}+p q_{10}}{1+pq_{11}}.
\end{eqnarray}
\end{subequations}

The above system has some properties which take their rise in the derivation of the Darboux matrix. In particular, we have the following.

\begin{proposition}
System (\ref{LPcompatEq}) admits two first integrals, $\mathcal{F}:=f-pq_{10}$ and $\mathcal{G}:=g-pq_{01}$, and the following conservation law
	\begin{equation}\label{conlaw}
	(\mathcal{T}-1)\mathcal{F}=(\mathcal{S}-1)\mathcal{G}
	\end{equation}
\end{proposition}
\begin{proof}
Relation (\ref{nls-const}) suggests that
	\begin{equation}\label{1stInts}
	({\mathcal{T}}-1)\left(f-pq_{10}\right)=0\quad {\mbox{and}} \quad ({\mathcal{S}}-1)\left(g-pq_{01}\right)=0,
	\end{equation}
which can be shown by straightforward calculation, and it is left as an exercise. Thus, $F=f-pq_{10}$ and $G:=g-pq_{01}$ are first integrals. Moreover, equation (\ref{nls-comp-1}) can be written in the form of the conservation law (\ref{conlaw}).
\end{proof}

\begin{corollary}
The following relations hold.
\begin{equation} \label{nls-fi}
f-pq_{10}=\alpha(n)\quad {\mbox{and}} \quad g-pq_{01}=\beta(m).
\end{equation}
\end{corollary}

\begin{remark}\normalfont
In view of relations (\ref{nls-fi}), we can interpret functions $f$ and $g$ as being given on the edges of the quadrilateral where system (\ref{LPcompatEq}) is defined, and, consequently, consider system (\ref{LPcompatEq}) as a vertex-bond system \cite{HV}.
\end{remark}

\begin{exercise}
Show relations (\ref{1stInts}) using (\ref{nls-comp}).
\end{exercise}

Our choice to solve system (\ref{nls-comp}) for $p_{01}$, $q_{01}$, $f_{01}$ and $g$ is motivated by the initial value problem related to system (\ref{LPcompatEq}). Suppose that the initial values for $p$ and $q$ are given on the vertices along a staircase as shown in Figure \ref{fig-ivp}. Functions $f$ and $g$ are given on the edges of this initial value configuration in a consistent way with the first integrals (\ref{nls-fi}). In particular, horizontal edges carry the initial values of $f$ and vertical edges the corresponding ones of $g$.

\begin{figure}[ht]
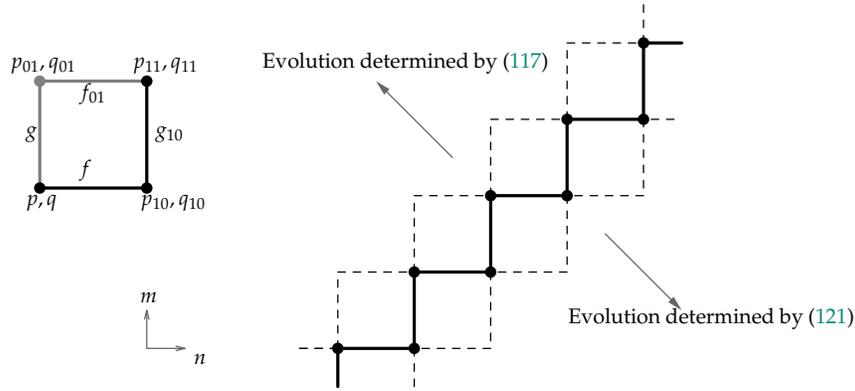

\centertexdraw{
\setunitscale 0.4
\move(-4.5 -2) \linewd 0.02 \setgray 0.4 \arrowheadtype t:V \arrowheadsize l:.12 w:.06 \avec(-4.5 -1.5)
\move(-4.5 -2) \arrowheadtype t:V  \avec(-4 -2)
\arrowheadsize l:.20 w:.10
\move(-.5 .5) \linewd 0.02 \setgray 0.4 \arrowheadtype t:F \avec(-1.5 1.5)
\move(1.5 -.5) \linewd 0.02 \setgray 0.4 \arrowheadtype t:F \avec(2.5 -1.5)
\setgray 0.5
\linewd 0.04 \move (-4.5 1.5)  \lvec (-5.9 1.5) \lvec (-5.9 .1)
\move (-5.9 1.5) \fcir f:0.5 r:0.075
\htext (-6.1 .7) {\scriptsize{$g$}}
\htext (-4.4 .7) {\scriptsize{$g_{10}$}}
\setgray 0.0
\linewd 0.04 \move (-2 -2.5) \lvec (-2 -2) \lvec (-1 -2) \lvec (-1 -1) \lvec (0 -1) \lvec (0 0) \lvec (1 0) \lvec(1 1) \lvec (2 1) \lvec(2 2) \lvec(2.5 2)
\move (-5.9 .1) \lvec (-4.5 .1) \lvec (-4.5 1.5)
\linewd 0.015 \lpatt (.1 .1 ) \move (-2 -2) \lvec (-2 -1) \lvec(-1 -1) \lvec (-1 0) \lvec (0 0) \lvec (0 1) \lvec(1 1) \lvec (1 2) \lvec (2 2) \lvec (2 2.5)
\move(-2.5 -2) \lvec(-2 -2) \move(-2.5 -2) \lvec(-2 -2)
\move (-1 -2.5) \lvec (-1 -2) \lvec(0 -2) \lvec(0 -1) \lvec(1 -1) \lvec(1 0) \lvec(2 0) \lvec(2 1) \lvec(2.5 1)
%\move (-2 -3) \fcir f:0.0 r:0.075 \move (-3 -3) \fcir f:0.0 r:0.075
\move (-2 -2) \fcir f:0.0 r:0.075 \move (-1 -2) \fcir f:0.0 r:0.075
\move (-1 -1) \fcir f:0.0 r:0.075 \move (0 -1) \fcir f:0.0 r:0.075
\move (0 0) \fcir f:0.0 r:0.075 \move (1 0) \fcir f:0.0 r:0.075
\move (1 1) \fcir f:0.0 r:0.075 \move (2 1) \fcir f:0.0 r:0.075
\move (2 2) \fcir f:0.0 r:0.075
\move (-5.9 .1) \fcir f:0.0 r:0.075 \move (-4.5 .1) \fcir f:0.0 r:0.075 \move (-4.5 1.5) \fcir f:0.0 r:0.075
\htext (-3.9 -2.2) {\scriptsize{$n$}}
\htext (-4.6 -1.4) {\scriptsize{$m$}}
\htext (-6.1 -.2) {\scriptsize{$p,q$}}
\htext (-4.6 -.2) {\scriptsize{$p_{10},q_{10}$}}
\htext (-5.4 0.2) {\scriptsize{$f$}}
\htext (-5.4 1.2) {\scriptsize{$f_{01}$}}
\htext (-6.3 1.6) {\scriptsize{$p_{01},q_{01}$}}
\htext (-4.7 1.6) {\scriptsize{$p_{11},q_{11}$}}
\htext (-3 1.6) {{\scriptsize{Evolution determined by (\ref{LPcompatEq})}}}
\htext (1 -1.75) {{\scriptsize{Evolution determined by (\ref{LPcompatEq-2})}}}
}
\caption{{Initial value problem and direction of evolution}} \label{fig-ivp}
\end{figure}

With these initial conditions, the values of $p$ and $q$ can be uniquely determined at every vertex of the lattice, while $f$ and $g$ on the corresponding edges. This is obvious from the rational expressions (\ref{LPcompatEq}) defining the evolution above the staircase, cf. Figure \ref{fig-ivp}.

For the evolution below the staircase, one has to use
\begin{subequations}\label{LPcompatEq-2}
\begin{eqnarray}
p_{10} &=& \frac{q_{01} p^2 + (f_{01} - g) p + p_{01}}{1+p\, q_{11}},\quad
q_{10} = \frac{p_{01}{ q_{11}}^{2} + (g-f_{01})q_{11} + q_{01}}{1+pq_{11}},\\
g_{10} &=& \frac{q_{11} (p_{01} + p f_{01}) + g - p q_{01}}{1+ p q_{11}},\quad
f = \frac{q_{11} (p g- p_{01}) + f_{01}+p q_{01}}{1+p\,q_{11}},
\end{eqnarray}
\end{subequations}
which uniquely defines the evolution below the staircase as indicated in Figure \ref{fig-ivp}.

\begin{remark}\normalfont
We could consider more general initial value configurations of staircases of lengths $\ell_1$ and $\ell_2$ in the $n$ and $m$ lattice direction, respectively. Such initial value problems are consistent with evolutions (\ref{LPcompatEq}), (\ref{LPcompatEq-2}) determining the values of all fields uniquely at every vertex and edge of the lattice.
\end{remark}

Now, using first integrals we can reduce system (\ref{LPcompatEq}) to an \textit{Adler-Yamilov type} of system as those in \cite{Adler-Yamilov}. Specifically, we have the following.

\begin{proposition}
System (\ref{LPcompatEq}) can be reduced to the following non-autonomous Adler-Yamilov type of system for $p$ and $q$:
\begin{equation} \label{nls-pq-sys}
p_{01}=p_{10}-\frac{\alpha(n)-\beta(m)}{1+ pq_{11}}p,\quad q_{01}=q_{10}+\frac{\alpha(n)-\beta(m)}{1+ pq_{11}}q_{11}.
\end{equation}
\end{proposition}
\begin{proof}
The proof is straightforward if one uses relations (\ref{nls-fi}) to replace $f$ and $g$ in system (\ref{LPcompatEq}).
\end{proof}

Now, we will use two different Darboux matrices associated with the NLS equation to construct the discrete Toda equation \cite{Suris}.

In fact, we introduce a discrete Lax pair as (\ref{dLaxP}), with $M=M_1(p,f)$ in (\ref{M-degen}) and $K=M(p,q_{01},g)$ in (\ref{M-NLS}). That is, we consider the following system
\begin{subequations}
\begin{align}
\Psi_{10} &= \left(\lambda \left(\begin{array}{cc} 1 & 0\\0 & 0 \end{array}\right)+\left(\begin{array}{cc} f & p\\ \frac{1}{p} & 0 \end{array}\right)\right) \Psi,\\
\Psi_{01} &= \left( \lambda \left(\begin{array}{cc} 1 & 0\\ 0 & 0\end{array}\right)+\left(\begin{array}{cc} g & p\\ q_{01} & 1\end{array}\right) \right)\Psi,
\end{align}
\end{subequations}
and impose its compatibility condition.

From the coefficient of the $\lambda$-term in the latter condition we extract the following equations
\begin{subequations}
\begin{align}
f-f_{01}&=g-g_{10},\label{l1-term-1}\\
p_{01}&=\frac{1}{q_{11}}\label{l1-term-2}.
\end{align}
\end{subequations}
Additionally, the $\lambda^0$-term of the compatibility condition implies
\begin{subequations}
\begin{align}
f_{01}g-g_{10}f&=\frac{p_{10}}{p}-p_{01}q_{01},\label{l0-term-1}\\
g_{10}-f_{01}&=\frac{p_{01}}{p},\label{l0-term-2}\\
g-f&=\frac{p_{01}}{p}.\label{l0-term-3}
\end{align}
\end{subequations}

Now, recall from the previous section that, using (\ref{nls-fi}), the quantities $g$ and $g_{10}$ are given by
\begin{equation}
g=\beta(m)+pq_{01} \quad \text{and} \quad g_{10}=\beta(m)+p_{10}q_{11}.
\end{equation}
We substitute $g$ and $g_{10}$ into (\ref{l0-term-2}) and (\ref{l0-term-3}), and then replace $p$ and its shifts using (\ref{l1-term-2}). Then, we can express $f$ and $f_{01}$ in terms of the potential $q$ and its shifts:
\begin{subequations}\label{ff01}
\begin{align}
f&=\frac{q_{01}}{q_{10}}-\frac{q_{10}}{q_{11}}+\beta(m),\label{ff01-1}\\
f_{01}&=\frac{q_{11}}{q_{20}}-\frac{q_{10}}{q_{11}}+\beta(m).\label{ff01-2}
\end{align}
\end{subequations}

\begin{proposition}
The compatibility of system (\ref{ff01}) yields a fully discrete Toda type equation.
\end{proposition}
\begin{proof}
Applying the shift operator $\mathcal{T}$ on both sides of (\ref{ff01-1}) and demanding that its right-hand side agrees with that of (\ref{ff01-2}), we obtain
\begin{equation}
\frac{q_{11}}{q_{20}}-\frac{q_{02}}{q_{11}}+\frac{q_{11}}{q_{12}}-\frac{q_{10}}{q_{11}}=\beta(m+1)-\beta(m).
\end{equation}
Then, we make the transformation
\begin{equation}
q\rightarrow \exp(-w_{-1,-1}),
\end{equation}
which implies the following discrete Toda type equation
\begin{equation}\label{DToda}
\E^{w_{1,-1}-w}-\E^{w-w_{-1,1}}+\E^{w_{0,1}-w}-\E^{w-w_{0,-1}}=\beta(m+1)-\beta(m),
\end{equation}
and proves the statement.
\end{proof}

\begin{exercise}
Show that the discrete Toda equation (\ref{DToda}) can be written in the form of a conservation law.
\end{exercise}

\subsection{Lax-Darboux scheme for NLS type equations}
In the previous section we used the NLS equation as an illustrative example to describe the Lax-Darboux scheme. The NLS equation was not selected randomly, but as a simple example of recent classification results of automorphic Lie algebras. In particular, finite groups of fractional-linear transformations of a complex variable were classified by Klein \cite{Klein}, and they correspond to the cyclic groups $\mathbb{Z}_n$, the dihedral groups $\mathbb{D}_n$, the tetrahedral group $\mathbb{T}$, the octahedral group $\mathbb{O}$ and the icosahedral group $\mathbb{I}$. Following Klein's classification, in \cite{BuryPhD, Bury-Sasha} it has been shown that in the case of $2\times 2$ matrices ($n=2$), the essentially different reduction groups are
\begin{itemize}
\item the trivial group (with no reduction);
\item the cyclic reduction group $\mathbb{Z}_2$ (leading to the Kac-Moody algebra $A_1^1$);
\item the Klein reduction group $\mathbb{Z}_2\times \mathbb{Z}_2 \cong \mathbb{D}_2$.
\end{itemize}

Now, the following Lax operators
\begin{subequations}
\begin{align}
\label{Lax-NLS}
&\mathbf{L} =D_x +\lambda\left(\begin{array}{cc} 1 & 0 \\ 0 & -1\end{array}\right)+\left(\begin{array}{cc} 0 & 2p \\ 2q & 0\end{array}\right),\\
\label{Lax-DNLS}
&\mathbf{L}=D_x+\lambda^{2} \left(\begin{array}{cc} 1 & 0 \\ 0 & -1\end{array}\right)+\lambda \left(\begin{array}{cc} 0 & 2p \\ 2q & 0\end{array}\right),\\
\label{Lax-dDNLS}
&\mathbf{L}=D_x+(\lambda^2-\lambda^{-2})\left(\begin{array}{cc} 1 & 0 \\ 0 & -1\end{array}\right)+\lambda \left(\begin{array}{cc} 0 & 2\,p\\ 2\,q & 0\end{array}\right)+\lambda^{-1} \left(\begin{array}{cc} 0 & 2\,q\\ 2\,p & 0\end{array}\right),
\end{align}
\end{subequations}
constitute all the essential different Lax operators, with poles of minimal order, invariant with respect to the generators of $\mathbb{Z}_2$ and $\mathbb{D}_2$ groups with degenerate orbits\footnote{These are orbits corresponding to the fixed points of the fractional linear
transformations of the spectral parameter.}. In what follows, we study the Darboux transformations\index{Darboux transformation(s)} for all the above cases.

Operator (\ref{Lax-NLS}) is associated with the NLS equation  \cite{ZS}, while (\ref{Lax-DNLS}) and (\ref{Lax-dDNLS}) are associated with the DNLS equation \cite{Kaup-Newell}, and a deformation of the DNLS equation \cite{MSY}, respectively.

In \cite{sokor, SPS}, the Lax-Darboux scheme was applied to all cases \eqref{Lax-NLS}, \eqref{Lax-DNLS} and \eqref{Lax-dDNLS}. As a result, for all these cases, Darboux transformations were studied and novel discrete integrable systems were constructed.

%%%%%%%%%%%%%%%%%%%%%%%%%%%%%%%%%%%%%%%%%%%%%%%%%%%%%%%%%%%%%%%%%%%%%%%%%%%%%%%%%%%%%%%%%%%%%%%%%%
%%%%%%%%%%%%%%%%%   Yang-Baxter      %%%%%%%%%%%%%%%%%%%%%%%%%%%%%%%%%%%%%%%%%%%%%%%%%%%%%%%%%%%%%
%%%%%%%%%%%%%%%%%%%%%%%%%%%%%%%%%%%%%%%%%%%%%%%%%%%%%%%%%%%%%%%%%%%%%%%%%%%%%%%%%%%%%%%%%%%%%%%%%%

\section{Discrete integrable systems and Yang-Baxter maps}
\label{sec:4}
As we mentioned in the introduction, a very important integrability criterion is that of the 3D-\textit{consistency} and, by extension, the \textit{multidimensional consistency} \cite{Bobenko-Suris, Frank4}.

In what follows, we briefly explain what is the 3D-consistency property and we review some recent classification results. For more information we refer to \cite{Frank5, Hiet-Frank-Joshi} which are two of the few self-contained books on the integrability of discrete systems, as well as \cite{Gramm-Schw-Tam} for a collection of results.

\subsection{Equations on Quad-Graphs: 3D-consistency}
Let us consider a discrete equation of the form
\begin{equation}\label{QGeq}
Q(u,u_{10},u_{01},u_{11};a,b)=0,
\end{equation}
where $u_{ij}$, $i,j=0,1$, $u\equiv u_{00}$, belong in a set $\mathcal{A}$ and the parameters $a,b\in\mathbb{C}$. Moreover, we assume that (\ref{QGeq}) is uniquely solvable for any $u_i$ in terms of the rest. We can interpret the fields $u_i$ to be attached to the vertices of a square as in Figure \ref{3Dconsistency}-(a).

If equation (\ref{QGeq}) can be generalized in a consistent way on the faces of a cube, then it is said to be \textit{3D-consistent}. In particular, suppose we have the initial values $u$, $u_{100}$, $u_{010}$ and $u_{001}$ attached to the vertices of the cube as in Figure \ref{3Dconsistency}-(b). Now, since equation (\ref{QGeq}) is uniquely solvable, we can uniquely determine values $u_{110}$, $u_{101}$ and $u_{011}$, using the bottom, front and left face of the cube. Then, there are three ways to determine value $u_{111}$, and we have the following.

\begin{definition}
If for any choice of initial values $u$, $u_{100}$, $u_{010}$ and $u_{001}$, equation $Q=0$ produces the same value $u_{111}$ when solved using the left, back or top face of the cube, then it is called 3D-consistent.
\end{definition}

\begin{figure}[ht]
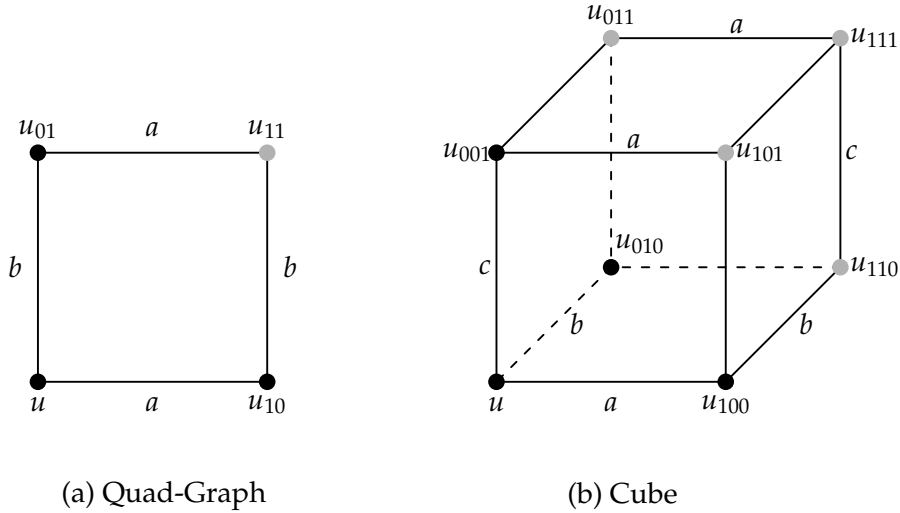

\centering
\centertexdraw{
\setunitscale 0.6
\move (-2 0)  \lvec(0 0) \lvec(0 2) \lvec(-2 2) \lvec(-2 0)
\move (-2 0) \fcir f:0.0 r:0.075 \move (0 0) \fcir f:0.0 r:0.075
\move (0 2) \fcir f:.7 r:0.075 \move (-2 2) \fcir f:0.0 r:0.075
\textref h:C v:C \htext(-2 -0.2){$u$}
\textref h:C v:C \htext(0 -0.2){$u_{10}$}
\textref h:C v:C \htext(-2 2.2){$u_{01}$}
\textref h:C v:C \htext(0 2.2){$u_{11}$}
\textref h:C v:C \htext(-1 -0.2){$a$}
\textref h:C v:C \htext(-1 2.2){$a$}
\textref h:C v:C \htext(-2.2 1){$b$}
\textref h:C v:C \htext(0.2 1){$b$}

\lpatt()
\setgray 0
\move (2 0)  \lvec(4 0) \lvec(5 1)
\lpatt(0.067 0.1) \lvec(3 1) \lvec(2 0)
\lpatt() \lvec(2 2) \lvec(3 3)
\lpatt (0.067 0.1) \lvec(3 1)
\lpatt() \move (3 3) \lvec(5 3) \lvec(4 2) \lvec(4 0)
\move (2 2) \lvec(4 2)
\move (5 3) \lvec(5 1)
\move (2 0) \fcir f:0.0 r:0.075 \move (4 0) \fcir f:0.0 r:0.075
\move (2 2) \fcir f:0.0 r:0.075 \move (4 2) \fcir f:0.7 r:0.075
\move (3 1) \fcir f:0.0 r:0.075 \move (5 1) \fcir f:0.7 r:0.075
\move (3 3) \fcir f:0.7 r:0.075 \move (5 3) \fcir f:0.7 r:0.075
\textref h:C v:C \htext(2 -.2){$u$}
\textref h:C v:C \htext(4 -.2){$u_{100}$}
\textref h:C v:C \htext(5.3 1){$u_{110}$}
\textref h:C v:C \htext(3.25 1.2){$u_{010}$}
\textref h:C v:C \htext(1.745 2){$u_{001}$}
\textref h:C v:C \htext(4.3 2){$u_{101}$}
\textref h:C v:C \htext(5.3 3){$u_{111}$}
\textref h:C v:C \htext(3 3.2){$u_{011}$}
\textref h:C v:C \htext(3 -.2){$a$}
\textref h:C v:C \htext(3.2 2.1){$a$}
\textref h:C v:C \htext(4.1 3.1){$a$}
\textref h:C v:C \htext(5.1 2){$c$}
\textref h:C v:C \htext(2.7 .5){$b$}
\textref h:C v:C \htext(4.7 .5){$b$}
\textref h:C v:C \htext(1.9 1){$c$}

\textref h:C v:C \htext(3.1 -1){(b) Cube}
\textref h:C v:C \htext(-0.9 -1){(a) Quad-Graph}
}
\caption{3D-consistency.}\label{3Dconsistency}
\end{figure}

\begin{note}\normalfont
In the above interpretation, we have adopted similar notation to \eqref{shiftST}-\eqref{shiftSiTj}: We consider the square in Figure \ref{3Dconsistency}-(a) to be an elementary square in a two dimensional lattice. Then, we assume that field $u$ depends on two discrete variables $n$ and $m$, i.e. $u=u(n,m)$. Therefore, $u_{ij}$s on the vertices of \ref{3Dconsistency}-(a) are
\begin{equation}
u_{00}=u(n,m),\quad u_{10}=u(n+1,m),\quad u_{01}=u(n,m+1),\quad u_{11}=u(n+1,m+1).
\end{equation}
Moreover, for the interpretation on the cube we assume that $u$ depends on a third variable $k$, such that
\begin{equation}
u_{000}=u(n,m,k),\quad u_{100}=u(n+1,m,k),\ldots\quad u_{111}=u(n+1,m+1,k+1).
\end{equation}
\end{note}

Now, as an illustrative example we use the discrete potential KdV equation which first appeared in \cite{HirotaKdV}.

\begin{example}\label{dpKdvEx}\normalfont (Discrete potential KdV equation)
Consider equation (\ref{QGeq}), where $Q$ is given by
\begin{equation}
Q(u,u_{10},u_{01},u_{11};a,b)=(u-u_{11})(u_{10}-u_{01})+b-a.
\end{equation}
Now, using the bottom, front and left faces of the cube \ref{3Dconsistency}-(b), we can solve equations
\begin{subequations}
\begin{align}
Q(u,u_{100},u_{010},u_{110};a,b)&=0,\\
Q(u,u_{100},u_{001},u_{101};a,c)&=0,\\
Q(u,u_{010},u_{001},u_{011};b,c)&=0,
\end{align}
\end{subequations}
to obtain solutions for $u_{110}$, $u_{101}$ and $u_{011}$, namely
\begin{subequations}\label{u110u101u001}
\begin{align}
u_{110}&=u+\frac{a-b}{u_{010}-u_{100}},\label{u110}\\
u_{101}&=u+\frac{a-c}{u_{001}-u_{100}},\label{u101}\\
u_{011}&=u+\frac{b-c}{u_{001}-u_{010}},\label{u011}
\end{align}
\end{subequations}
respectively.

Now, if we shift (\ref{u110}) in the $k$-direction, and then substitute $u_{101}$ and $u_{011}$ (which appear in the resulting expression for $u_{11}$) by (\ref{u110u101u001}), we deduce
\begin{equation}
u_{111}=-\frac{(a-b)u_{100}u_{010}+(b-c)u_{010}u_{001}+(c-a)u_{100}u_{001}}{(a-b)u_{001}+(b-c)u_{100}+(c-a)u_{010}}.
\end{equation}
It is obvious that, because of the symmetry in the above expression, we would obtain exactly the same expression for $u_{111}$ if we had alternatively shifted $u_{101}$ in the $m$-direction and substituted $u_{110}$ and $u_{011}$ by (\ref{u110u101u001}), or if we had shifted $u_{011}$ in the $n$-direction and substituted $u_{110}$ and $u_{101}$. Thus, the dpKdV equation is 3D-consistent.
\end{example}

\subsection{ABS classification of maps on quad-graphs}
In 2003 \cite{ABS-2004} Adler, Bobenko and Suris classified all the 3D-consistent equations in the case where $\mathcal{A}=\mathbb{C}$. In particular, they considered all the equations of the form (\ref{QGeq}), where $u,u_{10},u_{01},u_{11},a,b$$\in\mathbb{C}$, that satisfy the following properties:

\textbf{(I) Multilinearity.} Function $Q=Q(u,u_{10},u_{01},u_{11};a,b)$ is a first order polynomial in each of its arguments, namely linear in each of the fields $u,u_{10},u_{01},u_{11}$. That is,
\begin{equation}
Q(u,u_{10},u_{01},u_{11};a,b)=a_1uu_{10}u_{01}u_{11}+a_2uu_{10}u_{01}+a_3uu_{10}u_{11}+\ldots+a_{16},
\end{equation}
where $a_i=a_i(a,b)$, $i=1,\ldots,16$.

\textbf{(II) Symmetry.} Function $Q$ satisfies the following symmetry property
\begin{equation}
Q(u,u_{10},u_{01},u_{11};a,b)=\epsilon Q(u,u_{01},u_{10},u_{11};b,a)=\sigma Q(u_{10},u,u_{11},u_{01};a,b),
\end{equation}
with $\epsilon,\sigma=\pm1$.

\textbf{(III) Tetrahedron property.} That is, the final value $u_{111}$ is independent of $u$.

ABS proved that all the equations of the form (\ref{QGeq}) which satisfy the above conditions, can be reduced to seven basic equations, using M\"obius (fraction linear) transformations of the independent variables and point transformations of the parameters. These seven equations are distributed into two lists known as the $Q$-\textit{list} (list of 4 equations) and the $H$-\textit{list} (list of 3 equations).

\begin{remark}\normalfont
The discrete potential KdV (dpKdV) equation, which we shall consider in Example \ref{dpKdV}, is the first member of the $H$-list ($H1$ equation).
\end{remark}

Those equations of the form (\ref{QGeq}) which satisfy the multilinearity condition (I), admit Lax representation. In fact, in this case, introducing an auxiliary spectral parameter, $\lambda$, there is an algorithmic way to find a matrix $L$ such that equation (\ref{QGeq}) can be written as the following \textit{zero-curvature} equation
\begin{equation}\label{zerocurv}
L(u_{11},u_{01};a,\lambda)L(u_{01},u;b,\lambda)=L(u_{11},u_{10};b,\lambda)L(u_{10},u;a,\lambda).
\end{equation}

We shall see later on that 1) equations of the form (\ref{QGeq}) with the fields on the edges of the square \ref{3Dconsistency}-(a) are related to Yang-Baxter maps and 2) Yang-Baxter maps may have Lax representation as (\ref{zerocurv}).

\subsubsection{Classification of quadrirational maps: The $F$-list}
A year after the classification of the 3D-consistent equations, ABS in \cite{ABS-2005} classified all the quadrirational maps in the case where $\mathcal{A}=\mathbb{CP}^1$; the associated list of maps is known as the $F$-list. Recall that, a map $Y:(x,y)\mapsto (u(x,y),v(x,y))$ is called \textit{quadrirational}, if the maps
\begin{equation}
u(.,y):\mathcal{A}\rightarrow \mathcal{A},\quad v(x,.):\mathcal{A}\rightarrow \mathcal{A},
\end{equation}
are birational. In particular, we have the following.

\begin{theorem}(ABS, $F$-list) Up to M\"obius transformations, any quadrirational map on $\mathbb{CP}^1\times\mathbb{CP}^1$ is equivalent to one of the following maps
\begin{align}
 u&=ayP, ~~~\quad v=bxP,~~~\qquad P=\frac{(1-b)x+b-a+(a-1)y}{b(1-a)x+(a-b)xy+a(b-1)y}; \tag{$F_I$}\label{FI}\\
 u&=\frac{y}{a}P, \quad\quad v=\frac{x}{b}P,\quad\qquad P=\frac{ax-by+b-a}{x-y}; \tag{$F_{II}$}\label{FII}\\
 u&=\frac{y}{a}P,\quad \quad v=\frac{x}{b}P, \quad\qquad P=\frac{ax-by}{x-y};\tag{$F_{III}$}\label{FIII}\\
 u&=yP ~~\quad\quad v=xP,~\quad\qquad P=1+\frac{b-a}{x-y}\tag{$F_{IV}$};\label{FIV}\\
 u&=y+P,\quad v=x+P,\qquad P=\frac{a-b}{x-y}\tag{$F_{V}$},\label{FV}
\end{align}
up to suitable choice of the parameters $a$ and $b$.
\end{theorem}

We shall come back to the $F$-list in chapter 4, where we shall see that all the equations of the $F$-list have the Yang-Baxter property; yet, the other members of their equivalence classes may not satisfy the Yang-Baxter equation. However, we shall present a more precise list given in \cite{PSTV}.

\subsection{The Yang-Baxter equation}
The original (quantum) Yang-Baxter equation originates in the works of Yang \cite{Yang} and Baxter \cite{Baxter}, and it has a fundamental role in the theory of quantum and classical integrable systems.

Here, we are interested in the study of the set-theoretical solutions of the Yang-Baxter equation. The first examples of such solutions  appeared in 1988, in a paper of Sklyanin \cite{Sklyanin}. However, the study of the set-theoretical solutions was formally proposed by Drinfeld in 1992 \cite{Drinfel'd}, and gained a more algebraic flavor in \cite{Buchstaber}. Veselov, in \cite{Veselov}, proposed the more elegant term ``Yang-Baxter maps'' for this type of solutions and, moreover, he connected them with integrable mappings \cite{Veselov, Veselov3}.

Let $V$ be a vector space and $Y\in \End(V \otimes V)$ a linear operator. The Yang-Baxter equation is given by the following
\begin{equation}\label{YB_eq1}
Y^{12}\circ Y^{13} \circ Y^{23}=Y^{23}\circ Y^{13} \circ Y^{12},
\end{equation}
where $Y^{ij}$, $i,j=1,2,3$, $i\neq j$, denotes the action of $Y$ on the $ij$ factor of the triple tensor product $V \otimes V\otimes V$. In this form, equation (\ref{YB_eq1}) is known in the literature as the \textit{quantum YB equation}.

\subsubsection{Parametric Yang-Baxter maps}
Let us now replace the vector space $V$ by a set $A$, and the tensor product $V\otimes V$ by the Cartesian product $A \times A$. In what follows, we shall consider $A$ to be a finite dimensional algebraic variety in $K^N$, where $K$ is any field of zero characteristic, such as $\mathbb{C}$ or $\mathbb{Q}$.

Now, let $Y\in \End(A\times A)$ be a map defined by
\begin{equation}\label{Y-map}
Y:(x,y)\mapsto (u(x,y),v(x,y)).
\end{equation}
Furthermore, we define the maps $Y^{ij}\in \End(A\times A \times A)$ for $i,j=1,2,3,~i\neq j$, which appear in equation (\ref{YB_eq1}), by the following relations
\begin{subequations}\label{Yijs}
\begin{align}
 Y^{12}(x,y,z)&=(u(x,y),v(x,y),z), \\
 Y^{13}(x,y,z)&=(u(x,z),y,v(x,z)), \\
 Y^{23}(x,y,z)&=(x,u(y,z),v(y,z)).
\end{align}
\end{subequations}
Let also $Y^{21}=\pi Y \pi$, where $\pi\in\End(A\times A)$ is the permutation map: $\pi(x,y)=(y,x)$.

Map $Y$ is a YB map, if it satisfies the YB equation ($\ref{YB_eq1}$). Moreover, it is called \textit{reversible} if the composition of $Y^{21}$ and $Y$ is the identity map, i.e.
\begin{equation}\label{reversible}
Y^{21}\circ Y=Id.
\end{equation}

Now, let us consider the case where parameters are involved in the definition of the YB map. In particular we define the following map
\begin{equation}
Y_{a,b}:(x,y)\mapsto (u,v)\equiv (u(x,y;a,b),v(x,y;a,b)).
\end{equation}
This map is called \textit{parametric YB map} if it satisfies the \textit{parametric YB equation}
\begin{equation}\label{YB_eq}
Y^{12}_{a,b}\circ Y^{13}_{a,c} \circ Y^{23}_{b,c}=Y^{23}_{b,c}\circ Y^{13}_{a,c} \circ Y^{12}_{a,b}.
\end{equation}

One way to represent the map $Y_{a,b}$ is to consider the values $x$ and $y$ taken on the sides of the quadrilateral as in figure \ref{YBmap-Eq}-(a); the map $Y_{a,b}$ maps the values $x$ and $y$ to the values placed on the opposite sides of the quadrilateral, $u$ and $v$.

Moreover, for the YB equation, we consider the values $x$, $y$ and $z$ taken on the sides of the cube as in figure \ref{YBmap-Eq}-(b). Specifically, by the definition \ref{Yijs} of the functions $Y^{ij}$, the map $Y^{23}_{b,c}$ maps
\begin{equation}
(x,y,z)\stackrel{Y^{23}_{b,c}}{\rightarrow}(x,y^{(1)},z^{(1)}),
\end{equation}
using the right face of the cube. Then, map $Y^{13}_{a,c}$ maps
\begin{equation}
(x,y^{(1)},z^{(1)})\stackrel{Y^{13}_{a,c}}{\rightarrow}(x^{(1)},y^{(1)},z^{(2)})\equiv Y^{13}_{a,c} \circ Y^{23}_{b,c}(x,y,z),
\end{equation}
using the front face of the cube. Finally, map $Y^{12}_{a,b}$ maps
\begin{equation}
(x^{(1)},y^{(1)},z^{(2)})\stackrel{Y^{12}_{a,b}}{\rightarrow}(x^{(2)},y^{(2)},z^{(2)})\equiv Y^{12}_{a,b}\circ Y^{13}_{a,c} \circ Y^{23}_{b,c}(x,y,z),
\end{equation}
using the top face of the cube.

On the other hand, using the bottom, the back and the left face of the cube, the values $x$, $y$ and $z$ are mapped to the values $\hat{x}^{(2)}$, $\hat{y}^{(2)}$ and $\hat{z}^{(2)}$ via the map $Y^{23}_{b,c}\circ Y^{13}_{a,c} \circ Y^{12}_{a,b}$ which consists with the right hand side of equation, namely (\ref{Y-map})
\begin{equation}
Y^{23}_{b,c}\circ Y^{13}_{a,c} \circ Y^{12}_{a,b}(x,y,z)=(\hat{x}^{(2)},\hat{y}^{(2)},\hat{z}^{(2)}).
\end{equation}
Therefore, the map $Y_{a,b}$ satisfies the YB equation (\ref{YB_eq}) if and only if $x^{(2)}=\hat{x}^{(2)}$, $y^{(2)}=\hat{y}^{(2)}$ and $z^{(2)}=\hat{z}^{(2)}$.

\begin{figure}[ht]
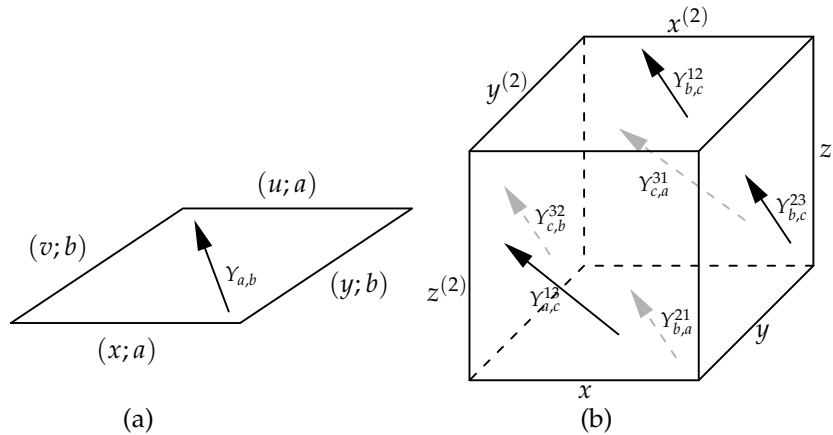

\centering
\centertexdraw{
\setunitscale 0.6
\move (-2 0.5)  \lvec(0 0.5) \lvec(1.5 1.5) \lvec(-0.5 1.5) \lvec(-2 0.5)
\textref h:C v:C \small{\htext(-1 0.25){$(x;a)$}}
\textref h:C v:C \htext(0.4 1.7){$(u;a)$}
\textref h:C v:C \htext(1.03 0.82){$(y;b)$}
\textref h:C v:C \htext(-1.6 1.15){$(v;b)$}
\textref h:C v:C \scriptsize{\htext(0 0.9){$Y_{a,b}$}}
\move (-0.1 0.6) \arrowheadtype t:F \avec(-0.4 1.4)
\move (2.7 1.1) \arrowheadtype t:F \lpatt(0.067 0.1)
\setgray 0.7
\avec(2.3 1.7)
\textref h:C v:C \scriptsize{\htext(2.7 1.4){$Y^{32}_{c,b}$}}
\move (4.4 1.4) \arrowheadtype t:F \footnotesize{\avec(3.3 2.2)}
\textref h:C v:C \scriptsize{\htext(3.6 1.7){$Y^{31}_{c,a}$}}
\move (3.8 0.2) \arrowheadtype t:F \avec(3.4 0.8)
\textref h:C v:C \scriptsize{\htext(3.8 0.5){$Y^{21}_{b,a}$}}
%labels & vectors for the 1st cube
\lpatt()
\setgray 0
\move (2 0)  \lvec(4 0) \lvec(5 1)
\lpatt(0.067 0.1) \lvec(3 1) \lvec(2 0)
\lpatt() \lvec(2 2) \lvec(3 3)
\lpatt (0.067 0.1) \lvec(3 1)
\lpatt() \move (3 3) \lvec(5 3) \lvec(4 2) \lvec(4 0)
\move (2 2) \lvec(4 2)
\move (5 3) \lvec(5 1)
\textref h:C v:C \small{\htext(3 -0.1){$x$}}
\textref h:C v:C \small{\htext(4.55 0.4){$y$}}
\textref h:C v:C \small{\htext(5.1 2){$z$}}
%labels & vectors for the 2nd cube
\move (4.8 1.2) \arrowheadtype t:F \avec(4.4 1.8)
\textref h:C v:C \scriptsize{\htext(4.8 1.5){$Y^{23}_{b,c}$}}
\move (3.3 0.4) \arrowheadtype t:F \avec(2.3 1.2)
\textref h:C v:C \scriptsize{\htext(2.65 0.7){$Y^{13}_{a,c}$}}
\textref h:C v:C \small{\htext(1.8 0.8){$z^{(2)}$}}
\move (3.9 2.3) \arrowheadtype t:F \avec(3.5 2.9)
\textref h:C v:C \scriptsize{\htext(3.9 2.6){$Y^{12}_{b,c}$}}
\textref h:C v:C \small{\htext(3.9 3.15){$x^{(2)}$}} \htext(2.33 2.55){$y^{(2)}$}
%Shape numbering
\textref h:C v:C \htext(3.1 -0.35){(b)}
\textref h:C v:C \htext(-0.9 -0.35){(a)}
}
\caption{Cubic representation of (a) the parametric YB map and (b) the corresponding YB equation.}\label{YBmap-Eq}
\end{figure}

Most of the examples of YB maps which appear in these lecture notes are parametric.

\begin{example}\normalfont
One of the most famous parametric YB maps is Adler's map \cite{Adler}
\begin{equation}\label{Adler_map}
(x,y)\stackrel{Y_{a,b}}{\rightarrow}(u,v)=\left(y-\frac{a-b}{x+y},x+\frac{a-b}{x+y}\right),
\end{equation}
which is related to the 3-D consistent discrete potential KDV equation \cite{Frank, PNC}.
\end{example}

\begin{exercise}
By straightforward substitution, show that Adler's map (\ref{Adler_map}) satisfies the YB equation (\ref{YB_eq1}).
\end{exercise}

\subsubsection{Matrix refactorisation problems and the Lax equation}
Let us consider the matrix $L$ depending on a variable $x$, a parameter $c$ and a \textit{spectral parameter} $\lambda$, namely $L=L(x;c,\lambda)$, such that the following matrix refactorisation problem
\begin{equation} \label{eqLax}
L(u;a,\lambda)L(v;b,\lambda)=L(y;b,\lambda)L(x;a,\lambda), \quad \text{for any $\lambda \in \mathbb{C}$,}
\end{equation}
is satisfied whenever $(u,v)=Y_{a,b}(x,y)$. Then, $L$ is called Lax matrix for $Y_{a,b}$, and (\ref{eqLax}) is called the \textit{Lax-equation} or \textit{Lax-representation} for $Y_{a,b}$.

\begin{note}\normalfont
In the rest of this thesis we use the letter ``$L$" when referring to Lax matrices of the refactorisation problem (\ref{eqLax}) and the bold ``$\mathbf{L}$" for Lax operators. Moreover, for simplicity of the notation, we usually omit the dependence on the spectral parameter, namely $L(x;a,\lambda)\equiv L(x;a)$.
\end{note}

Since the Lax equation (\ref{eqLax}) does not always have a unique solution for $(u,v)$, Kouloukas and Papageorgiou in \cite{kouloukas2} proposed the term \textit{strong Lax matrix} for a YB map. This is when the Lax equation is equivalent to a map
\begin{equation}\label{unique-sol}
  (u,v)=Y_{a,b}(x,y).
\end{equation}
The uniqueness of refactorisation (\ref{eqLax}) is a sufficient condition for the solutions of the Lax equation to define a reversible YB map of the form (\ref{unique-sol}). In particular, we have the following.

\begin{proposition}\label{PropVes}
(Veselov) Let $u=u(x,y)$, $v=v(x,y)$ and $L=L(x;\alpha)$ a matrix such that the refactorisation (\ref{eqLax}) is unique. Then, the map defined by (\ref{unique-sol}) satisfies the Yang-Baxter equation and it is reversible.
\end{proposition}

In the case where the map (\ref{unique-sol}) admits Lax representation (\ref{eqLax}), but it is not equivalent to (\ref{eqLax}), one may need to check the YB property separately.

In these lecture notes, we are interested in those YB maps whose Lax representation involves matrices with rational dependence on the spectral parameter, as the following.

\begin{example}\normalfont
In terms of Lax matrices, Adler's map (\ref{Adler_map}) has the following strong Lax representation \cite{Veselov2, Veselov3}
\begin{equation}
L(u;a,\lambda)L(v;b,\lambda)=L(y;b,\lambda)L(x;a,\lambda), \quad \text{for any $\lambda \in \mathbb{C}$,}
\end{equation}
where
\begin{equation}
L(x;a,\lambda)=
\left(\begin{matrix}
x & 1 \\
x^2-a & x
\end{matrix}\right)-\lambda \left(\begin{matrix}
0 & 0 \\
1 & 0
\end{matrix}\right).
\end{equation}
\end{example}

\subsection{Yang-Baxter maps and 3D consistent equations}
From the representation of the YB equation on the cube, as in Fig. \ref{YBmap-Eq}-(b), it is clear that the YB equation is essentially the same with the 3D consistency condition with the fields lying on the edges of the cube. Therefore, one would expect that we can derive YB maps from equations having the 3D consistency property.

The connection between YB maps and the multidimensional consistency condition for equations on quad graphs originates in the paper of Adler, Bobenko and Suris in 2003 \cite{ABS-2004}. However, a more systematic approach was presented in the paper of Papageorgiou, Tongas and Veselov \cite{PTV} a couple of years later and it is based on the symmetry analysis of equations on quad-graphs. In particular, the YB variables constitute invariants of their symmetry groups.

We present the example of the discrete potential KdV (dpKdV) equation \cite{PNC, Frank} which was considered in \cite{PTV}.

\begin{example}\normalfont
The dpKdV equation is given by
\begin{equation}\label{dpKdV}
(f_{11}-f)(f_{10}-f_{01})-a+b=0,
\end{equation}
where the fields are placed on the vertices of the square as in figure (\ref{dpKdVtoAdler}). We consider the values on the edges to be the difference of the values on the vertices, namely
\begin{equation}\label{invar}
x=f_{10}-f, \quad y=f_{11}-f_{10}, \quad u=f_{11}-f_{01} \quad \text{and} \quad v=f_{01}-f,
\end{equation}
as in figure (\ref{dpKdVtoAdler}). This choice of the variables is motivated by the fact that the dpKdV equation is invariant under the translation $f\rightarrow f+const.$ Now, the invariants (\ref{invar}) satisfy the following equation
\begin{equation}\label{eq1}
x+y=u+v.
\end{equation}
Moreover, equation (\ref{dpKdV}) can be rewritten as
\begin{equation}\label{eq2}
(x+y)(x-v)=a-b.
\end{equation}

Solving (\ref{eq1}) and (\ref{eq2}), we obtain the Adler's map (\ref{Adler_map}).
\end{example}

\begin{figure}[ht]
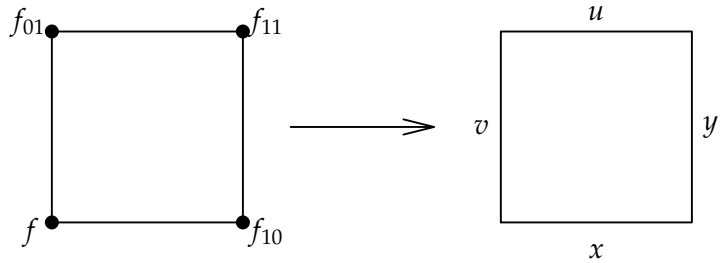

\centering
\centertexdraw{\setunitscale 0.5 \move (-1 0) \fcir f:0 r:0.075 \lvec(1 0)\fcir f:0 r:0.075 \lvec(1 2) \fcir f:0 r:0.075 \lvec (-1 2) \fcir f:0 r:0.075 \lvec (-1 0)
\textref h:C v:C \htext(-1.22 -0.1){$f$} \htext(-1.25 2.1){$f_{01}$} \htext(1.25 -0.1){$f_{10}$} \htext(1.25 2.1){$f_{11}$}
\move (1.5 1) \arrowheadtype t:V \avec(3 1)
%%%%%%%%%%%
\move (3.7 0) \lvec(5.7 0) \lvec(5.7 2) \lvec (3.7 2) \lvec (3.7 0)
\textref h:C v:C \htext(4.7 -0.3){$x$} \htext(5.9 1){$y$} \htext(4.7 2.2){$u$} \htext(3.5 1){$v$}
}
\caption{(a) dpKdV equation: fields placed on vertices (b) Adler's map: fields placed on the edges.}\label{dpKdVtoAdler}
\end{figure}

\subsection{Classification of quadrirational YB maps: The $H$-list}
All the quadrirational maps in the $F$-list presented in the first chapter satisfy the YB equation. However, in principle, their M\"obius-equivalent maps do not necessarily have the YB property, as in the following.

\begin{example}\normalfont
Consider the map $F_V$ of the $F$-list. Under the change of variables
\begin{equation}\label{Mchange}
(x,y,u,v)\rightarrow(-x,-y,u,v),
\end{equation}
it becomes
\begin{equation}
(x,y)\rightarrow(-y-\frac{a-b}{x-y},-x-\frac{a-b}{x-y}).
\end{equation}
The above map does not satisfy the YB equation.
\end{example}

In fact, all the maps of the $F$-list lose the YB property under the transformation (\ref{Mchange}).

The quadrirational maps which satisfy the YB equation were classified in \cite{PSTV}. Particularly, their classification is based on the following.

\begin{definition}
Let $\rho_\lambda:X\rightarrow X$ be a $\lambda$-parametric family of bijections. The parametric YB maps $Y_{a,b}$ and $\tilde{Y}_{a,b}$ are called equivalent, if they are related as follows
\begin{equation}\label{equivYB}
\tilde{Y}_{a,b}=\rho_{a}^{-1}\times \rho_{b}^{-1}~Y_{a,b}~\rho_a \times\rho_b.
\end{equation}
\end{definition}

\begin{remark}\normalfont
It is straightforward to show that the above equivalence relation is well defined; if $Y_{a,b}$ has the YB property, so does the map $\tilde{Y}_{a,b}$.
\end{remark}

The representative elements of the equivalence classes, with respect to the equivalence relation (\ref{equivYB}), are given by the following list.

\begin{theorem} Every quadrirational parametric YB map is equivalent (in the sense (\ref{equivYB})) to one of the maps of the F-list or one of the maps of the following list
\begin{align}
 u&=yQ^{-1}, ~~\quad v=xQ,~\quad\qquad Q=\frac{(1-b)xy+(b-a)y+b(a-1)}{(1-a)xy+(a-b)x+a(b-1)}; \tag{$H_I$}\label{HI}\\
 u&=yQ^{-1}, ~~\quad v=xQ,~\quad\qquad Q=\frac{a+(b-a)y-bxy}{b+(a-b)x-axy}; \tag{$H_{II}$}\label{HII}\\
 u&=\frac{y}{a}Q,~\quad \quad v=\frac{x}{b}Q, \quad\qquad Q=\frac{ax+by}{x+y};\tag{$H_{III}$}\label{HIII}\\
 u&=yQ^{-1} \quad\quad v=xQ,~\quad\qquad Q=\frac{axy+1}{bxy+1}\tag{$H_{IV}$};\label{HIV}\\
 u&=y-P, ~~\quad v=x+P,\qquad P=\frac{a-b}{x+y}\tag{$H_{V}$}.\label{HV}
\end{align}
\end{theorem}

We refer to the above list as the $H$-list. Note that, the map $H_V$ is the Adler's map (\ref{Adler_map}).

\subsection{Derivation of Yang-Baxter maps from Darboux transformations}
In this section we shall show how we can use Darboux transformations of particular Lax operators to construct Yang-Baxter maps, which can then be restricted to completely integrable ones (in the Liouville sense) on invariant leaves.

\subsubsection{Invariants and integrability of Yang-Baxter maps}

\begin{proposition}\label{genInv}
If $L=L(x,a;\lambda)$ is a Lax matrix with corresponding YB map, $Y:(x,y)\mapsto (u,v)$, then the $\tr(L(y,b;\lambda)L(x,a;\lambda))$ is a generating function of invariants of the YB map.
\end{proposition}
\begin{proof}
Since,
\begin{equation} \label{trace}
\tr(L(u,a;\lambda)L(v,b;\lambda))\overset{(\ref{eqLax})}{=}\tr(L(y,b;\lambda)L(x,a;\lambda))=\tr(L(x,a;\lambda)L(y,b;\lambda)),
\end{equation}
and the function $\tr(L(x,a;\lambda)L(y,b;\lambda))$ can be written as $\displaystyle\tr(L(x,a;\lambda)L(y,b;\lambda))=\sum_k \lambda^k I_k(x,y;a,b)$, from (\ref{trace}) follows that
\begin{equation}
I_i(u,v;a,b)=I_i(x,y;a,b),
\end{equation}
which are invariants for $Y$.
\end{proof}

The invariants of a YB map are essential towards its integrability in the Liouville sense. Note that, the generated invariants, $I_i(x,y;a,b)$, in proposition \ref{genInv}, may not be functionally independent. In what follows, we define the complete (Liouville) integrability of a YB map, following \cite{Fordy, Veselov4}. However, for Liouville integrability, the reader is expected to have some basic knowledge of Poisson geometry \cite{Arnold, Marsden-Ratiu}.

\newtheorem{CompleteIntegrability}{Definition}[section]
\begin{CompleteIntegrability}
A $2N$-dimensional Yang-Baxter map,
\begin{equation}
Y:(x_1,\ldots,x_{2N})\mapsto (u_1,\ldots,u_{2N}), \quad u_i=u_i(x_1,\ldots,x_{2N}), \quad i=1,\ldots,2N, \nonumber
\end{equation}
is said to be completely integrable or Liouville integrable if
\begin{enumerate}
%\begin{itemize}
	\item there is a Poisson matrix,\index{Poisson!matrix} $J_{ij}=\left\{x_i,x_j\right\}$, of rank $2r$, which is invariant under the action of the YB map, namely $J_{ij}$ and $\tilde{J_{ij}}=\left\{u_i,u_j\right\}$ have the same functional form of their respective arguments,
	\item map $Y$ has $r$ functionally independent invariants, $I_i$, namely $I_i\circ Y=I_i$, which are in involution with respect to the corresponding Poisson bracket, i.e. $\left\{I_i,I_j\right\}=0$, $i,j=1,\ldots,r$, $i\neq j$,
	\item there are $k=2N-2r$ Casimir functions, namely functions $C_i$, $i=1,\ldots,k$, such that $\left\{C_i,f\right\}=0$, for any arbitrary function $f=f(x_1,...,x_{2N})$. These are invariant under $Y$, namely $C_i\circ Y=C_i$.
%\end{itemize}
\end{enumerate}
\end{CompleteIntegrability}

\subsubsection{Example: From the NLS equation to the Adler-Yamilov YB map}
Recall that, in the case of NLS equation, the Lax operator is given by
\begin{equation}
\mathbf{L}(p,q;\lambda)=D_x+\lambda U_{1}+U_{0},\quad \text{where} \quad U_1=\sigma_3,\quad U_0=\left(\begin{matrix}
        0 & 2p \\
        2q & 0
    \end{matrix}\right),
\end{equation}
where $\sigma_3$ is the standard Pauli matrix, i.e. $\sigma_3=\text{diag}(1,-1)$.

Moreover, a Darboux matrix for $\mathbf{L}$ is given by
\begin{equation}\label{NLSDarboux}
  M=\lambda \left(
     \begin{matrix}
         1 & 0\\
         0 & 0
     \end{matrix}\right)+\left(
     \begin{matrix}
         f & p\\
         q_{10} & 1
     \end{matrix}\right).
\end{equation}
The entries of (\ref{NLSDarboux}) must satisfy the following system of equations
\begin{equation}\label{baecklundNLS}
\partial_x f=2(pq-p_{10}q_{10}), \qquad \partial_x p=2(pf-p_{10}), \qquad \partial_x q_{10}=2(q-q_{10}f),
\end{equation}
which admits the following first integral
\begin{equation}\label{integralNLS}
\partial_x(f-pq_{10})=0.
\end{equation}
This integral implies that $\partial_x \det M=0$.

In correspondence with (\ref{NLSDarboux}), we define the matrix
\begin{equation} \label{3d-Darboux-NLS}
  M(\textbf{x};\lambda)=\lambda \left(
     \begin{matrix}
         1 & 0\\
         0 & 0
     \end{matrix}\right)+\left(
     \begin{matrix}
         X & x_1\\
         x_2 & 1
     \end{matrix}\right),
\qquad \textbf{x}=(x_1,x_2,X),
\end{equation}
and substitute it into the Lax equation (\ref{eqLax})
\begin{equation}\label{laxM}
M(\textbf{u};\lambda)M(\textbf{v};\lambda)=M(\textbf{y};\lambda)M(\textbf{x};\lambda),
\end{equation}
to derive the following system of equations
\begin{eqnarray}
&v_1 = x_1,\ u_2 = y_2,\ U +V = X + Y ,\ u_2 v_1 = x_1 y_2,& \nonumber \\
&u_1 +U v_1 = y_1+x_1 Y,\ u_1 v_2+U V = x_2 y_1+X Y,\ v_2+u_2 V = x_2 + X y_2.& \nonumber
\end{eqnarray}

The corresponding algebraic variety is a union of two six-dimensional components. The first one is obvious from the refactorisation problem (\ref{laxM}), and it corresponds to the permutation map
\begin{equation}
 \textbf{x}\mapsto \textbf{u}=\textbf{y}, \quad \textbf{y}\mapsto \textbf{v}=\textbf{x},  \nonumber
\end{equation}
which is a (trivial) YB map. The second one can be represented as a rational six-dimensional non-involutive map of $K^3\times K^3 \rightarrow K^3\times K^3$
\small
\begin{subequations}\label{NLS-3d}
 \begin{align}
&x_1\mapsto u_1=\frac{y_1+x_1^2x_2-x_1X+x_1Y}{1+x_1y_2},\qquad~~ y_1\mapsto v_1=x_1, \\
&x_2\mapsto u_2=y_2,\quad\qquad\qquad\qquad\qquad \quad~~~ y_2\mapsto v_2=\frac{x_2+y_1y_2^2+y_2X-y_2Y}{1+x_1y_2},  \\
&X\mapsto U=\frac{y_1y_2-x_1x_2+X+x_1y_2Y}{1+x_1y_2}, \qquad~ Y\mapsto V=\frac{x_1x_2-y_1y_2+x_1y_2X+Y}{1+x_1y_2},
 \end{align}
\end{subequations}
\normalsize
which, one can easily check that, satisfies the YB equation.

The trace of $M(\textbf{y};\lambda)M(\textbf{x};\lambda)$ is a polynomial in $\lambda$ whose coefficients are
\begin{equation}
\mbox{tr}(M(\textbf{y};\lambda)M(\textbf{x};\lambda))=\lambda^2+\lambda I_1(\textbf{x},\textbf{y})+I_2(\textbf{x},\textbf{y}), \nonumber
\end{equation}
where
\begin{equation}\label{NLS3dInv}
I_1(\textbf{x},\textbf{y})=X+Y \qquad \text{and} \qquad I_2(\textbf{x},\textbf{y})=x_2y_1+x_1y_2+XY,
\end{equation}
and those, according to proposition \ref{genInv}, are invariants for the YB map (\ref{NLS-3d}).

In the following section we show that the YB map (\ref{NLS-3d}) can be restricted to a four-dimensional YB map which has Poisson structure.

%%%%%%%%%%%%%%%%%%%%%%%%%%%%%%%%%%%%%%%%%%%%%%%%%%%%%%%%%%%%%%%%%%%%%%%%%%%%%%%%%%%%

Now, we show that map (\ref{NLS-3d}) can be restricted to the Adler-Yamilov map on symplectic leaves, by taking into account the first integral, (\ref{integralNLS}), of the system (\ref{baecklundNLS}).

In particular, we have the following.

\begin{proposition} For the six-dimensional map (\ref{NLS-3d}) we have the following:
\begin{enumerate}
	\item The quantities $\Phi =X-x_1x_2$ and $\Psi=Y-y_1y_2$ are its invariants (first integrals).
	\item It can be restricted to a four-dimensional map $Y_{a,b}:A_a\times A_b \longrightarrow A_a\times A_b$, where $A_a$, $A_b$ are level sets of the first integrals $\Phi$ and $\Psi$, namely
\begin{subequations}\label{symleaves}
\begin{align}
A_a&=\{(x_1,x_2,X)\in K^3; X=a+x_1x_2\}, \\
A_b&=\{(y_1,y_2,Y)\in K^3; Y=b+y_1y_2\}.
\end{align}
\end{subequations}
\end{enumerate}

Moreover, map $Y_{a,b}$ is the Adler-Yamilov map.
\end{proposition}

\begin{proof}
\begin{enumerate}
	\item It can be readily verified that (\ref{NLS-3d}) implies $U-u_1u_2=X-x_1x_2$ and $V-v_1v_2=Y-y_1y_2$. Thus, $\Phi$ and $\Psi$ are invariants, i.e.  first integrals of the map.
	\item The existence of the restriction is obvious. Using the conditions $X=x_1x_2+a$ and $Y=y_1y_2+b$, one can eliminate $X$ and $Y$ from (\ref{NLS-3d}). The resulting map, $\textbf{x}\rightarrow \textbf{u}(\textbf{x},\textbf{y})$, $\textbf{y}\rightarrow \textbf{v}(\textbf{x},\textbf{y})$, is given by
\begin{eqnarray} \label{YB_NLS}
(\textbf{x},\textbf{y})\overset{Y_{a,b}}{\longrightarrow }\left(y_1-\frac{a -b}{1+x_1y_2}x_1,y_2,x_1,x_2+\frac{a -b}{1+x_1y_2}y_2\right).
\end{eqnarray}
\end{enumerate}
Map (\ref{YB_NLS}) coincides with the Adler-Yamilov map.
\end{proof}

Map (\ref{YB_NLS}) originally appeared in the work of Adler and Yamilov \cite{Adler-Yamilov}. Moreover, it appears as a YB map in \cite{kouloukas, PT}.

Now, one can use the condition $X=x_1x_2+a$ to eliminate $X$ from the Lax matrix\index{Lax matrix(-ces)} (\ref{3d-Darboux-NLS}), i.e.
\begin{equation} \label{laxNLS}
M(\textbf{x};a,\lambda)=\lambda \left(
\begin{matrix}
 1 & 0\\
 0 & 0
\end{matrix}\right)+\left(
\begin{matrix}
 a+x_1x_2 & x_1\\
 x_2 & 1
\end{matrix}\right), \quad \textbf{x}=(x_1,x_2).
\end{equation}
The form of Lax matrix (\ref{laxNLS}) coincides with the well known Darboux transformation for the NLS equation (see \cite{Rog-Schief} and references therein). Now, Adler-Yamilov map follows from the strong Lax representation
\begin{equation} \label{lax_eq_NLS}
  M(\textbf{u};a,\lambda)M(\textbf{v};b,\lambda)=M(\textbf{y};b,\lambda)M(\textbf{x};a,\lambda).
\end{equation}
Therefore, the Adler-Yamilov map \eqref{YB_NLS} is a reversible parametric YB map with strong Lax matrix ($\ref{laxNLS}$). Moreover, it is easy to verify that it is not involutive.

For the integrability of this map we have the following

\begin{proposition}\label{cintA-Y}
The Adler-Yamilov map \eqref{YB_NLS} is completely integrable.
\end{proposition}

\begin{proof}
The $4\times 4$ Poisson matrix associated to the following Poisson bracket
\begin{equation}\label{Pbracket}
\left\{x_1,x_2\right\}=\left\{y_1,y_2\right\}=1,\qquad \text{and all the rest}
\qquad \left\{x_i,y_j\right\}=0,
\end{equation}
is invariant under the YB map \eqref{YB_NLS}, namely the latter is a Poisson map with respect to the Poisson bracket $\pi=\frac{\partial}{\partial x_1}\wedge \frac{\partial}{\partial x_2}+\frac{\partial}{\partial y_1}\wedge \frac{\partial}{\partial y_2}$.
Now, from the trace of $M(\textbf{y};b,\lambda)M(\textbf{x};a,\lambda)$ we obtain the following invariants for the map \eqref{YB_NLS}
\begin{eqnarray}
&&I_1(\textbf{x},\textbf{y})=x_1 x_2+y_1 y_2+a+b, \\
&&I_2(\textbf{x},\textbf{y})=(a+x_1 x_2)(b+y_1 y_2)+x_1 y_2+x_2 y_1+1.
\end{eqnarray}
It is easy to check that $I_1,I_2$ are in involution with respect to \eqref{Pbracket}, namely $\left\{I_1,I_2\right\}=0$. Therefore the map \eqref{YB_NLS} is completely integrable.
\end{proof}

The above proposition implies the following.

\begin{corollary}
The invariant leaves $A_a$ and $B_b$, given in (\ref{symleaves}), are symplectic.
\end{corollary}

\subsection{Yang-Baxter maps for NLS type equations. Noncommutative extensions}
In the previous section we showed how one can use Darboux transformations to construct Yang-Baxter maps which can restrict to completely integrable ones on invariant leaves. In particular, using a Darboux transformation for the NLS equation, namely matrix \eqref{NLSDarboux}, we constructed the six-dimensional Yang-Baxter map \eqref{NLS-3d} which was restricted to the completely integrable Adler-Yamilov map \eqref{YB_NLS} on the invariant (symplectic) leaves \eqref{symleaves}. In \cite{sokor, Sokor-Sasha}, all Yang-Baxter maps related to the cases \eqref{Lax-NLS}, \eqref{Lax-DNLS} and \eqref{Lax-dDNLS}, as well as their integrability, were studied, using the associated Darboux transformations.

Motivated by some results on noncommutative extensions (in a Grassmann setting) of Darboux transformations and their use in the construction of nonocommutative discrete integrable systems \cite{Georgi}, in \cite{GKM, Sokor-Sasha-2}, the first steps towards extending the theory of Yang-Baxter maps were made.

\section*{Acknowledgement}
We would like to thank the organizers of the summer school Abecedarian of SIDE 12 for the opportunity to participate as lecturers, as well as for the financial support. 
%%%%%%%%%% References %%%%%%%%%%%
\bibliographystyle{plain}
\bibliography{refs}

\end{document}